\documentclass[twocolumn]{svjour3}           %
\usepackage{amsmath}
\usepackage{amsfonts}
\usepackage{amssymb}
\usepackage{ifpdf}
\usepackage{subfig}
\usepackage{wrapfig}
\usepackage{cite}
\usepackage{clrscode}
\usepackage{comment}
\usepackage{appendix}

\newif\ifarxiv
\newif\ifnaco

\ifnaco
	\arxivfalse
\else
	\arxivtrue
\fi

\ifpdf

  \usepackage[pdftex]{epsfig}
  \usepackage[pdftex]{hyperref}

\else

    \usepackage[dvips]{epsfig}
    \newcommand{\href}[2]{#2}

\fi

\newcommand{\Z}{\mathbb{Z}}
\newcommand{\N}{\mathbb{N}}

\newcommand{\dom}{{\rm dom} \;}

\newcommand{\termasm}[1]{\mathcal{A}_{\Box}[\mathcal{#1}]}
\newcommand{\prodasm}[1]{\mathcal{A}[\mathcal{#1}]}

\newcommand{\calT}{\mathcal{T}}

\begin{document}

\title{Hierarchical Self-Assembly of Fractals with Signal-Passing Tiles}

\author{
 Jacob Hendricks%
    \thanks{Department of Computer Science and Information Systems, University of Wisconsin - River Falls, River Falls, WI, USA
    \protect\url{jacob.hendricks@uwrf.edu}}
\and
 Meagan Olsen
    \thanks{Fayetteville High School, Fayetteville, AR, USA
    \protect\url{olsen.megs@gmail.com}}
\and
 Matthew J. Patitz
    \thanks{Department of Computer Science and Computer Engineering, University of Arkansas, Fayetteville, AR, USA
    \protect\url{patitz@uark.edu} This author's research was supported in part by National Science Foundation Grant CCF-1422152.}
\and
 Trent A. Rogers
    \thanks{Department of Computer Science and Computer Engineering, University of Arkansas, Fayetteville, AR, USA
    \protect\url{tar003@uark.edu}.  This author's research was supported by the National Science Foundation Graduate Research Fellowship Program under Grant No. DGE-1450079, and National Science Foundation Grant CCF-1422152.}
\and
 Hadley Thomas
    \thanks{Fayetteville High School, Fayetteville, AR, USA
    \protect\url{hadleythomas88@gmail.com}}
}

\institute{}

\date{}

\maketitle

\begin{abstract}
In this paper, we present high-level overviews of tile-based self-assembling systems capable of producing complex, infinite, aperiodic structures known as discrete self-similar fractals.  Fractals have a variety of interesting mathematical and structural properties, and by utilizing the bottom-up growth paradigm of self-assembly to create them we not only learn important techniques for building such complex structures, we also gain insight into how similar structural complexity arises in natural self-assembling systems.  Our results fundamentally leverage hierarchical assembly processes, and use as our building blocks square ``tile'' components which are capable of activating and deactivating their binding ``glues'' a constant number of times each, based only on local interactions.  We provide the first constructions capable of building arbitrary discrete self-similar fractals at scale factor 1, and many at temperature 1 (i.e. ``non-cooperatively''), including the Sierpinski triangle.
\end{abstract}

\section{Introduction}

Fractal patterns have mathematically interesting characteristics, such as recursive self-similarity, and structural properties which lend naturally occurring fractal structures, such as branch patterns and circulatory systems, impressive abilities to efficiently maximize coverage, dissipate heat, etc. Such fractal patterns in nature tend to arise via local processes following relatively simple sets of rules, as forms of self-assembly.  Because of this, and the complex aperiodic nature of fractals, they are a natural target of study during the development of artificial self-assembling systems.  As one of the first mathematical abstractions of self-assembling systems, Winfree's abstract Tile Assembly Model (aTAM) \cite{Winf98} has been the platform for several results showing the impossibility of self-assembling discrete self-similar fractals such as the Sierpinski triangle\footnote{In this paper we refer only to ``strict'' self-assembly, wherein a shape is made by placing tiles only within the domain of the shape, as opposed to ``weak'' self-assembly where a pattern representing the shape can be formed embedded within a framework of additional tiles.}\cite{jSSADST} and similar fractals \cite{TreeFractals}, and also for designing systems which can approximate them \cite{jSSADST,LutzShutters12,jSADSSF}. In a more generalized model called the 2-Handed Assembly Model \cite{AGKS05g} (2HAM, a.k.a. Hierarchical Assembly Model) which allows pairs of large assemblies to bind together, rather than being restricted to only single tile additions per step like the aTAM, the impossibility of self-assembling the Sierpinski triangle \cite{Versus} has also been shown.  In further generalizations allowing larger numbers of assemblies to combine in single steps \cite{MHAM} shapes were shown to self-assemble as well one unscaled fractal, the Sierpinski carpet.

A more recently developed model of tile-based self-assembly, called the Signal-passing Tile Assembly Model (STAM) was developed in \cite{jSignals} to model the behavior of DNA-based tiles capable of strand displacement reactions initiated during the binding of their glues which can then either activate or deactivate other glues on the same tile.  Such signal-passing tiles have been experimentally demonstrated \cite{SignalTilesExperimental}, and various theoretical results have demonstrated the power of systems using such tiles to efficiently simulate Turing machines \cite{jSignals}, replicate patterns \cite{STAMPatternRep} and shapes \cite{STAMshapes}, and also to self-assemble the Sierpinski triangle at scale factor 2.  Additionally, a similar theoretical model, the Active Tile Assembly Model, has been shown to be capable of universal computation in non-cooperative ``temperature-1'' systems \cite{JonoskaSignals1} and also of self-assembling an infinite, self-similar substitution tiling pattern \cite{JonoskaSignals2} which fills the plane (rather than having arbitrarily large holes, and fractal dimension less than $2$, like self-similar fractals).

In this paper, we provide constructions in the STAM which include: (1) the first capable of self-assembling the Sierpinski triangle at scale factor 1, which in fact even works in non-cooperative assembly (i.e. at temperature 1), and (2) an algorithmic method which uses the definition of a fractal as input in order to develop an STAM system which self-assembles that fractal at scale factor 1.  The second result develops systems at temperature 1 for an infinite class of fractals, and for the full class of discrete self-similar fractals at temperature 2.  Our results fundamentally leverage techniques of hierarchical self-assembly and utilize specifically designed instances of geometric hindrance to allow fractals to grow in a carefully controlled, stage-by-stage manner.  In the following sections, we first give an overview of the models and terminology used in the paper, then provide both main constructions.
\ifnaco
\else
Please note that an extended abstract version of this paper was published in \cite{STAM-fractals}, and this version contains the full set of technical details.
\fi

\def\latent/{{\texttt{latent}}}
\def\on/{{\texttt{on}}}
\def\off/{{\texttt{off}}}
\def\calF{{\mathcal{F}}}

\section{Preliminaries}\label{sec:prelims}

Here we provide informal descriptions of the models and terms used in this paper.  Due to space limitations, the formal definitions can be found in \cite{Signals3DArxiv}.

\section{Informal definition of the 2HAM}\label{sec:2ham-informal}

The 2HAM \cite{AGKS05g,DDFIRSS07} is a generalization of the abstract Tile Assembly Model (aTAM) \cite{Winf98} in that it allows for two assemblies, both possibly consisting of more than one tile, to attach to each other. We now give a brief, informal, sketch of the 2HAM.

A \emph{tile type} is a unit square with each side having a \emph{glue} consisting of a \emph{label} (a finite string) and \emph{strength} (a non-negative integer).   We assume a finite set $T$ of tile types, but an infinite number of copies of each tile type, each copy referred to as a \emph{tile}.
A \emph{supertile} is (the set of all translations of) a positioning of tiles on the integer lattice $\Z^2$.  Two adjacent tiles in a supertile \emph{interact} if the glues on their abutting sides are equal and have positive strength.
Each supertile induces a \emph{binding graph}, a grid graph whose vertices are tiles, with an edge between two tiles if they interact.
The supertile is \emph{$\tau$-stable} if every cut of its binding graph has strength at least $\tau$, where the weight of an edge is the strength of the glue it represents.
That is, the supertile is stable if at least energy $\tau$ is required to separate the supertile into two parts.  Note that throughout this paper, we will use the term \emph{assembly} interchangeably with supertile.

A \emph{(two-handed) tile assembly system} (\emph{TAS}) is an ordered triple $\mathcal{T} = (T, S, \tau)$, where $T$ is a finite set of tile types, $S$ is the \emph{initial state}, and $\tau\in\N$ is the temperature.  For notational convenience we sometimes describe $S$ as a set of supertiles, in which case we actually mean that $S$ is a multiset of supertiles with one count of each supertile. We also assume that, in general, unless stated otherwise, the count for any single tile in the initial state is infinite.  Commonly, 2HAM systems are defined as pairs $\mathcal{T} = (T, \tau)$, with the initial state simply consisting of an infinite number of copies of each singleton tile type of $T$, and throughout this paper this is the notation we will use.

Given a TAS $\calT=(T,\tau)$, a supertile is \emph{producible}, written as $\alpha \in \prodasm{T}$, if either it is a single tile from $T$, or it is the $\tau$-stable result of translating two producible assemblies without overlap.
A supertile $\alpha$ is \emph{terminal}, written as $\alpha \in \termasm{T}$, if for every producible supertile $\beta$, $\alpha$ and $\beta$ cannot be $\tau$-stably attached.
A TAS is \emph{directed} if it has only one terminal, producible supertile. A set, or shape, $X$ \emph{strictly self-assembles} if there is a TAS $\mathcal{T}$ for which every assembly $\alpha\in\termasm{T}$ satisfies $\dom \alpha = X$. Essentially, strict self-assembly means that tiles are only placed in positions defined within the shape.  This is in contrast to the notion of \emph{weak self-assembly} in which only specially marked tiles can and must be in the locations of $X$ but other locations can perhaps receive tiles of other types.  All results in this paper are for strict self-assembly of shapes.

\subsection{Informal description of the STAM}
The STAM, as formulated, is intended to provide a model based on experimentally plausible mechanisms for glue activation and deactivation. %
A detailed, technical definition of the STAM model is provided in \cite{Signals3DArxiv}.

In the STAM, tiles are allowed to have sets of glues on each edge (as opposed to only one glue per side as in the aTAM and 2HAM).  Tiles have an initial state in which each glue is either ``\on/'' or ``\latent/'' (i.e. can be switched \on/ later).  Tiles also each implement a transition function which is executed upon the binding of any glue on any edge of that tile.  The transition function specifies, for each glue $g$ on a tile, a set of glues (along with the sides on which those glues are located) and an action, or \emph{signal} which is \emph{fired} by $g$'s binding, for each glue in the set.  The actions specified may be to: 1. turn the glue \on/ (only valid if it is currently \latent/), or 2. turn the glue \off/ (valid if it is currently \on/ or \latent/).  This means that glues can only be \on/ once (although may remain so for an arbitrary amount of time or permanently), either by starting in that state or being switched \on/ from \latent/ (which we call \emph{activation}), and if they are ever switched to \off/ (called \emph{deactivation}) then no further transitions are allowed for that glue.  This essentially provides a single ``use'' of a glue (and the signal sent by its binding).  Note that turning a glue \off/ breaks any bond that that glue may have formed with a neighboring tile. Also, since tile edges can have multiple active glues, when tile edges with multiple glues are adjacent, it is assumed that all matching glues in the \on/ state bind (for a total binding strength equal to the sum of the strengths of the individually bound glues).  The transition function defined for a tile type is allowed a unique set of output actions for the binding event of each glue along its edges, meaning that the binding of any particular glue on a tile's edge can initiate a set of actions to turn an arbitrary set of the glues on the sides of the same tile \on/ or \off/.

As the STAM is an extension of the 2HAM, binding and breaking can occur between tiles contained in pairs of arbitrarily sized supertiles.  It was designed to model physical mechanisms which implement the transition functions of tiles but are arbitrarily slower or faster than the average rates of (super)tile attachments and detachments.  Therefore, rather than immediately enacting the outputs of transition functions, each output action is put into a set of ``pending actions'' which includes all actions which have not yet been enacted for that glue (since it is technically possible for more than one action to have been initiated, but not yet enacted, for a particular glue). Any event can be randomly selected from the set, regardless of the order of arrival in the set, and the ordering of either selecting some action from the set or the combination of two supertiles is also completely arbitrary.  This provides fully asynchronous timing between the initiation, or firing, of signals (i.e. the execution of the transition function which puts them in the pending set) and their execution (i.e. the changing of the state of the target glue), as an arbitrary number of supertile binding events may occur before any signal is executed from the pending set, and vice versa.  %

An STAM system consists of a set of tiles and a temperature value.  To define what is producible from such a system, we use a recursive definition of producible assemblies which starts with the initial tiles and then contains any supertiles which can be formed by doing the following to any producible assembly:  1. executing any entry from the pending actions of any one glue within a tile within that supertile (and then that action is removed from the pending set), 2. binding with another supertile if they are able to form a $\tau$-stable supertile, or 3. breaking into $2$ separate supertiles along a cut whose total strength is $< \tau$.

\subsection{Discrete Self-Similar Fractals}

We define $\mathbb{N}_g$ as the subset $\{0,1,...,g-1\}$ of $\mathbb{N}$, and
if $A,B \subseteq \mathbb{N}^2$, then $A+(x,y)B = \{(x_a,y_a) + (x\cdot x_b,y \cdot y_b) | (x_a,y_a) \in A$ and $(x_b,y_b) \in B\}$.  We then define discrete self-similar fractals as follows:

We say that $\bf{X} \subset \mathbb{N}^2$ is a \emph{discrete self-similar fractal} (or \emph{dssf} for short) if there
exists a set $\{(0,0)\} \subset G \subset \mathbb{N}^2$ where $G$ is connected, $w_G = \max(\{x|(x,y) \in G\})+1$, $h_G = \max(\{y|(x,y)\in G\})+1$, $w_G$ and $h_G > 1$, and $G \subsetneq \mathbb{N}_{w_G} \times \mathbb{N}_{h_G}$, such that $\bf{X}$ $ = \bigcup^\infty_{i=1} X_i$, where $X_i$, the $i^{th}$ stage of $\bf{X}$, is defined by $X_1 = G$ and $X_{i+1} = X_i + (w_G^i,h_G^i)G$.  We say that $G$ is the generator of $\bf{X}$.  Essentially, the generator is a connected set of points in $\mathbb{N}^2$ containing $(0,0)$, points at both $x>0$ and $y>0$, and is not a completely filled rectangle.  Every stage after the generator is composed of copies of the previous stage arranged in the same pattern as the generator.

A connected discrete self-similar fractal is one in which every component is connected in every stage, i.e. there is only one connected component in the grid graph formed by the points of the shape.

Figure~\ref{fig:triangle} shows, as an example, the first $4$ stages of the discrete self-similar fractal known as the Sierpinski triangle.  In this example, $G = \{(0,0),(1,0),(0,1)\}$.

\begin{figure}[htp]
\begin{center}
\includegraphics[width=2.0in]{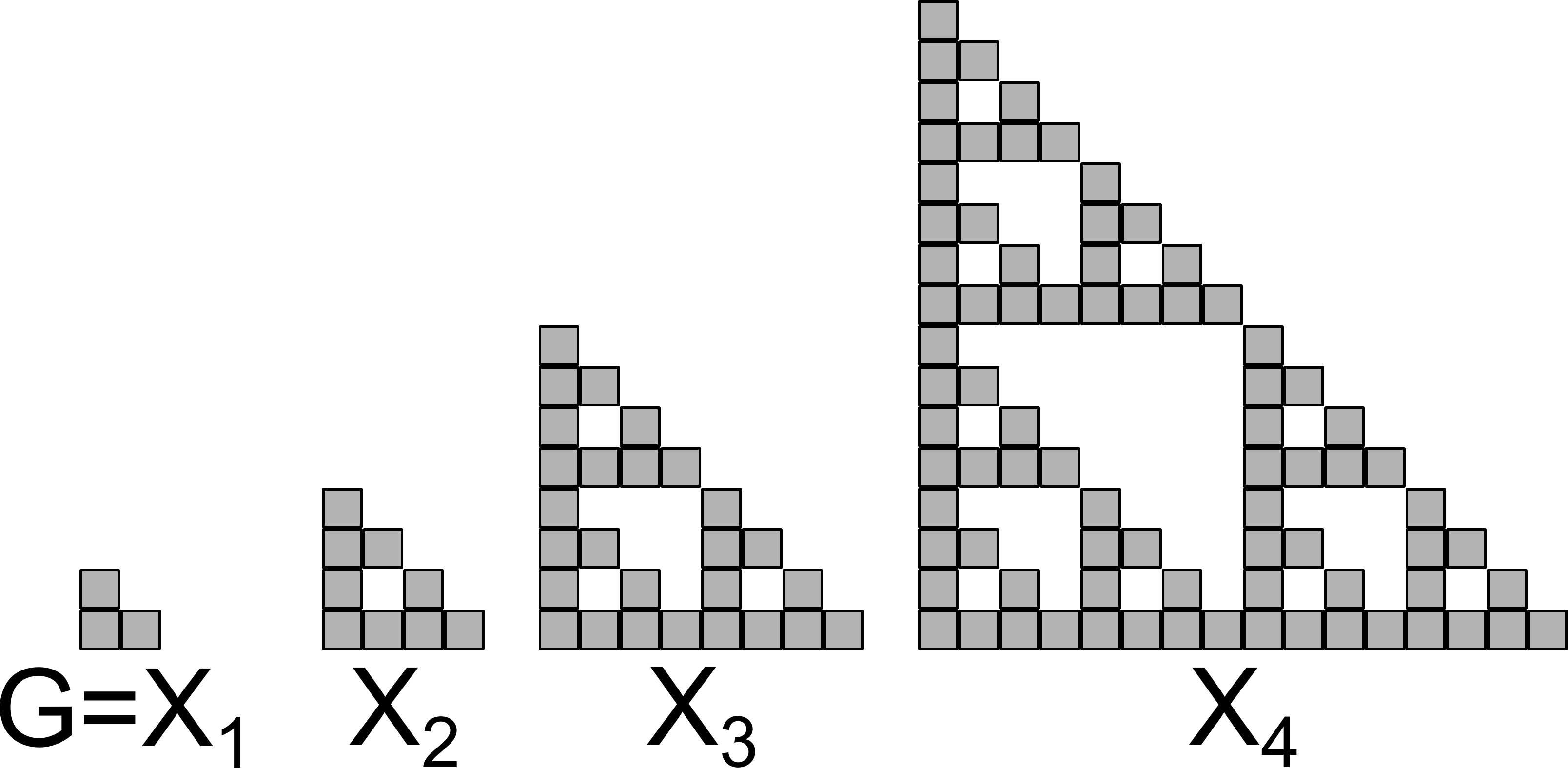}
\caption{Example discrete self-similar fractal:  the first $4$ stages of the Sierpinski triangle}
\label{fig:triangle}
\end{center}
\end{figure}

We also define a subset of connected discrete self-similar fractals which we call \emph{singly-concave} as containing any connected discrete self-similar fractal $\calF$ such that, if stage 2 of $\calF$, $\calF_2$, is contained within a bounding box, on the straight line path $p$ from any point on the bounding box into the first location adjacent to $\calF_2$, the set of all edges along which $p$ is adjacent to $\calF_2$ are contiguous. Intuitively, singly-concave fractals do not have concavities which occur within the ``sides'' of other concavities.  Examples can be seen in Figure~\ref{fig:singly-concave-examples}.

\begin{figure}[htp]
\centering
\includegraphics[width=3.0in]{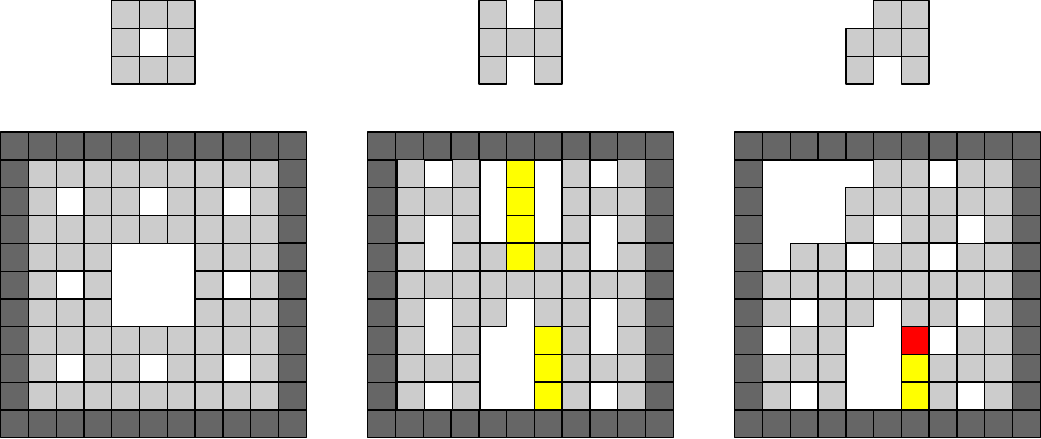}
\caption{Three example generators, and their associated second stages contained within bounding boxes (darker grey).  The left two are singly-concave because any paths (examples in yellow) from the bounding boxes to the edge of their second stage tiles meet only contiguous sets of edges of the fractal.  The rightmost isn't because the red and yellow tiles of one path meet two non-contiguous sets.}
\label{fig:singly-concave-examples}
\end{figure}

\def\on/{{\texttt{on}}}
\def\off/{{\texttt{off}}}
\def\latent/{{\texttt{latent}}}

\section{Strict Self-Assembly of the Discrete Sierpinski Triangle}

\begin{theorem}\label{thm:triangle}
There exists an STAM system $\calT_\Delta$ = (T, 1) such that $\calT_\Delta$ has exactly one infinite terminal supertile $\alpha_\Delta$, and $\dom(\alpha_\Delta) = S_\Delta$, i.e. is exactly the discrete Sierpinski triangle, and for all $\alpha \in \termasm{T}$ such that $\alpha \not= \alpha_\Delta, |\dom(\alpha)| \le 4$.
\end{theorem}

\begin{proof}
We prove Theorem~\ref{thm:triangle} by construction, and thus present an STAM tile assembly system $\calT_\Delta$ and show that it strictly self assembles $S_\Delta$, while any assemblies which detach from the assembly (or otherwise form) during its growth (which we call ``junk'' assemblies) all become terminal at sizes $\le 4$.  At a high level, $\calT_\Delta$ uses 2HAM principles (i.e. combinations of large supertiles) to combine a northern, southern, and western version of each stage $n$ for $1 < n < \infty$ through geometric matching, to produce stage $n + 1$.
\ifnaco
\else
In this section we provide a functional overview of the construction, and full technical details and tile type definitions can be found in Section~\ref{sec:tricon-details}.
\fi

From a ``hard-coded'' start at stage two (i.e. base tiles initially combine to form this stage before allowing formation of subassemblies), each stage $n$ must completely grow before the subassemblies that make that stage are able to combine and form the next stage, $n + 1$.  Only when an assembly representing stage $n$ is completely built can an \textit{initiator} tile attach to it and turn \on/ a specific glue that allows for nondeterministic binding of one of three tiles, which then tells that copy of stage $n$ to become one of three substages for stage $n + 1$.  By definition of the Sierpinski triangle, there are three substages of each stage that correspond to the three points in the generator, i.e. $(0,0)$, $(1,0)$, and $(0,1)$.  The nondeterministic binding of one of the initiator types initiates an assembly sequence which grows either a \textit{tooth} or \textit{gap} on the assembly to which it is attached.  A tooth is a one-tile protrusion from a flat surface, and a gap is a one-tile cavity in a flat surface (see Figure~\ref{fig:stage1} for examples).  One produces a southern tooth and becomes the northwest portion of $S_\Delta$ stage \textit{n} + 1, and a second produces a western tooth and becomes the southeast portion of $S_\Delta$ stage \textit{n} + 1; the third goes through two main phases, first filling in along the diagonal to roughly make a square with a gap in the north face and then opening a gap in the east after connecting to the northern piece to allow its connection to the eastern piece.  We will call the substage assemblies $S_{\Delta{s}}$, $S_{\Delta{w}}$, and $S_{\Delta{u}}$, respectively, and the tile sets (which are subsets of $T$) that form them $T_{\Delta{s}}$, $T_{\Delta{w}}$, and $T_{\Delta{u}}$.  Note that the glues that allow connections of the substages are only activated after the necessary geometry is in place to verify the sizes of the complementary pieces; after the substage connections, all tiles not within the domain of stage $n$ fall off of the assembly.

\begin{figure}[htp]
\centering
\includegraphics[width=\linewidth]{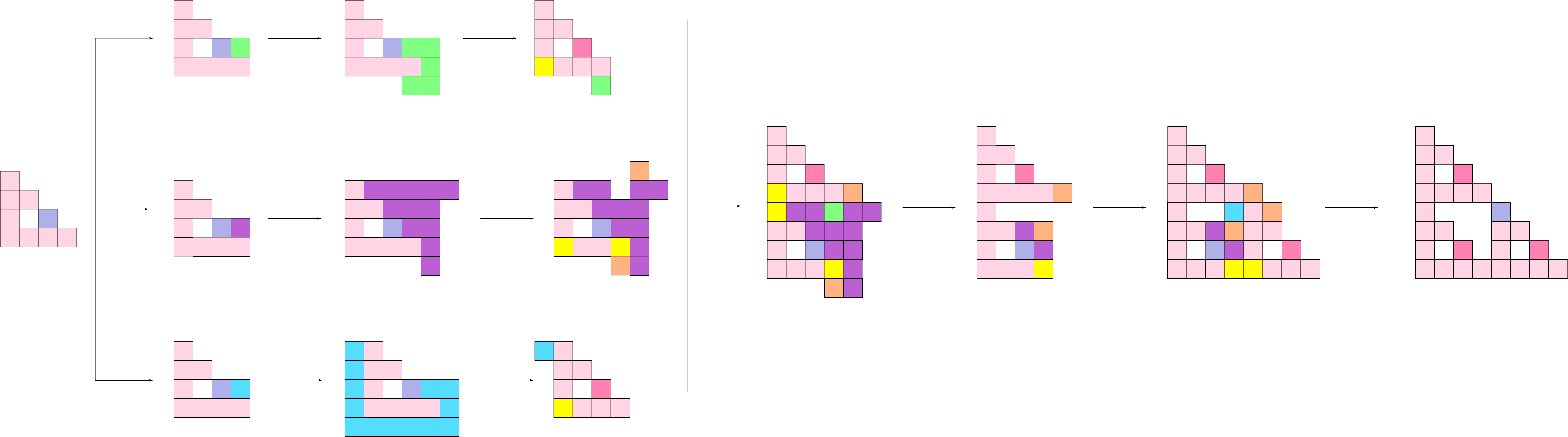}
\caption{High-level formation process of the stage of $S_\Delta$ immediately following the initial stage two formation.  From top to bottom are shown $S_{\Delta{s}}$, $S_{\Delta{u}}$, and $S_{\Delta{w}}$.}
\label{fig:stage1}
\end{figure}

As depicted in Figure~\ref{fig:stage1}, $S_{\Delta{s}}$ and $S_{\Delta{u}}$ are the first subassemblies to combine with attachment points on the southwest and northwest corners of their respective assemblies.  The southern tooth of $S_{\Delta{s}}$, depicted in green, aligns with the slot created in the $S_{\Delta{u}}$ assembly, shown in purple.  The two assemblies can only align when they are of the proper size due to the orange blocker tile located to the immediate right of the $S_{\Delta{u}}$ slot. $S_{\Delta{u}}$ cannot decay appropriately (i.e. cause ``unwanted'' tiles to fall off) until this blocker tile is in place; only after the blocker tile binds does a series of glues activate that result in the removal of the gap tile.  After a sequence of detachments removes the uppermost row and eastmost column of the resulting assembly, a second blocker tile attaches to the southeastern corner of the $S_{\Delta{s}}$ assembly, finishing the decay and enabling the alignment of the $S_{\Delta{w}}$ assembly of the same stage.  A final series of decay removes all other tiles that do not fit with the formation of stage $n$ of the Sierpinski triangle.

\begin{figure}[htp]
\centering
\includegraphics[width=\linewidth]{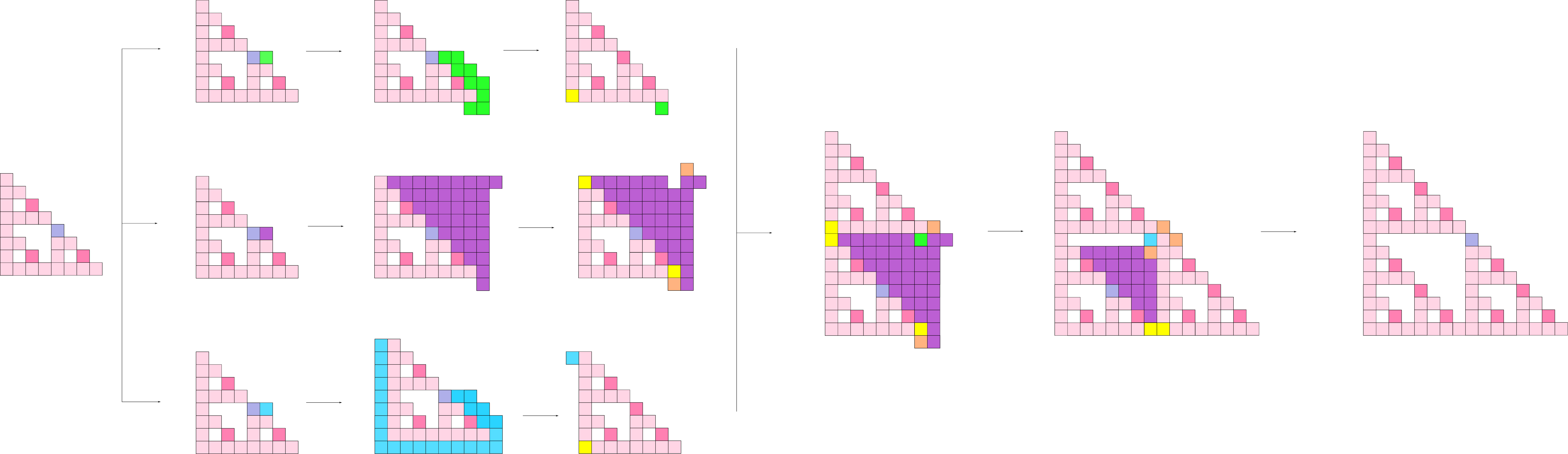}
\caption{High-level formation process of the fourth stage of $S_Δ$ following formation of stage three.  From top to bottom are shown $S_{\Delta{s}}$, $S_{\Delta{u}}$, and $S_{\Delta{w}}$.}
\label{fig:stage2}
\end{figure}

Assembly of substages and complete stages after the first combination depicted in Figure~\ref{fig:stage2} follow the same general pattern of creation and decay, with few notable differences.  The formation of the $S_\Delta$ subassemblies does require slightly more intricate systems of signals, largely to create the stair-step mechanism seen in $S_{\Delta{s}}$ and $S_{\Delta{w}}$.

\begin{figure}[htp]
\centering
\includegraphics[width=\linewidth]{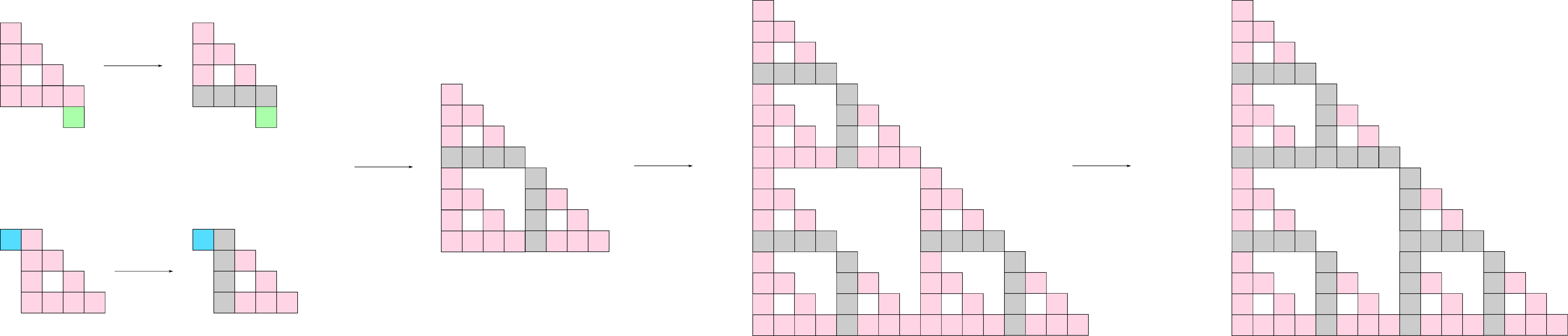}
\caption{Base tiles that carry either vertical or horizontal signals are depicted in gray; the location of the signals demonstrate that no tiles used to carry either vertical or horizontal signals are used in the same manner again, i.e. to pass signals during the formation of more than one stage.}
\label{fig:basesignals}
\end{figure}

Due to STAM properties that maintain that no tile can send its signals more than once, care has been taken to ensure that no signals sent through the tiles of the $\calT_\Delta$ tile set are used more than once.  As shown in Figure~\ref{fig:basesignals}, the two regions that send signals through the base tiles are never in a location to be used for the same horizontal or vertical signal paths more than once.

\begin{figure}[htp]
\centering
\includegraphics[width=\linewidth]{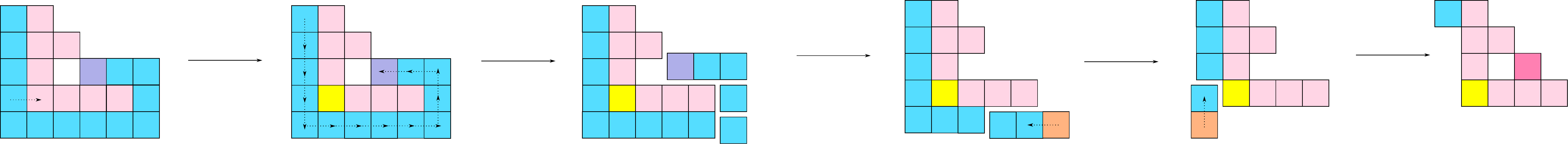}
\caption{Blocker tiles depicted in orange function alongside the blue tiles that make up the $T_{\Delta{w}}$ set to ensure that potential junk tiles do not negatively affect any further assembly.}
\label{fig:westdetach}
\end{figure}

Throughout the assembly process, junk tiles and subassemblies are continuously removed with the help of \emph{blocker} tiles.  Figure~\ref{fig:westdetach} displays this process for $S_{\Delta{w}}$.  Each junk assembly removes itself only when the appropriate signals have been passed through it and, for many assemblies, a corresponding blocker tile has attached.  This prevents active bonds that cannot be guaranteed deactivation within the asynchronous STAM model from potentially interfering with active constructions.  By binding a blocker tile, junk assemblies are created with no volatile perimeter glues. Blocker tiles also change the geometry of the junk assemblies to prevent any interference.

It is worth noting that some junk assemblies, particularly within the $T_{\Delta{u}}$ set, are capable of interacting with $\calT_\Delta$ subassemblies at various stages.  The glues that are capable of interaction on these assemblies, however, have already had their signals used and function like existing blocker and helper tiles.  This means that their interaction does not result in a negative impact on the assembly as a whole, instead assisting in the proper formation of $S_\Delta$ stages.

\begin{figure}[htp]
\centering
\includegraphics[width=\linewidth]{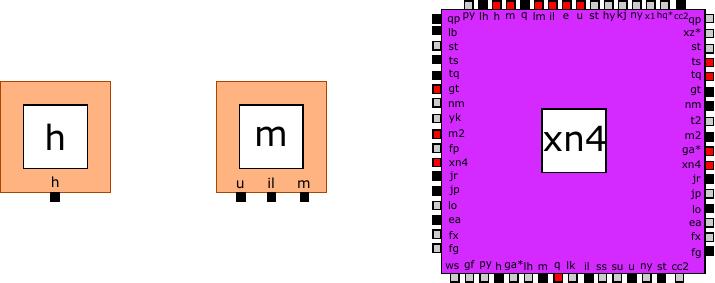}
\caption{Comparison of two helper tiles, \textit{h} and \textit{m}, alongside $T_{\Delta{u}}$ tile \textit{xn4}, after it has attached to a subassembly and then subsequently detached, displays similar glues that are \on/ (black) and \off/ (red).  All glues that are \on/ for the \textit{xn4} tile have either had their signals used, as is the case for the \textit{m}, \textit{il}, \textit{h}, and \textit{u} glues, or are not capable of interacting with subassemblies due to their counterparts turning \on/ in isolation.}
\label{fig:junk}
\end{figure}

Figure~\ref{fig:junk} depicts an example of the previously described scenario.  In this case, the \textit{xn4} tile functions like either the \textit{m} or \textit{h} tile, depending on which of its glues bind.  The only glues that remain exposed in junk assemblies are either capable of performing a similar function, or do not interact with subassembly formation due to their counterpart glues turning \on/ in isolation (i.e. a western counterpart for an eastern glue turns \on/ only when a tile has attached to the eastern face of the tile in question).

In this process in which three separate versions of any assembly at stage \textit{n} form and combine through the creation and alignment of teeth and gaps to ensure proper size integration to produce stage \textit{n} + 1, and as throughout the process of assembly the junk assemblies are detaching in constant sized pieces that will not interact with the assembly in a negative manner, the correct strict self assembly of the Sierpinski triangle at scale 1 is produced.

\ifnaco
\else
In this section we have provided a high-level overview of the construction.  The next section contains technical details and the full tile set.
\fi

\end{proof}

\ifarxiv
\section{Technical Details for the Sierpinski Triangle Construction}\label{sec:tricon-details}

Here we provide the technical details for the construction of the discrete self-similar Sierpinski Triangle at temperature one. The tile sets for the assembly are provided below.

\begin{figure}[htp]
\centering
\includegraphics[width=\linewidth]{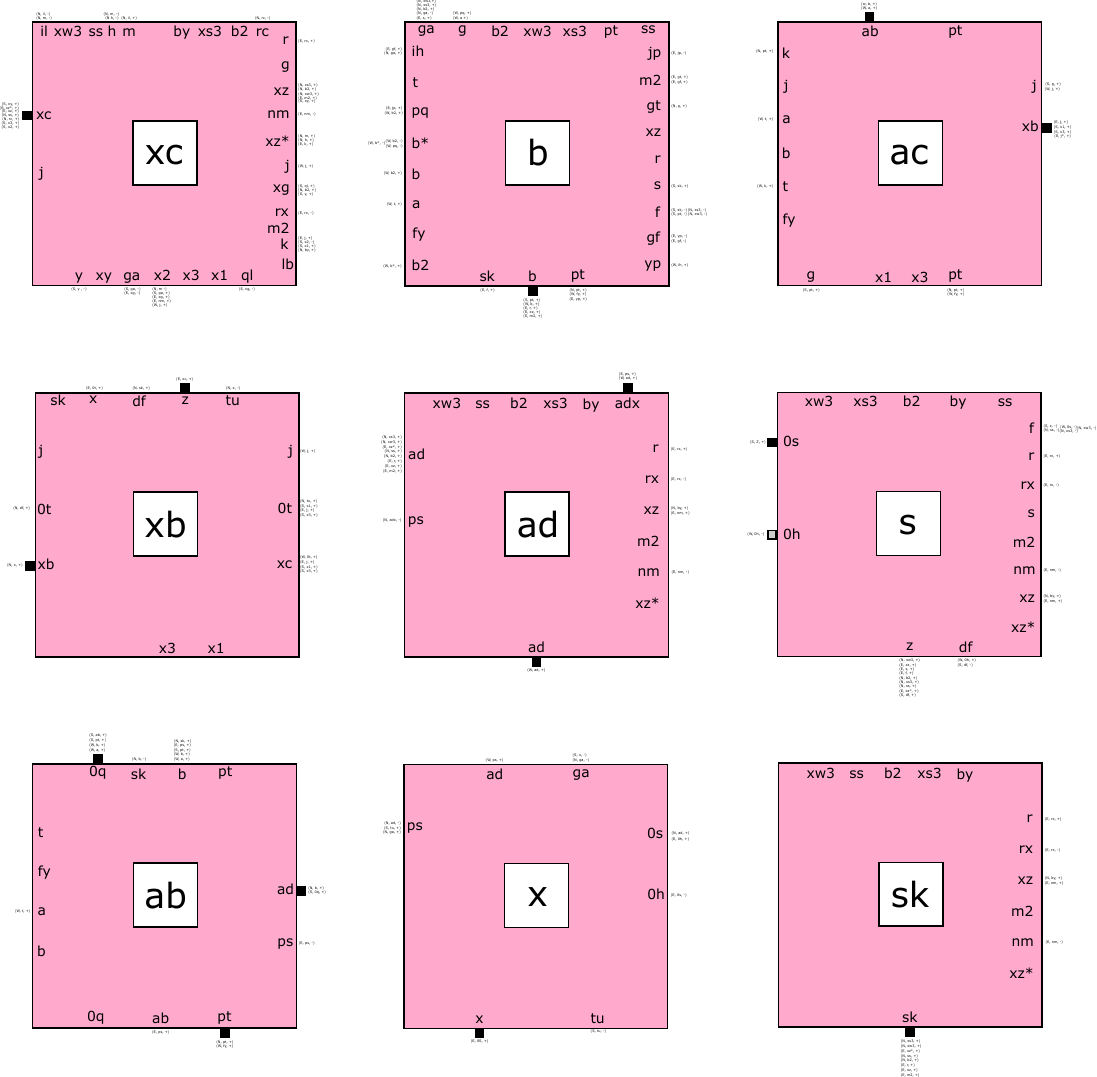}
\caption{The tiles that comprise the base of the Sierpinski triangle.  These are the tiles that form the triangle itself, not the tiles that assist in the formation and combination of subassemblies.}
\label{fig:basetile}
\end{figure}

\begin{figure}[htp]
\centering
\includegraphics[width=\linewidth]{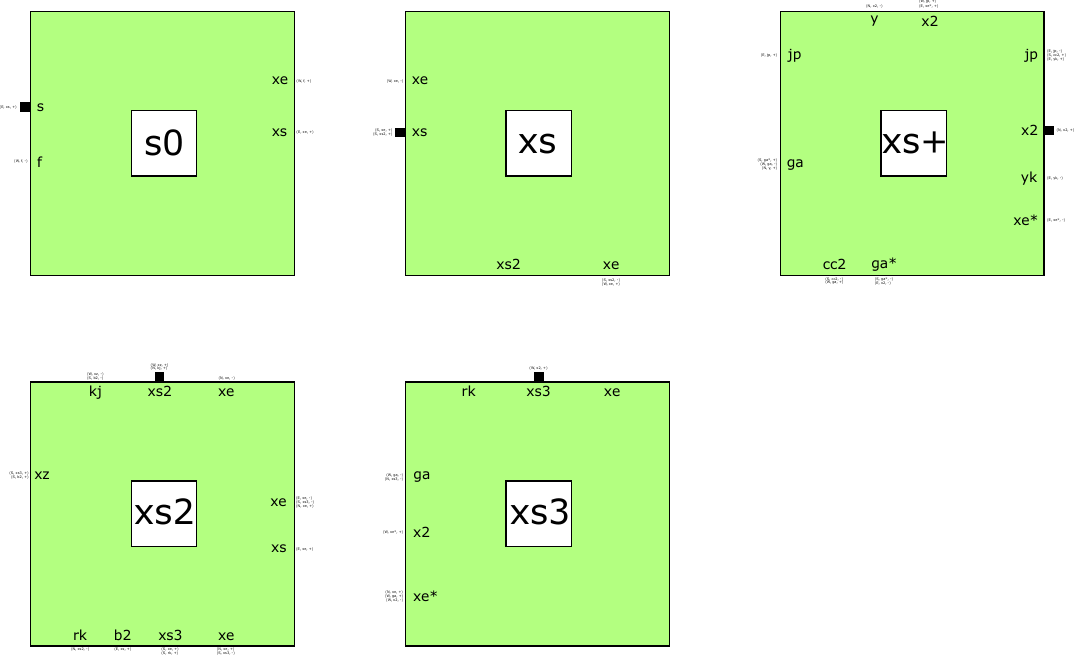}
\caption{The tiles that work to form the south tooth in the $S_{\Delta{s}}$ subassembly are shown here.}
\label{fig:southtile}
\end{figure}

\begin{figure}[htp]
\centering
\includegraphics[width=\linewidth]{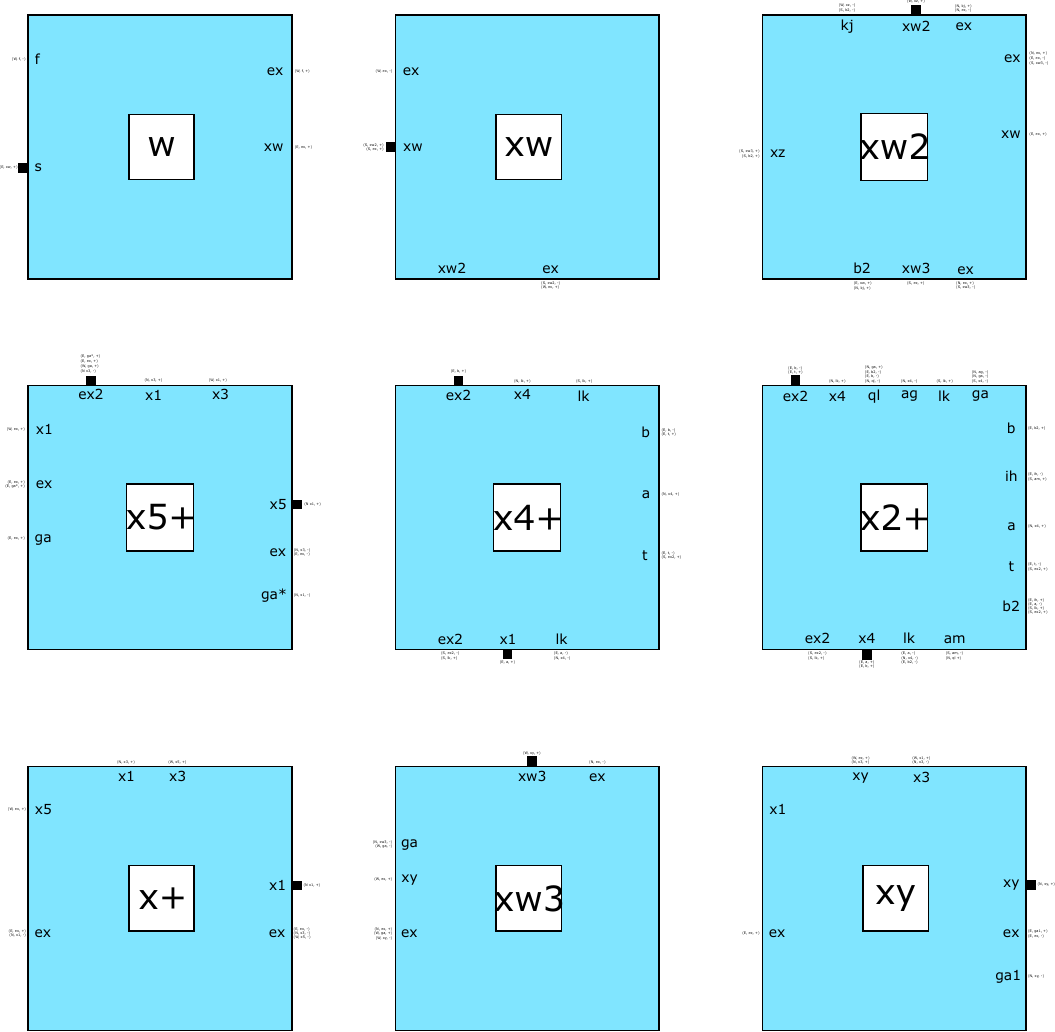}
\caption{The tiles that work to form the west tooth in the $S_{\Delta{s}}$ subassembly are shown here.}
\label{fig:westtile}
\end{figure}

\begin{figure}[htp]
\centering
\includegraphics[width=\linewidth]{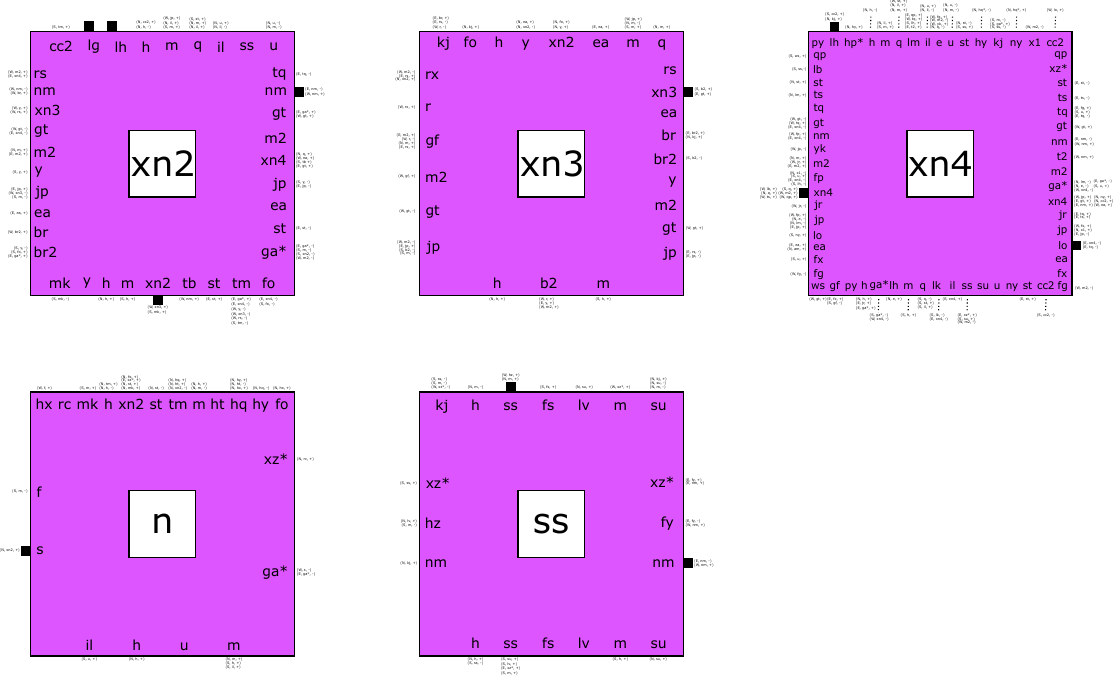}
\caption{The tiles that work to form the upper gap in the $S_{\Delta{s}}$ subassembly are shown here.}
\label{fig:uppertile}
\end{figure}

\begin{figure}[htp]
\centering
\includegraphics[width=\linewidth]{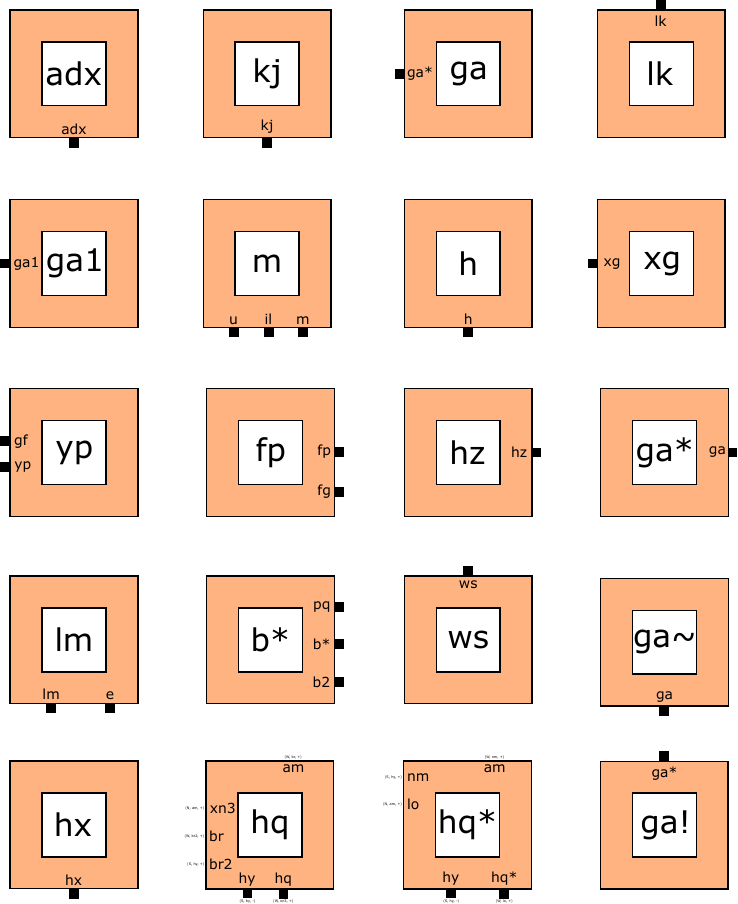}
\caption{The helper tiles shown here assist in the proper formation of subassemblies and correct combination of subassemblies.  They largely consist of blocker tiles, although some act as activator tiles for particular glues.}
\label{fig:helpertile}
\end{figure}

Helper tiles have been classified in this manner because they do not fit in the general geometry of the subassemblies, but they are essential in ensuring proper formation of $S_\Delta$.  The smaller size shown in the figures merely identifies the particular tiles as helper tiles and is not indicative of any difference in importance or function.  Many of these tiles, such as \textit{kj} and \textit{xg}, are blocker tiles that prevent the interaction of volatile glues, which cannot be guaranteed to turn \off/ before detachment from a subassembly, from interfering with further fractal growth.  Other tiles, such as \textit{m}, serve as activators for particular glues within assemblies.  These activators are external (i.e. not part of the subassembly tile set) to ensure that particular glues are only activated when exposed in a manner that cannot be guaranteed solely by tiles in subassembly tile sets.

\subsection{Formation of $S_\Delta$ Base}

Assembly of $S_\Delta$ begins with the binding of hard-coded base tiles to form the second stage of $S_\Delta$.

\begin{figure}[htp]
\centering
\includegraphics[width=\linewidth]{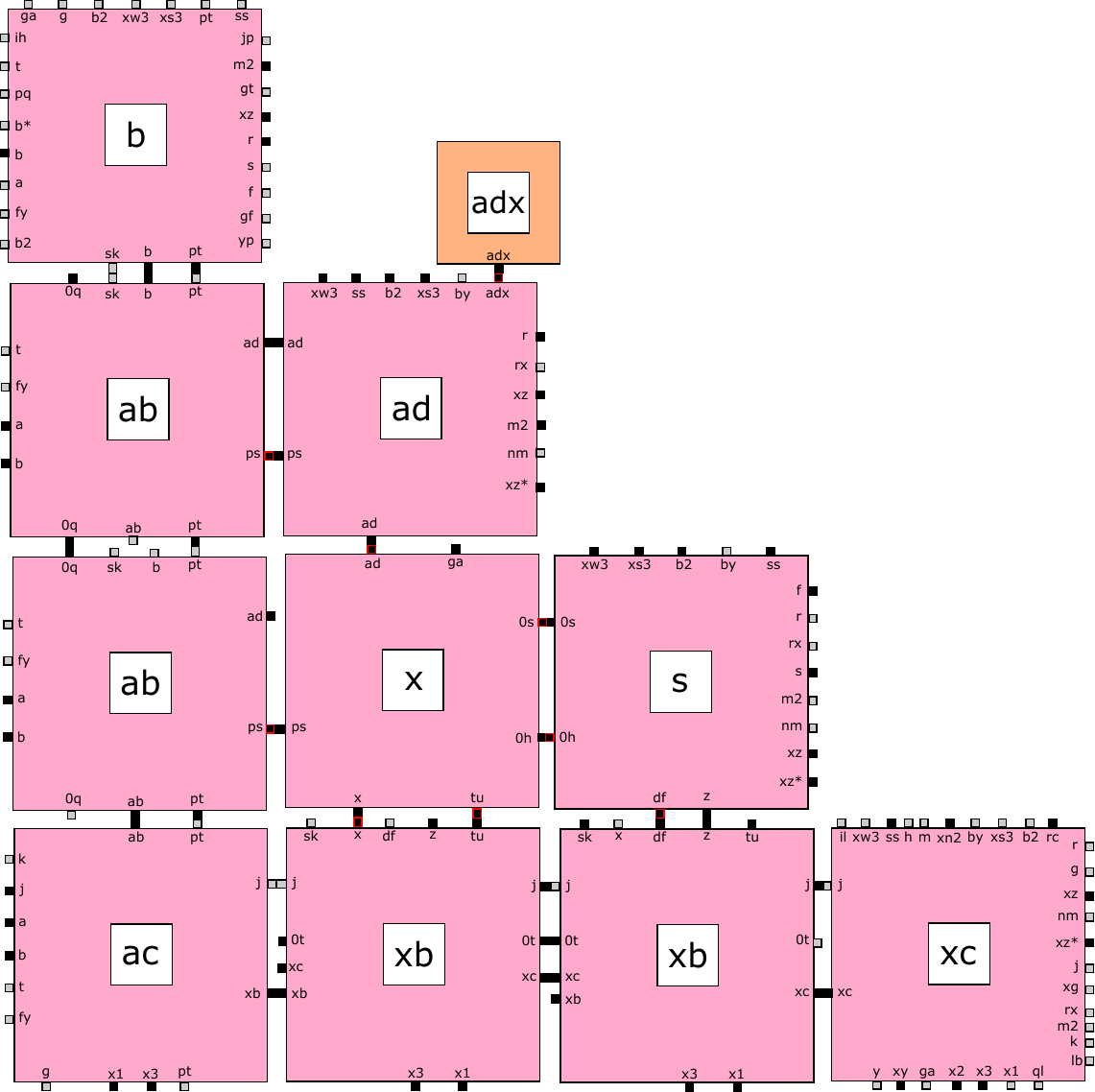}
\caption{The initial assembly of the second stage of $S_\Delta$ through the binding of hard-coded base tiles, before detachment.  Glues outlined in red are those for which their tiles have received deactivation signals but which may or may not yet have turned \off/ due to the asynchronous nature of the STAM.}
\label{fig:base1}
\end{figure}

The base tiles, shown in Figure~\ref{fig:basetile}, initially bind in a particular order that prevents interference from already formed assemblies and junk assemblies.  Key points in this binding pattern are described here and the overall assembly is displayed in Figure~\ref{fig:base1}.  To prevent rebinding of the \textit{s} tile after it detaches at the end of subassembly formation or combination, the \textit{x} tile cannot detach until its eastern \textit{0h} glue binds; this glue is turned \off/ on detached \textit{s} tiles, which will eventually turn of their western \textit{0s} glues.  To prevent rebinding of the \textit{b} tile after it detaches, the \textit{adx} tile must be in place.  Because detached \textit{b} tiles that have already been used are bound with at minimum one eastern tile, the \textit{adx} tile does not allow the junk assembly to interfere with base tile assembly.  Instead, new \textit{b} tiles are the only ones capable of joining.  Only after attachment can the base tiles turn \on/ their glues that interact with other portions of $\calT_\Delta$ to form subassemblies and expand $S_\Delta$.

\begin{figure}[htp]
\centering
\includegraphics[width=\linewidth]{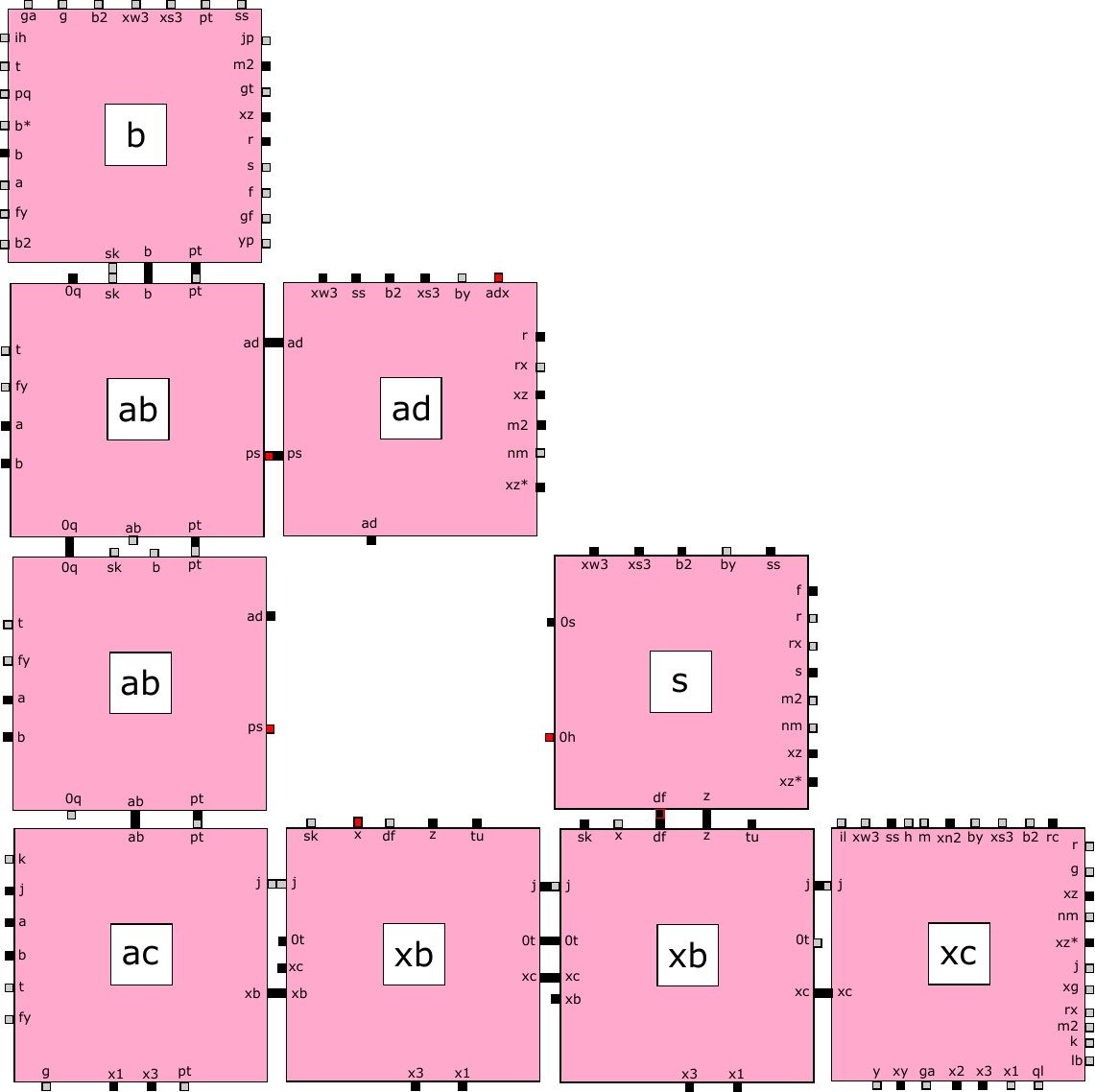}
\caption{The initial binding of base tiles after detachment of the \textit{adx} and \textit{x} tiles and before interaction with initiator tiles and formation of subassemblies.}
\label{fig:base2}
\end{figure}

After the detachment of the \textit{adx} and \textit{x} tiles, which play roles in the formation of the stage two $S_\Delta$ but do not comprise the base itself, the initial combination of base tiles of $S_\Delta$ at stage two is ready to begin formation of subassemblies.  At this point, depicted in Figure~\ref{fig:base2}, all glues necessary to form any of the three substages are \on/ and an initiator tile can bind with the initiator point (i.e. the glue \textit{s}).

Due to the nature of the STAM, subassemblies of various sizes and types are grown simultaneously.  Here, we will describe the three subassemblies $S_{\Delta{s}}$, $S_{\Delta{w}}$, and $S_{\Delta{u}}$ separately, in no relation to the order that they occur.

\subsection{Formation of the $S_{\Delta{s}}$ Subassembly}

Formation of the south tooth and the $S_{\Delta{s}}$ assembly consists of a series of tiles that wrap around the southeast portion of the base assembly before terminating at the southern face of the southeasternmost tile and beginning the series of detachment.  At larger stages, a ``stair-step'' pattern is used to move around the southeastern portion of the hypotenuse of $S_\Delta$.

\begin{figure}[htp]
\centering
\includegraphics[width=\linewidth]{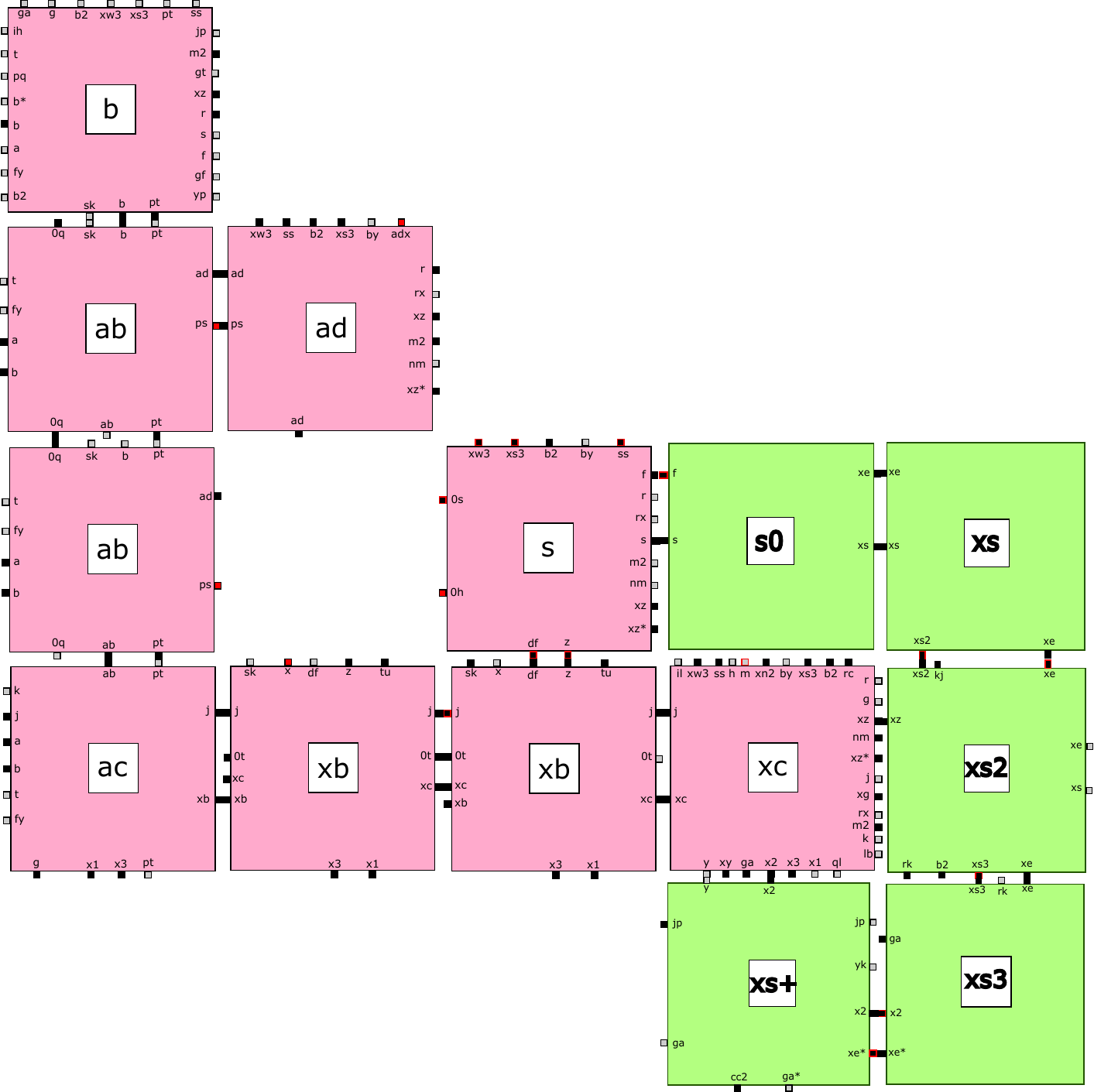}
\caption{The $S_{\Delta{s}}$ subassembly, at the first stage after initial base tile binding, after wrapping around the southeastern portion of the base tiles but before tile detachment.}
\label{fig:south1-1}
\end{figure}

South tooth assembly begins with the binding of initiator tile \textit{s0} to the base assembly.  This turns \on/ an eastern glue, which enables binding of the next tile, that then presents a southern face.  The tile that binds to this southern face senses whether it is next to a hypotenuse tile or the southeastern corner tile; at the stage depicted in Figure~\ref{fig:south1-1},it is next to the corner tile and therefore activates a southern glue.  If the \textit{xs2} tile was not next to the corner tile, as shown in Figure~\ref{fig:south2stair}, a series of alternating \textit{xs2} and \textit{xs} tiles forms the stair-step pattern.  Once the tiles reach the southeastern corner, the \textit{xs3} is added to open a western face and allow attachment of the \textit{xs+} tile.  This tile can only bind to the western glue activated on the \textit{xs3} tile and can therefore only attach under the southeastern corner tile.  Once the \textit{xs+} tile has bound, detachment can begin.

\begin{figure}[htp]
\centering
\includegraphics[width=\linewidth]{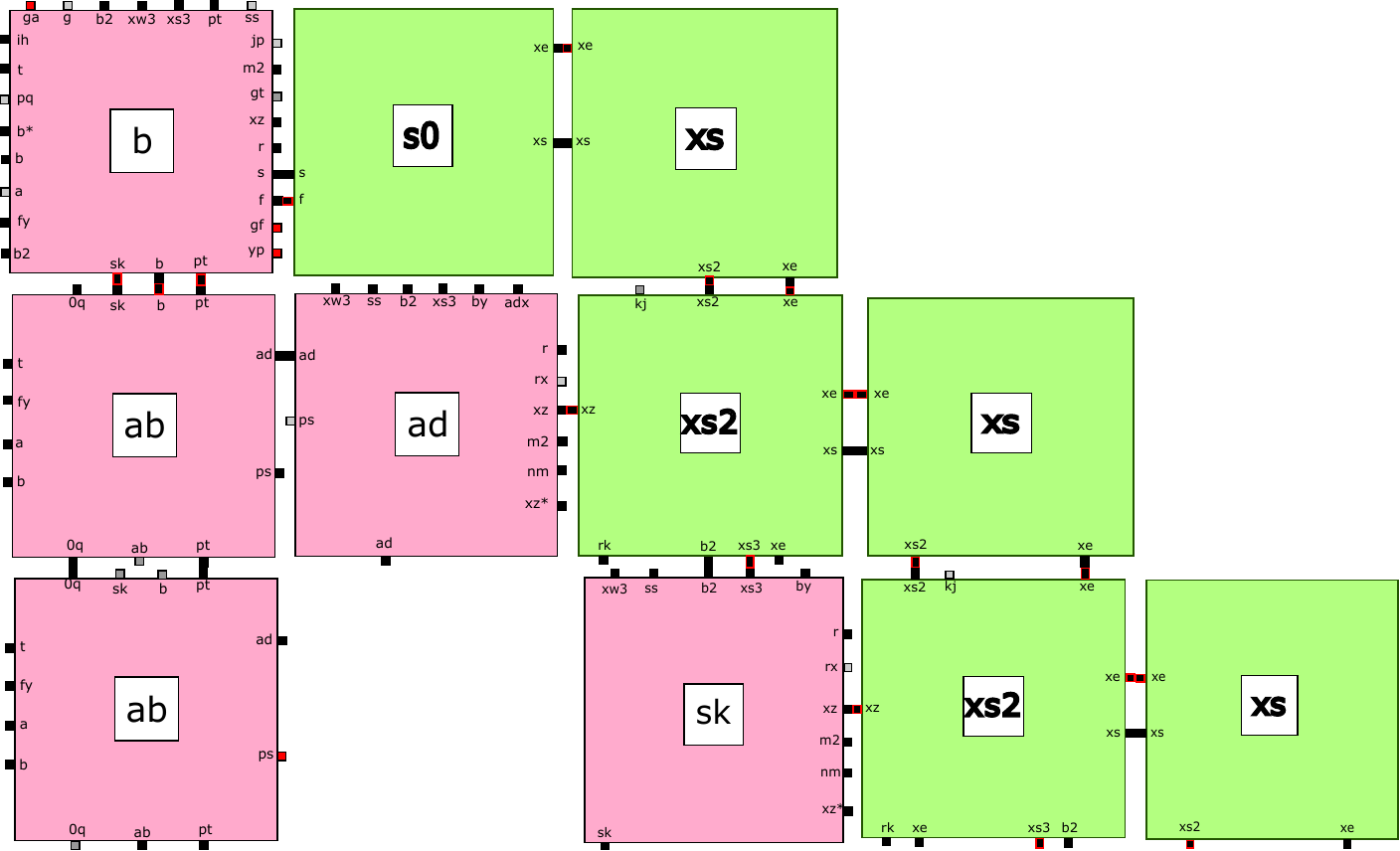}
\caption{Depiction of the stair-step pattern that forms at larger subassembly stages for $S_{\Delta{s}}$.}
\label{fig:south2stair}
\end{figure}

Stair-steps form to maneuver around the hypotenuse of the $S_\Delta$ base assembly.  General assembly of the stair-steps begin with the addition of a tile to the southern face of the tile to the right of the initiator tile (such as \textit{xs2} below the \textit{xs} tile in Figure~\ref{fig:south2stair}).  If this added tile senses that it is next to a corner of the hypotenuse, it opens an eastern glue that allows the binding of another tile that then opens a southern glue.  The process then begins again, continuing until the stair-steps reach the southeastern \textit{xc} corner tile.  At this point, different glues on the eastern face of the base tile turn \on/ a southern glue for the \textit{xs2} or similar tile, preventing further eastern expansion and causing the growing assembly to turn the corner and continue growing.

\begin{figure}[htp]
\centering
\includegraphics[width=\linewidth]{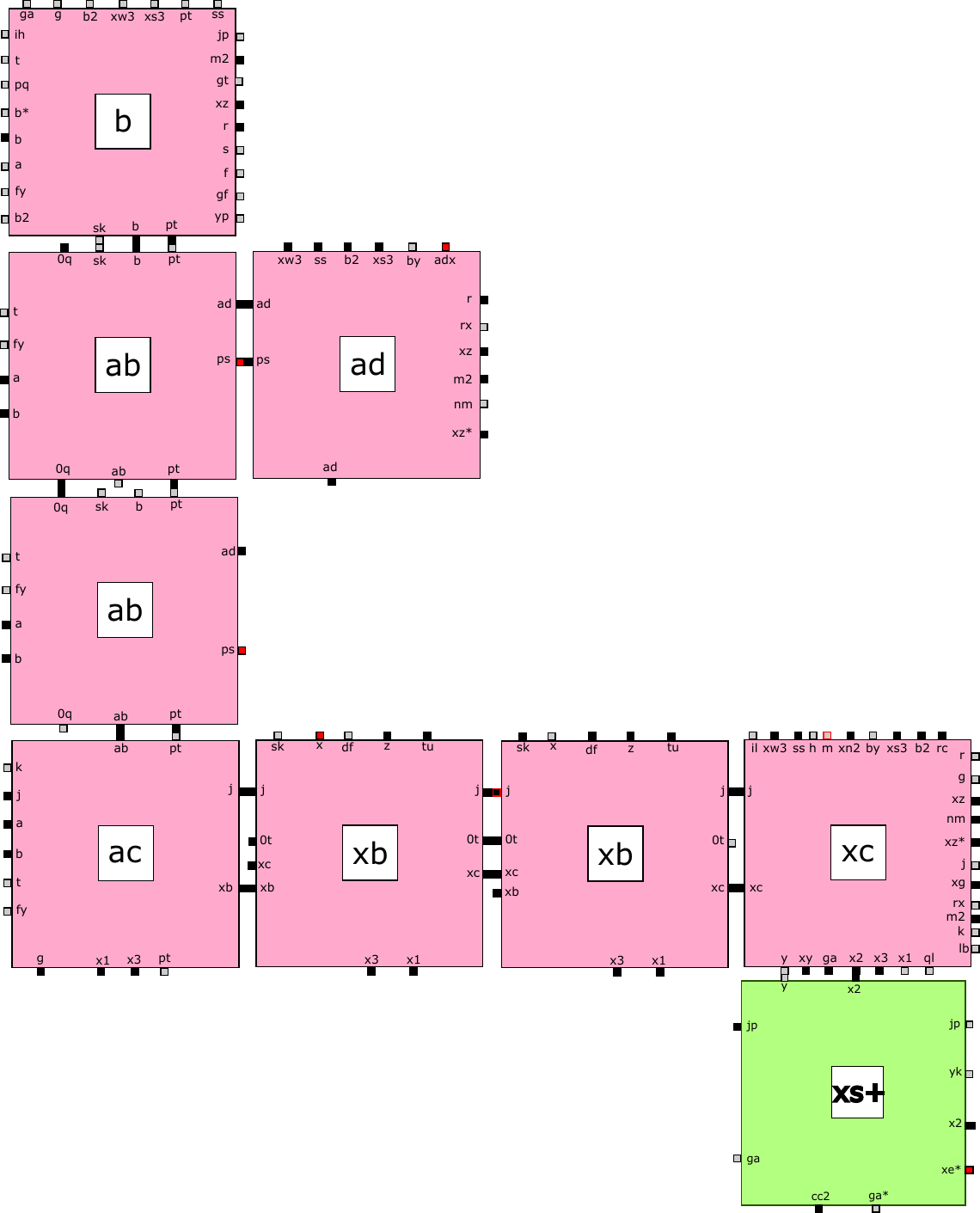}
\caption{The $S_{\Delta{s}}$ subassembly after detachment of all $\calT_{\Delta{s}}$ tiles apart from \textit{xs+}, which serves as the south tooth.}
\label{fig:south1-2}
\end{figure}

Detachment of the $\calT_{\Delta{s}}$ tiles begins with an exit signal that is sent through all attached $\calT_{\Delta{s}}$ tiles, back to the initiator tile and through to the base tile that presented the initiator glue.  Removal of this base tile is necessary to prevent interference with other assemblies — otherwise the \on/ initiator glue and the initiator tile would both have to simultaneously deactivate the same glue, something that cannot be guaranteed in the STAM due to its asynchronous nature (i.e. both tiles containing those glues could receive signals to deactivate them, but one may turn \off/ before the other, allowing detachment while the other is still temporarily \on/).  Therefore, the base tile, initiator tile, and the tile to the immediate right of the initiator tile are removed in a size 3 junk assembly.  The \textit{xs3} tile detaches with an active north face glue that could bind to other active $S_{\Delta{s}}$ subassemblies, but helper tiles prevent permanent attachment, as explained later.  The \textit{xs2} tile, and therefore the strings of 2 tiles that form the stair-steps, cannot detach until a blocker tile binds to its northern face.  This prevents interference from the northern \textit{xs2} glue that cannot be guaranteed to turn \off/ and could otherwise interact with other subassemblies.  In this manner, all $\calT_{\Delta{s}}$ tiles that do not comprise the south tooth itself are removed in junk assemblies.

\begin{figure}[htp]
\centering
\includegraphics[width=\linewidth]{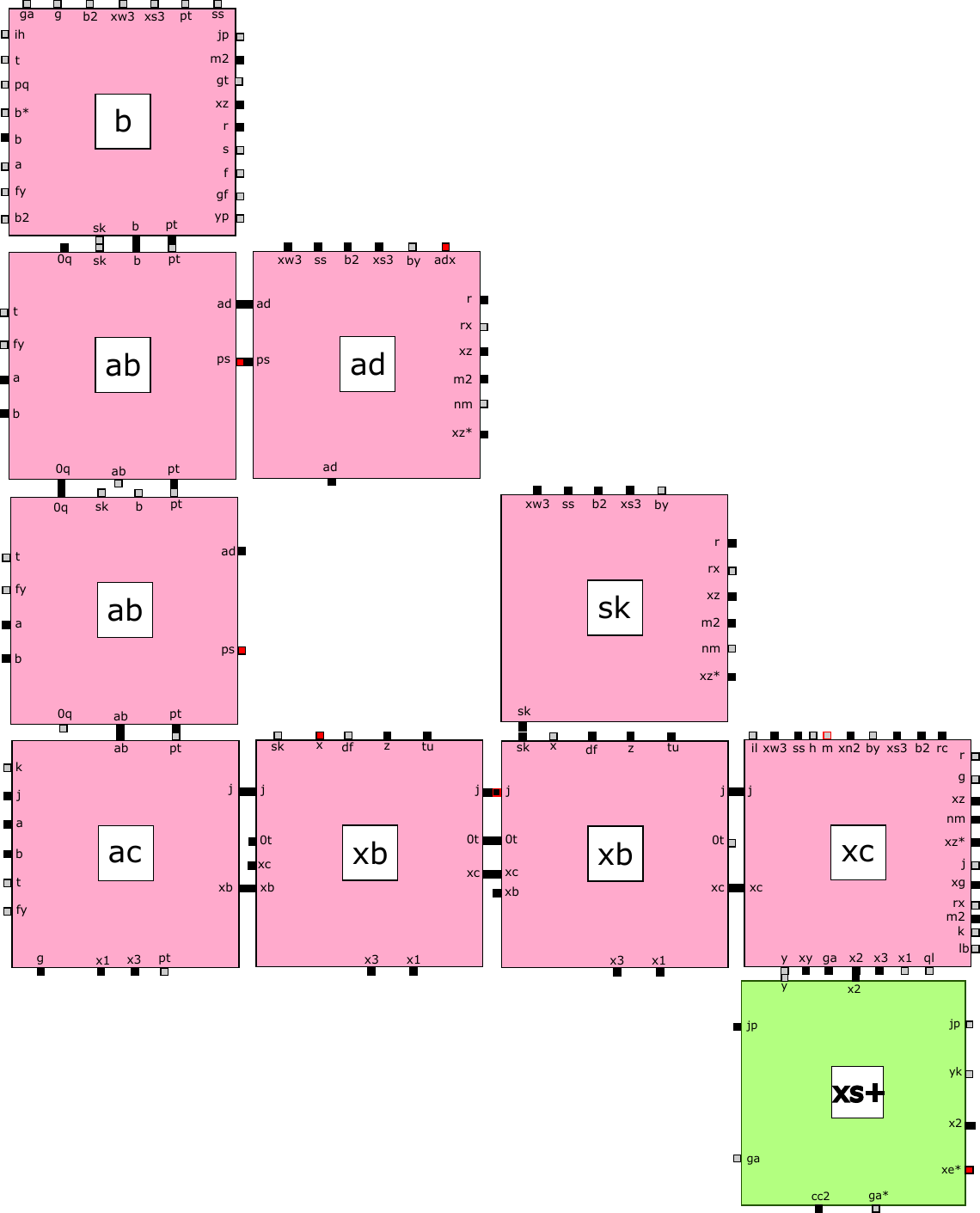}
\caption{The $S_{\Delta{s}}$ subassembly after attachment of the \textit{sk} tile to replace the removed base tile and immediately before combination begins with the $S_{\Delta{u}}$ subassembly.}
\label{fig:south1-3}
\end{figure}

The addition of the \textit{sk} tile replaces the removed base tile, explained in the previous paragraph, and reforms the $S_\Delta$ base.  It acts like the already present \textit{ad} tile.  The $S_{\Delta{s}}$ subassembly is ready for combination with the $S_{\Delta{u}}$ subassembly after this addition and after the binding glue on the southwestern corner tile has been turned \on/.  This glue is activated after the \textit{xs+} tile binds with the southeastern corner; a series of signals is sent through the base tiles to turn \on/ the glue.  Signals can be sent through the base tiles because the south row of the $S_{\Delta{s}}$ subassembly becomes the middle row of the next stage and will never need to be used in the same manner again, as depicted in Figure~\ref{fig:basesignals}.  Now, $S_{\Delta{s}}$ is ready to begin combining with other subassemblies to form the next stage of $S_\Delta$.

\begin{figure}[htp]
\centering
\includegraphics[width=\linewidth]{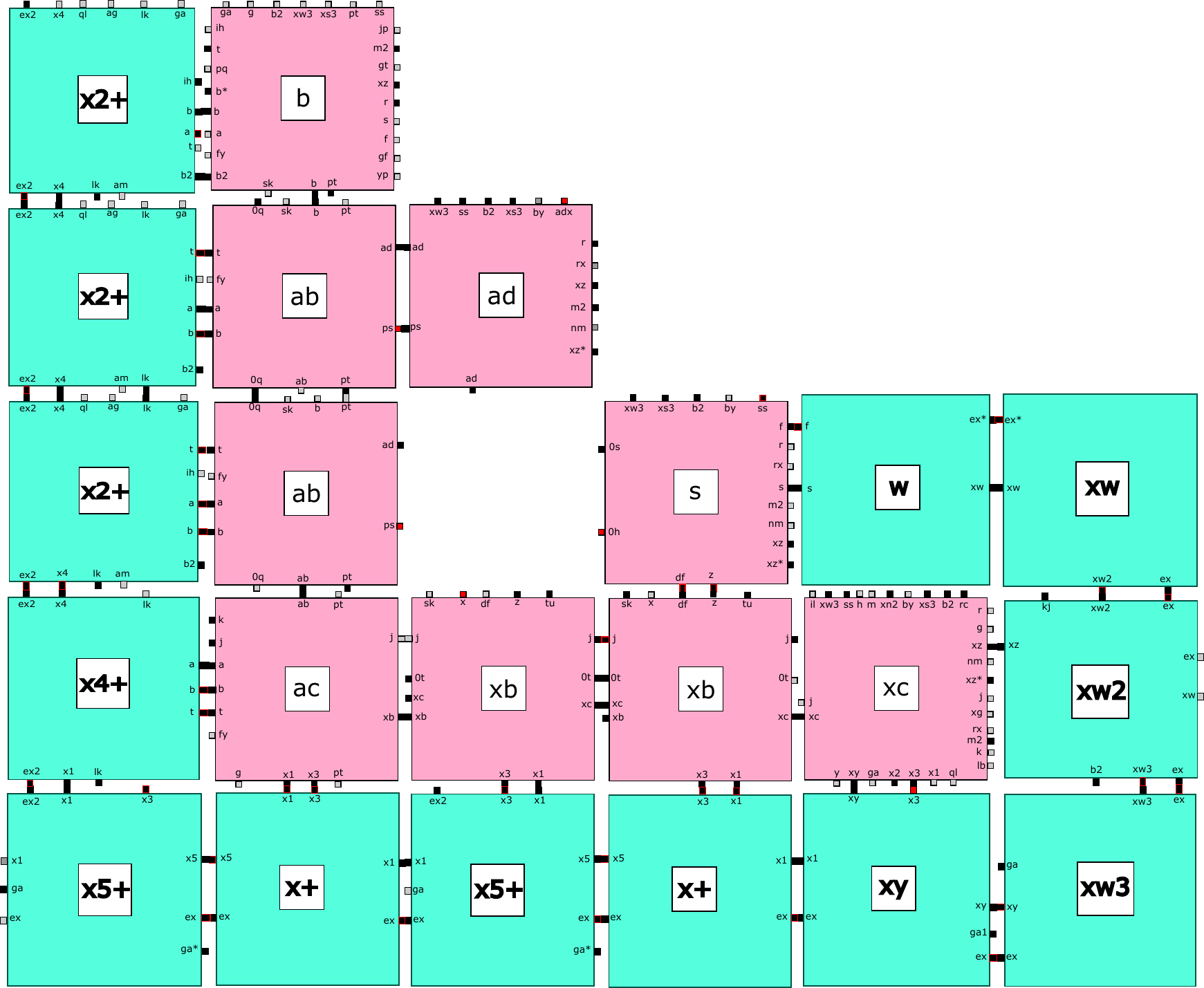}
\caption{The $S_{\Delta{w}}$ subassembly, at the first stage after initial base tile binding, after wrapping around the $S_\Delta$ base tiles but before tile detachment.}
\label{fig:west1-1}
\end{figure}

\subsection{Formation of the $S_{\Delta{w}}$ Subassembly}

Formation of the west tooth and the $S_{\Delta{w}}$ subassembly consists of a series of tiles that wrap around the southeast portion, the southern face, and western face of the base assembly before terminating at the western face of the northwesternmost tile and beginning the series of detachment (as shown in Figure~\ref{fig:west1-1}.  At larger stages, a ``stair-step'' pattern is used to move around the southeastern portion of the hypotenuse of $S_\Delta$, similar to the south tooth assembly, as seen in Figure~\ref{fig:south2stair}.  After the stair-step, two tiles attach to the south, presenting a west face.  A tile then attaches to the west, turns on a northern glue, and if the tile detects it is below the assembly, it subsequently turns on the appropriate glue on its west face.  Once the tile is unable to detect the assembly above it and has an exposed north face, a tile can attach to the north.  The tile turns \on/ an eastern glue to detect its relation to the west edge of the assembly, and opens a north face glue until it reaches the top \textit{b} tile.  The tile that attaches to the top \textit{b} tile becomes the west tooth and causes the stage to begin detachment.

\begin{figure}[htp]
\centering
\includegraphics[width=\linewidth]{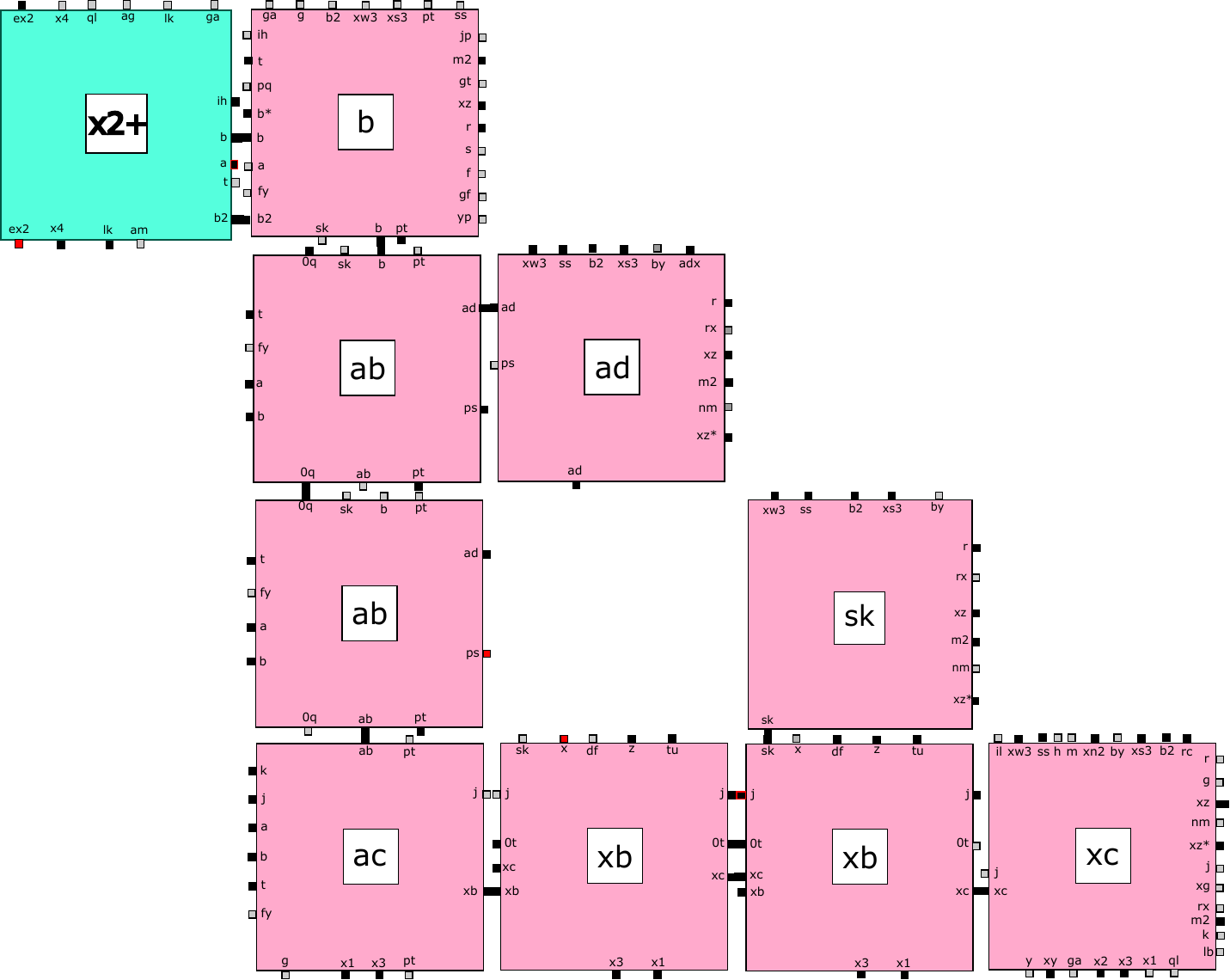}
\caption{The $S_{\Delta{w}}$ subassembly after detachment of all filler tiles, attachment of the \textit{sk} tile to replace the removed base tile, and immediately before combination begins with the other subassemblies.}
\label{fig:west1-3}
\end{figure}

Detachment of $S_{\Delta{w}}$ begins with an exit signal sent through all attached $\calT_{\Delta{w}}$, starting from the top \textit{b} tile (the tile that stays in place and becomes the west tooth).  The stair-step detaches similarly to the way the stair-step detaches in the south tooth assembly, as seen in Figure~\ref{fig:south2stair}.  After the stair-step detaches, the bottom rightmost tile detaches, allowing a helper tile to attach in its place.  The helper tile sends a signal through the bottom row of tiles, allowing another helper tile to attach to the left, and causing the bottom row to detach.  A helper tile is then able to attach to the south of the leftmost column of tiles, causing a exit signal to pass through the leftmost column and the remaining four tiles to detach.  Throughout the detachment similar precautions were taken to the precautions taken in the south tooth assembly to ensure nothing can attach to the completed stage, except via the west tooth.  For example, $\calT_{\Delta{w}}$ tiles attached to the bottom or west faces of the $S_\Delta$ base cannot detach until a blocker tile is in place to prevent interaction of volatile glues.

\subsection{Formation of the $S_{\Delta{u}}$ Subassembly}

Formation of the upper gap and the $S_{\Delta{u}}$ subassembly consists of a series of tiles that fill the space above the hypotenuse of the $S_\Delta$ stage before sending a series of signals, with the help of a blocker tile, that results in the removal of a single tile to form the upper gap.  Larger stages follow the same general principle of assembly, with the only difference being the sequence of tiles a particular detachment signal is sent through.

\begin{figure}[htp]
\centering
\includegraphics[width=\linewidth]{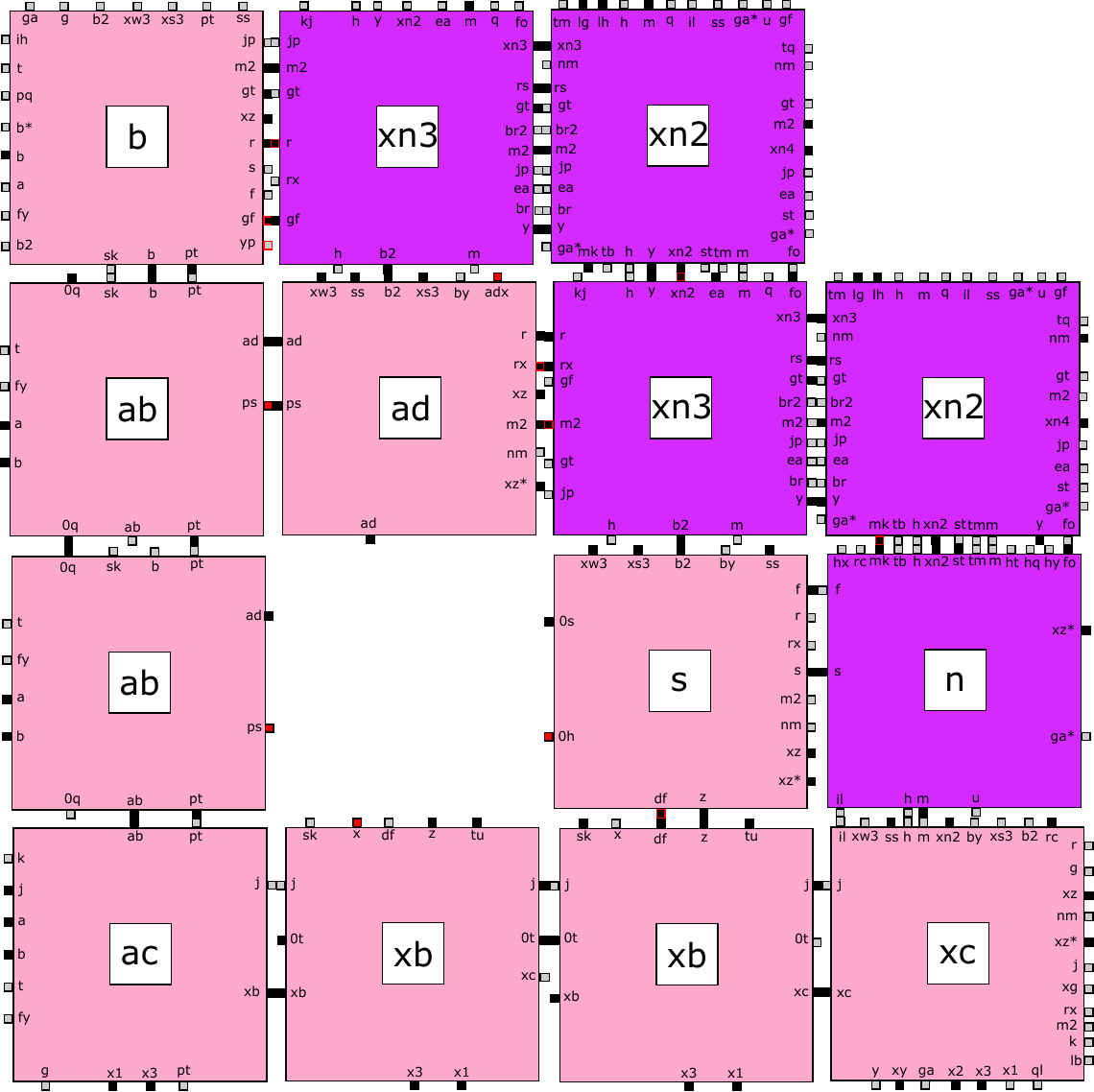}
\caption{Formation of the initial stair-step pattern for the $S_{\Delta{u}}$ subassembly.  From this point, tiles can begin attaching to the eastern portions of the subassembly.}
\label{fig:upgap1-1}
\end{figure}

Upper gap assembly begins with the binding of initiator tile \textit{n} to the base assembly.  The \textit{n} tile subsequently turns \on/ a northern glue, beginning the stair-step proliferation of tile types \textit{xn2} and \textit{xn3} along the hypotenuse of the $S_\Delta$ assembly.  After the western \textit{xn3} tile binds to a row in the stair-step, the \textit{xn2} of that row turns \on/ an eastern glue that enables further filling of space.  This eastern glue can turn \on/ before the completion of the stair-step because the filler tiles that will bind to it can only extend to a certain length, as described below, which prevents unchecked growth.

\begin{figure}[htp]
\centering
\includegraphics[width=\linewidth]{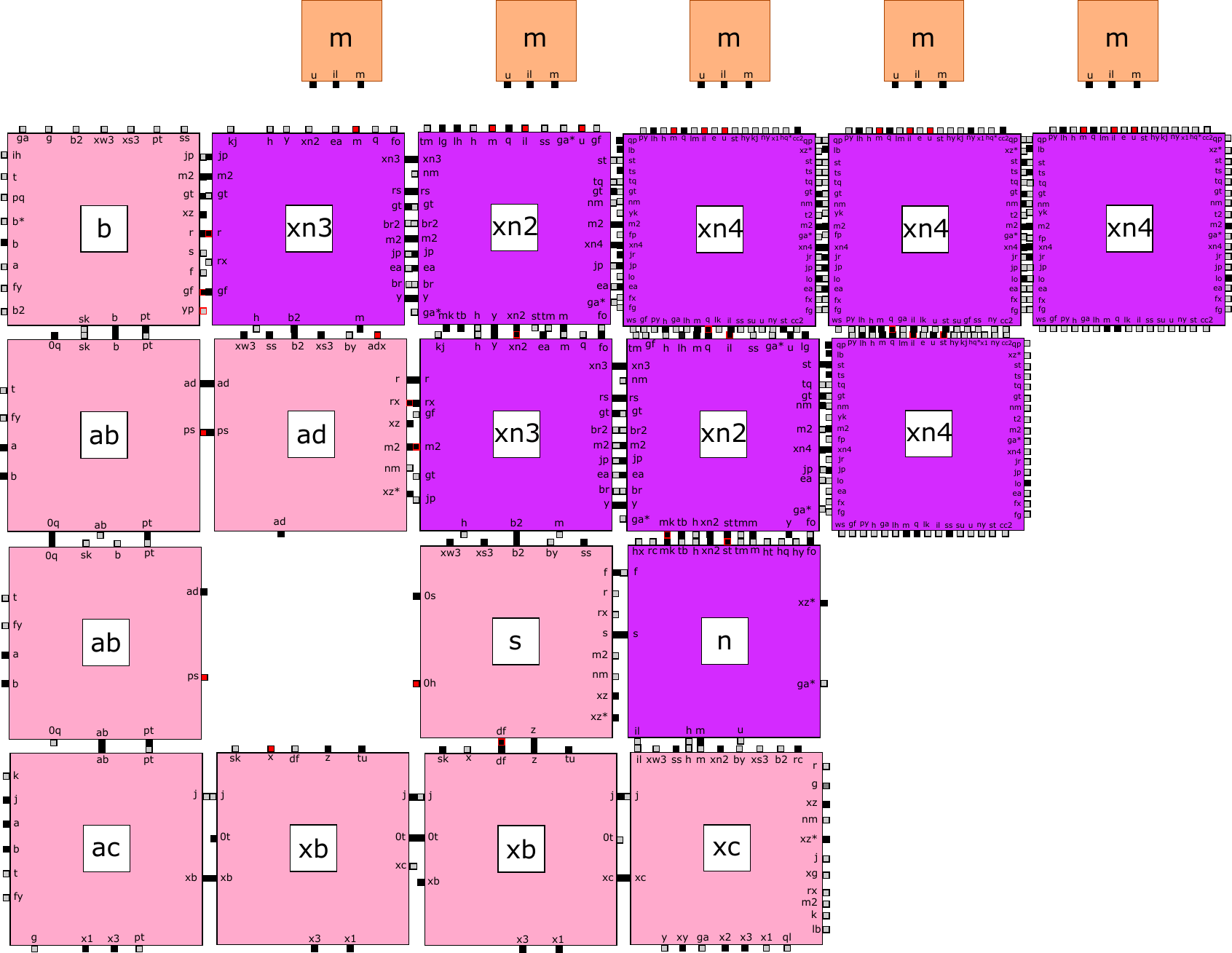}
\caption{Addition of the space-filling \textit{xn4} tiles to the $S_{\Delta{u}}$ subassembly.}
\label{fig:upgap1-2}
\end{figure}

As depicted in Figure~\ref{fig:upgap1-2}, addition of \textit{xn4} tiles fills a large portion of the space above the $S_\Delta$ assembly.  Each \textit{xn4} tile can only turn on its respective eastern glue, allowing for another \textit{xn4} tile to bind, when it has bound to a tile underneath it. However, because the \textit{n} tile that binds immediately above the initiator tile site has no \textit{xn2} or \textit{xn4} tile underneath it, it can not extend an eastern glue.  This prevents unchecked growth of the \textit{xn4} tiles to the right and results in an inverted stair-step appearance.  This inverse stair-step carries a signal that turns \on/ southern glues only on the tiles that have a southern exposed face and are part of the immediate inverse stair-step.  Additional filler tiles can then attach to these tiles and fill in the remaining area above the hypotenuse of the $S_{\Delta{u}}$ subassembly.

\begin{figure}[htp]
\centering
\includegraphics[width=\linewidth]{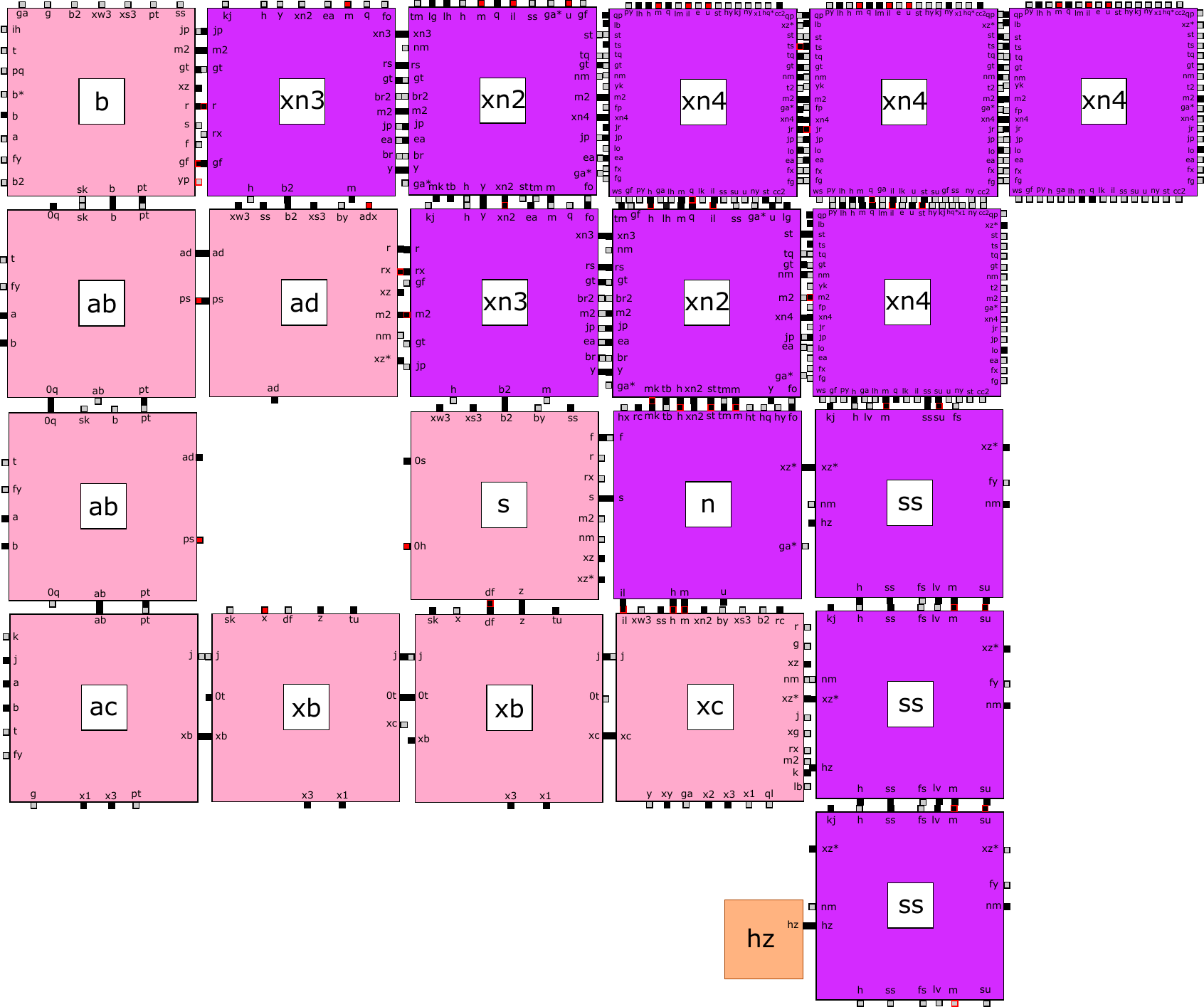}
\caption{Addition of the space-filling \textit{ss} tiles to the $S_{\Delta{w}}$ subassembly.}
\label{fig:upgap1-3}
\end{figure}

The addition of the \textit{ss} tiles takes place in a similar manner to the addition of the \textit{xn4} tiles — only when the \textit{ss} tile is next to a western tile that presents a particular glue can it turn \on/ a southern glue, allowing for further assembly.  Due to this, the \textit{ss} tile column directly to the side of the southeasternmost corner tile extends one below the south face of $S_{\Delta{u}}$.  A blocker tile binds to this extending \textit{ss} tile to prevent interference after it detaches.  The binding of a \textit{ss} tile to the southeasternmost corner tile \textit{xc} also turns \on/ the connecting glue on the eastern face of said tile.  This connecting glue is not exposed until detachment of the \textit{ss} tiles, which does not occur until after the $S_{\Delta{s}}$ and $S_{\Delta{u}}$ subassemblies combine.

\begin{figure}[htp]
\centering
\includegraphics[width=\linewidth]{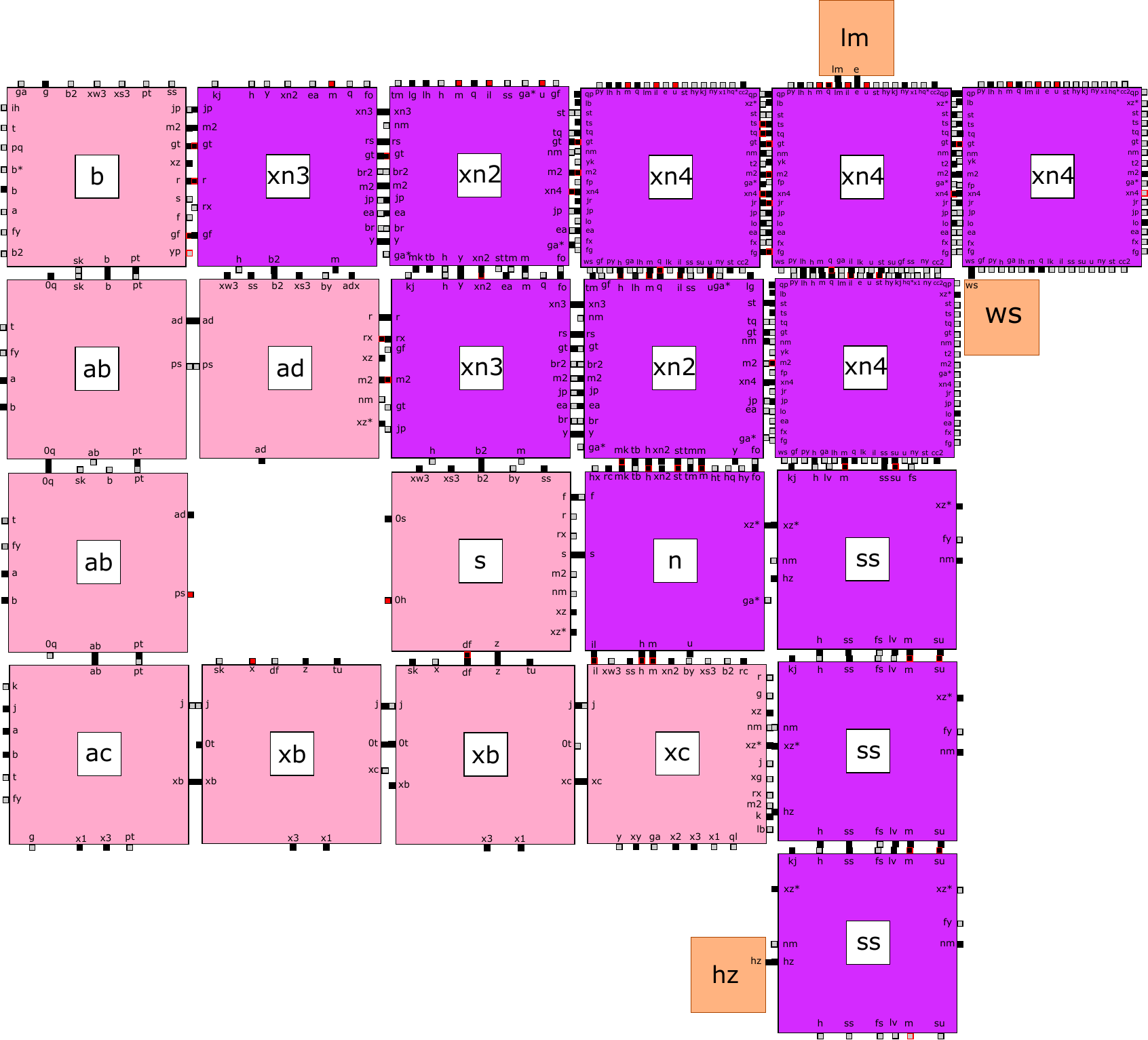}
\caption{Addition of the blocker tile to the east of the upper gap of $S_{\Delta{u}}$ to ensure combination of proper size subassemblies.}
\label{fig:upgap1-4}
\end{figure}

A series of signals that stems from the northern face of the southeastern \textit{xc} corner tile travels north, eventually signaling the tile that will be removed from the filler tiles to form the upper gap as well as the tile to which a blocker tile will bind.  This presents the only major difference between the first stage of assembly after formation of the stage two $S_\Delta$ and the other stages — this northern signal must move through a series of \textit{ss} tiles at larger stages.  This, however, does not present any changes in the overall construction.  After these signals have been sent, a blocker tile attaches to the \textit{xn4} tile that lies to the immediate right of where the gap will be.  Only after this blocker tile has attached can a series of western-moving glues be initiated, the gap tile be removed, and the connecting glue on the northwestern \textit{b} corner tile be turned \on/.

\begin{figure}[htp]
\centering
\includegraphics[width=\linewidth]{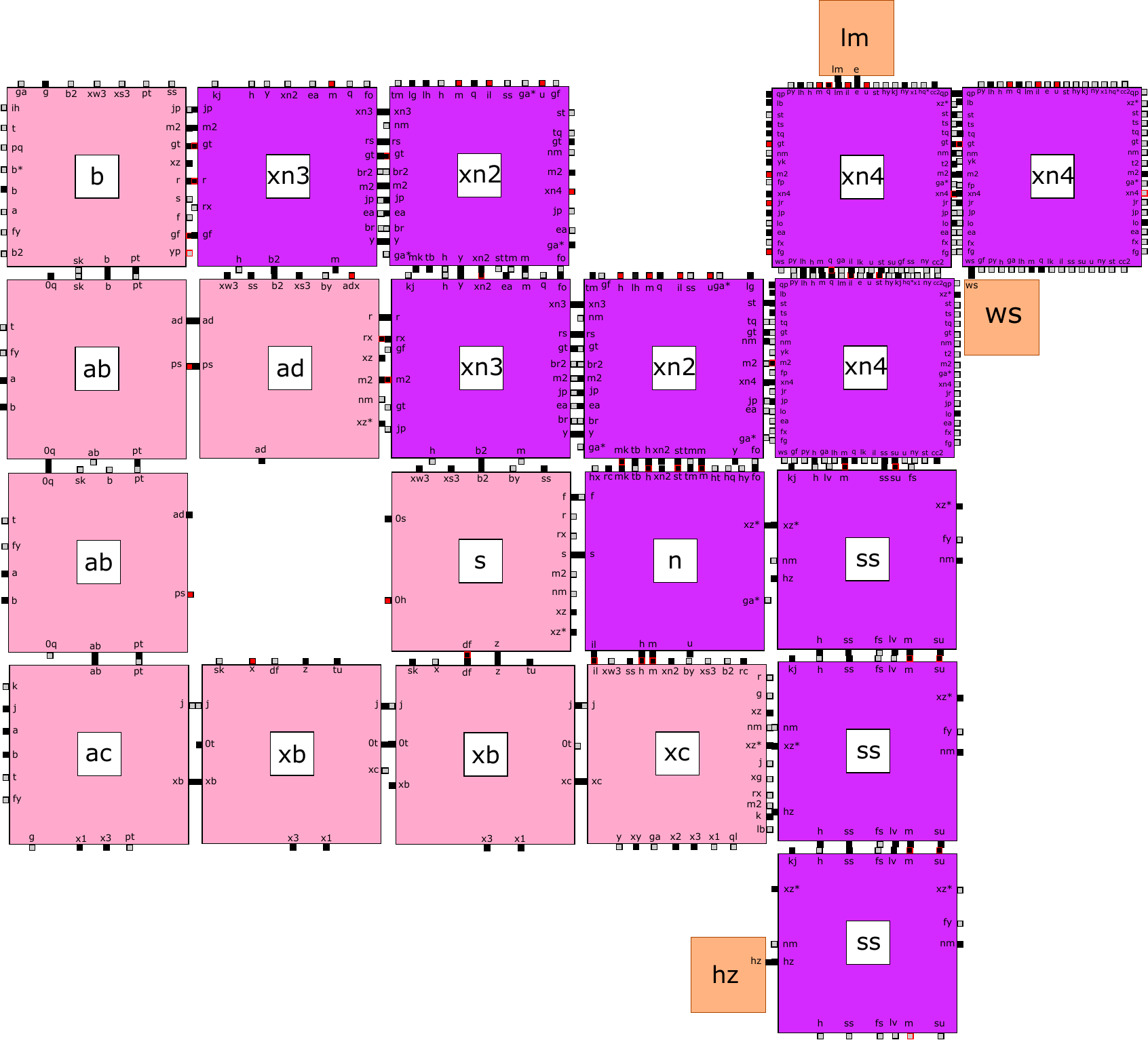}
\caption{The $S_{\Delta{u}}$ subassembly after formation of the upper gap and immediately before combination with $S_{\Delta{s}}$.}
\label{fig:upgap1-5}
\end{figure}

The $S_{\Delta{u}}$ subassembly shown in Figure~\ref{fig:upgap1-5} is ready for combination with the $S_{\Delta{s}}$ subassembly at this point.

\subsection{Combination of Subassemblies}

The first portion of subassembly combination involves the binding of the $S_{\Delta{s}}$ and $S_{\Delta{u}}$ subassemblies.  The southern tooth of the $S_{\Delta{s}}$ subassembly fits into the slot formed by the upper gap; the blocker tile to the right of the upper gap ensures that only the appropriate-sized subassemblies combine.  Appropriate sizing of subassemblies is also guaranteed by only one connecting glue on the $S_{\Delta{s}}$ and $S_{\Delta{u}}$ subassemblies that can bind to one another; this glue is located on the northwestern \textit{b} tile of the $S_{\Delta{u}}$ subassembly and the southwestern \textit{ac} tile of the $S_{\Delta{s}}$ subassembly.  The combination of these two subassemblies results in the detachment of the northern row of $S_{\Delta{u}}$ filler tiles, the south tooth, and the rightmost column of \textit{ss} filler tiles.

\begin{figure}[htp]
\centering
\includegraphics[width=\linewidth]{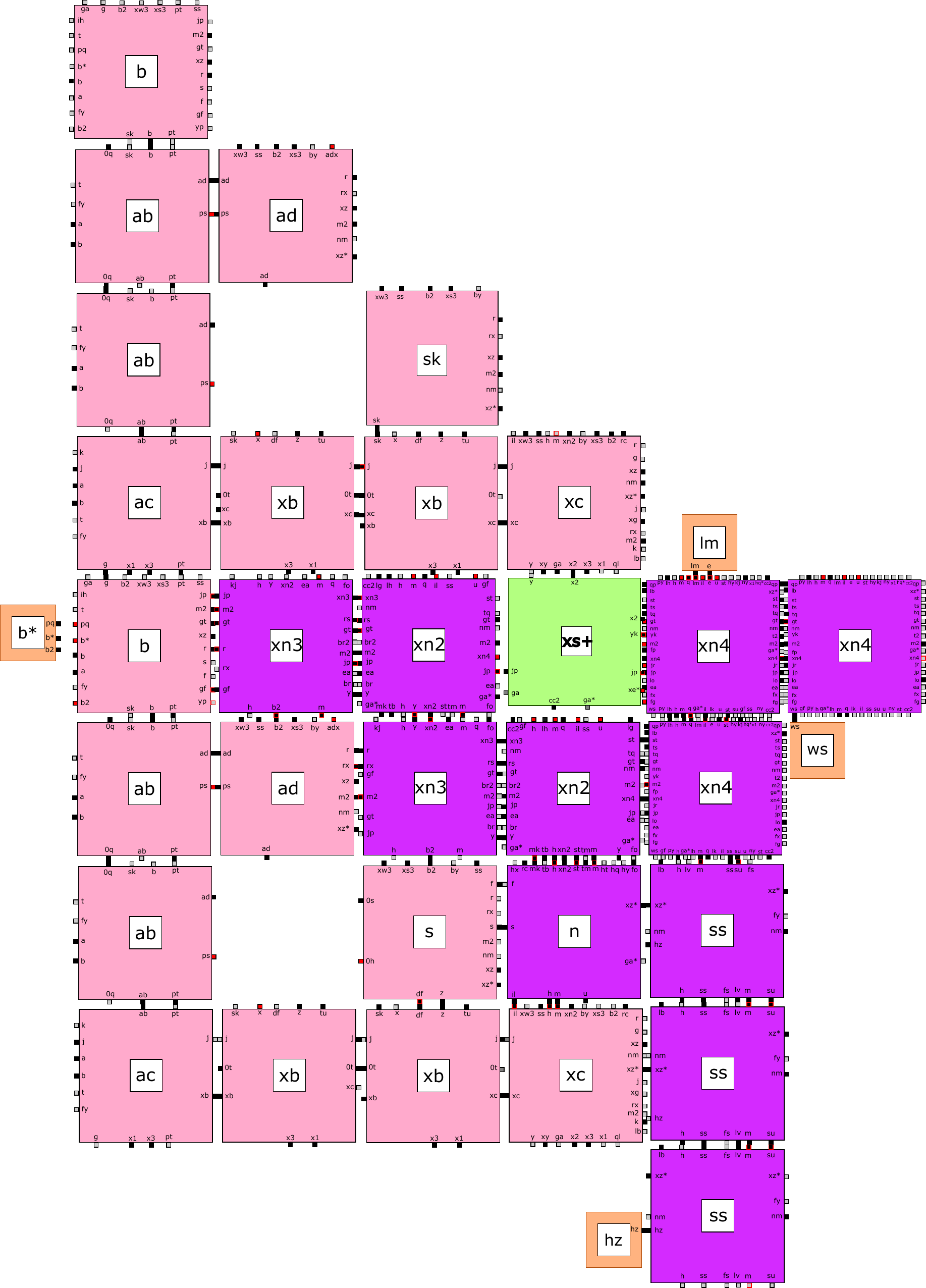}
\caption{Initial combination of $S_{\Delta{s}}$ and $S_{\Delta{u}}$ subassemblies.}
\label{fig:combo1a-1}
\end{figure}

Upon binding the connecting glues of the $S_{\Delta{s}}$ and $S_{\Delta{u}}$ subassemblies, as shown in Figure~\ref{fig:combo1a-1}, the connecting glues activate a variety of signals in the base tiles.  The southwest \textit{ac} tile of $S_{\Delta{s}}$ activates a southern glue that is used if the resulting stage of $S_\Delta$ differentiates into a $S_{\Delta{w}}$ subassembly.  The northwest \textit{b} tile of $S_{\Delta{u}}$ turns \on/ two western glues — one that serves in the possible formation of a $S_{\Delta{w}}$ subassembly and another that binds a helper tile.  This helper tile deactivates western glues on the \textit{b} tile that would result in premature termination of the $S_{\Delta{w}}$ assembly (if this stage of $S_\Delta$ differentiates to a $S_{\Delta{w}}$ subassembly).  Only after this helper tile \textit{b*} has bound does a series of eastern glues travel through the northern row of the $S_{\Delta{u}}$ filler tiles, initiating detachment.  The first tiles to be removed are the \textit{xn3} and \textit{xn2} tiles, which detach as a pair.  At larger stages, the row of \textit{xn4} tiles that precede the south tooth are removed through the help of a blocker tile that ensures all volatile \on/ glues cannot interfere with other assemblies.

\begin{figure}[htp]
\centering
\includegraphics[width=\linewidth]{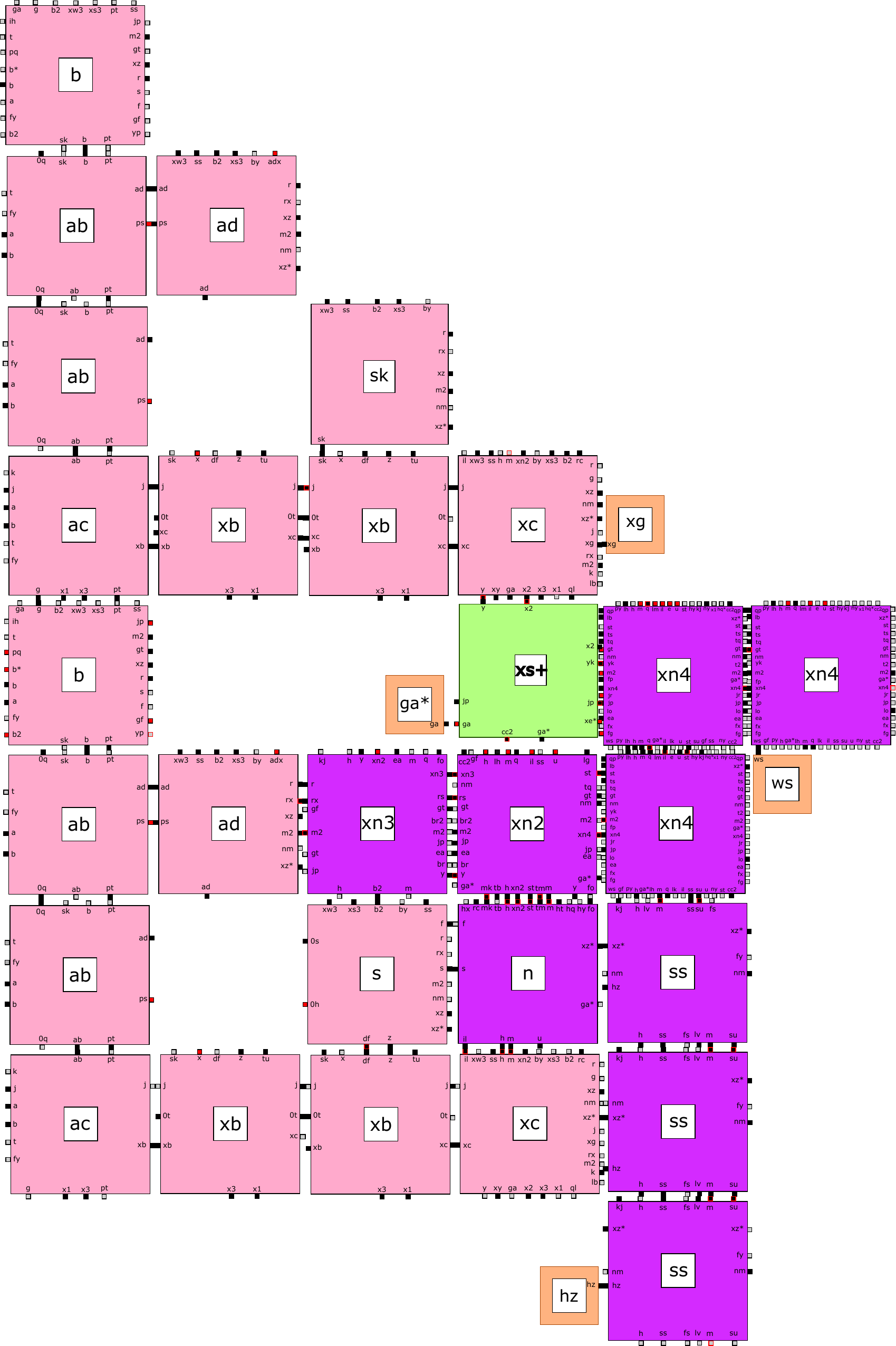}
\caption{Addition of a blocker tile to ensure proper combination of the first stage of subassembly binding and the $S_{\Delta{w}}$ subassembly, as well as preparation for the removal of the south tooth.}
\label{fig:combo1a-2}
\end{figure}

Only after the addition of the blocker tile \textit{xg} to the southeast \textit{xc} corner tile of what was the $S_{\Delta{s}}$ subassembly can the south tooth begin to detach.  Through interactions with southern glues on the \textit{xc} tile (after it binds to \textit{xg}), the south tooth turns \on/ a western glue to which a helper tile attaches.  This helper tile results in the detachment of the \textit{xs+} south tooth.  The $S_{\Delta{u}}$ filler tile to the immediate south of the \textit{xs+} south tooth is also removed as a result of these signals.  This removal is necessary to prevent interference of junk \textit{xn2} tiles with the $S_{\Delta{w}}$ subassembly and with other \textit{xn2} tiles.

\begin{figure}[htp]
\centering
\includegraphics[width=\linewidth]{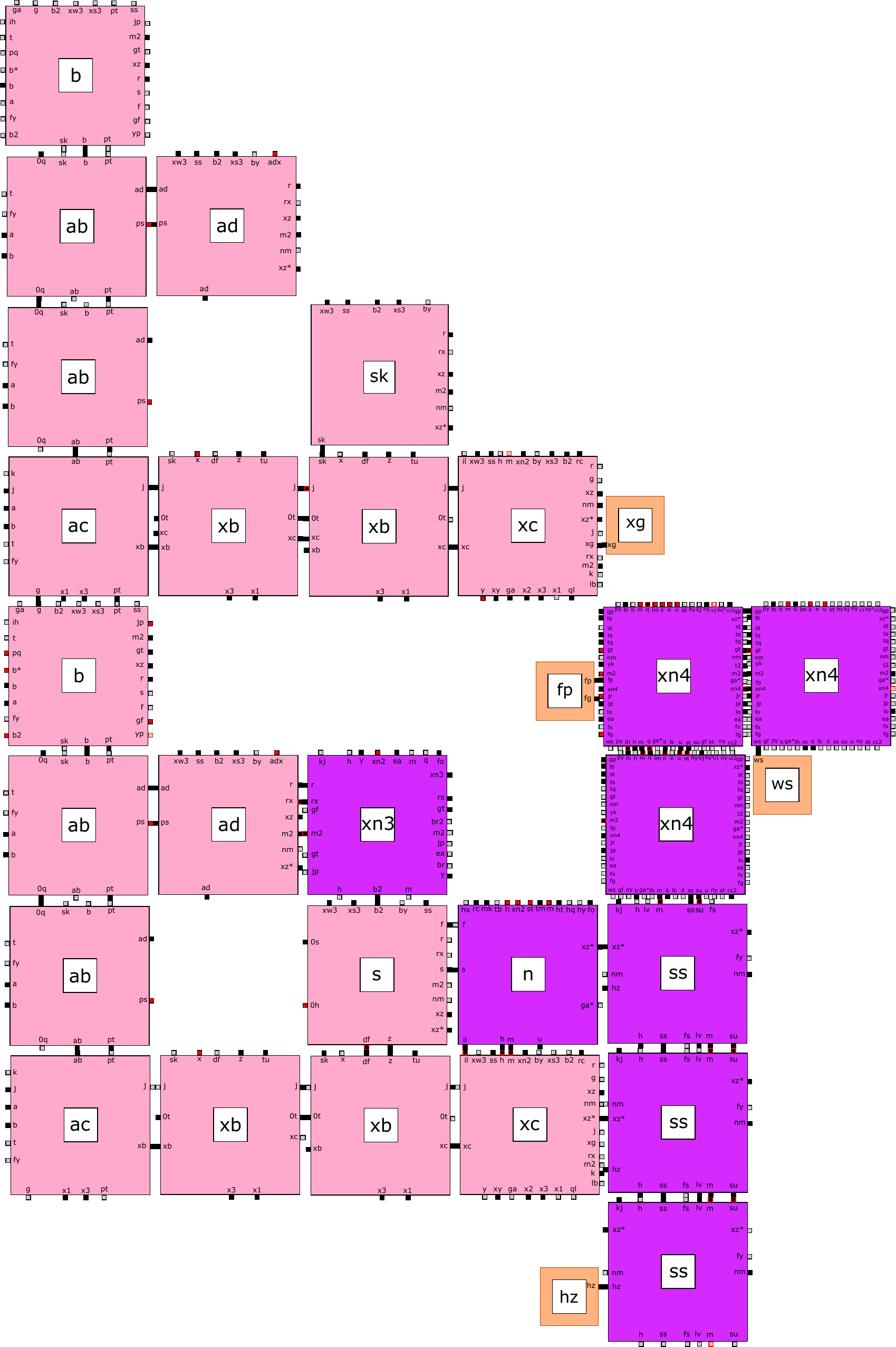}
\caption{Preparation for removal of the final portion of the top row of $S_{\Delta{u}}$.}
\label{fig:combo1a-3}
\end{figure}

The end portion of the top $S_{\Delta{u}}$ filler tile row can only be removed after the binding of a helper tile \textit{fp}, which sends a series of signals resulting in the turning \off/ of glues holding the set of tiles in place.  This removal creates a space for a helper tile to interact with the rightmost column of filler tiles.  The assembly depicted in Figure~\ref{fig:combo1a-3} and Figure~\ref{fig:combo1a-4} depicts a column composed only of \textit{ss} tiles.  However, larger assemblies also feature \textit{xn4} tiles in this rightmost column — detachment of the \textit{xn4} tiles follows the same pattern of detachment seen in \textit{ss} tiles (e.g. involving a helper tile).

\begin{figure}[htp]
\centering
\includegraphics[width=0.8\linewidth]{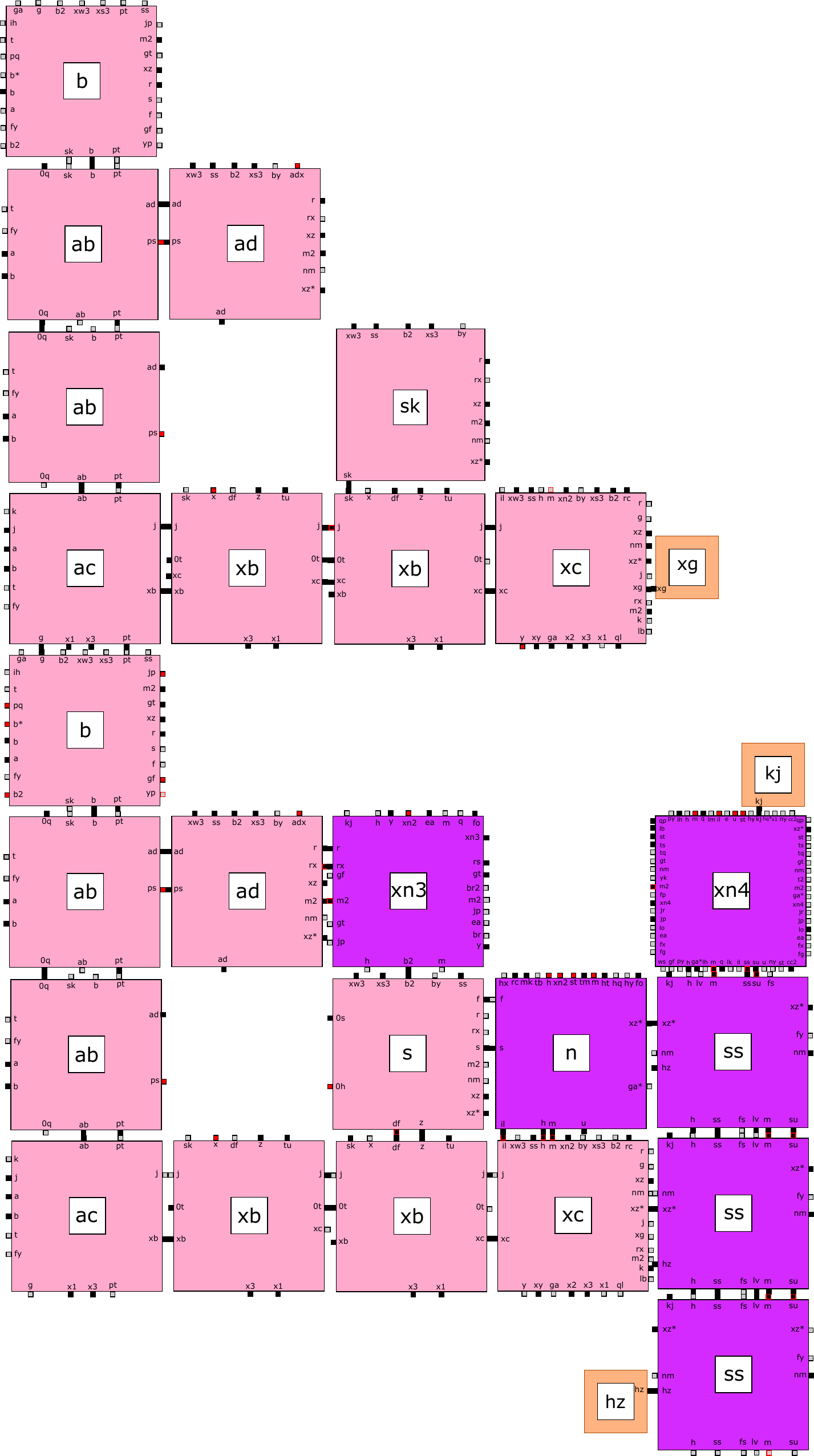}
\caption{Removal of the eastern column of $S_{\Delta{u}}$ filler tiles with helper tiles.}
\label{fig:combo1a-4}
\end{figure}

The $S_{\Delta{u}}$ filler tiles in the eastern row are removed through the assistance of a helper tile; only after the helper tile \textit{kj} has bound can western glues be turned \off/ and the filler tiles removed.  After removal of the column of filler tiles, the binding glue on the southeastern \textit{xc} corner tile is exposed, allowing binding of the $S_{\Delta{w}}$ subassembly.  This can happen before the binding of the helper tile, as depicted in Figure~\ref{fig:combo1a-5}, because after binding of the $S_{\Delta{w}}$ subassembly occurs, nothing else can take place until the helper tile \textit{hq} or \textit{hq*} (depending on if a \textit{n} or \textit{xn4} is presented for binding) attaches.  Addition of the helper tile results in an interface that allows interaction with the west tooth \textit{x2+} tile.

\begin{figure}[htp]
\centering
\includegraphics[width=0.8\linewidth]{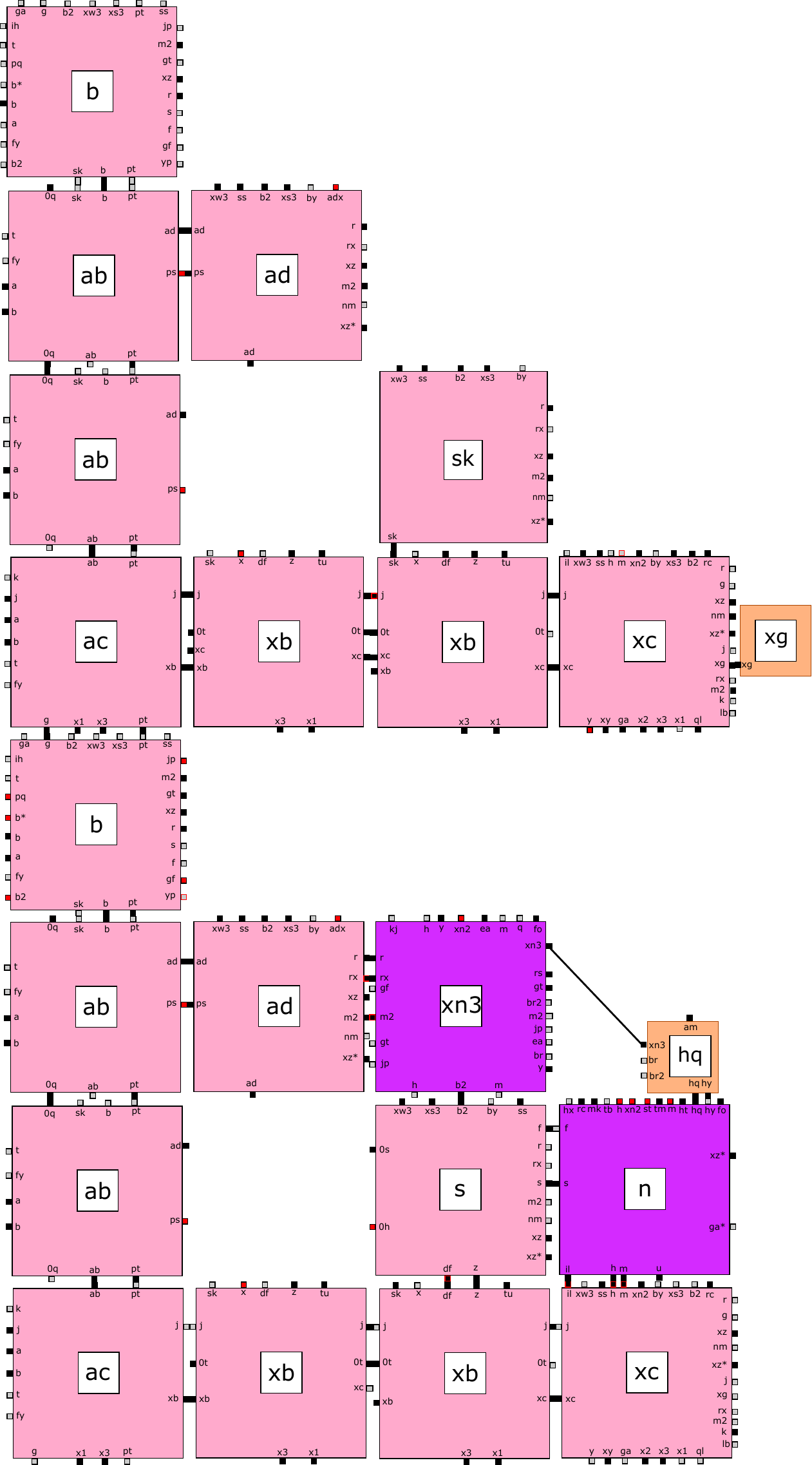}
\caption{$S_{\Delta{u}}$ after attachment of the helper tile \textit{hq}.  It is shown here before before binding of $S_{\Delta{w}}$ subassembly for clarity in depiction.}
\label{fig:combo1a-5}
\end{figure}

The second portion of subassembly combination involves the binding of the $S_{\Delta{w}}$ subassembly and the subassembly formed in the first portion of combination.  The western tooth of the $S_{\Delta{w}}$ subassembly fits into the slot formed by the combination of $S_{\Delta{s}}$ and $S_{\Delta{u}}$. Proper sized assemblies are ensured to bind together in the same manner as the first portion of subassembly combination (i.e. blocker tiles and only one connecting glue being \on/).  This portion of combination results in the removal of all filler and helper tiles and leaves the newly-formed stage of $S_\Delta$ ready for differentiation.

\begin{figure}[htp]
\centering
\includegraphics[width=\linewidth]{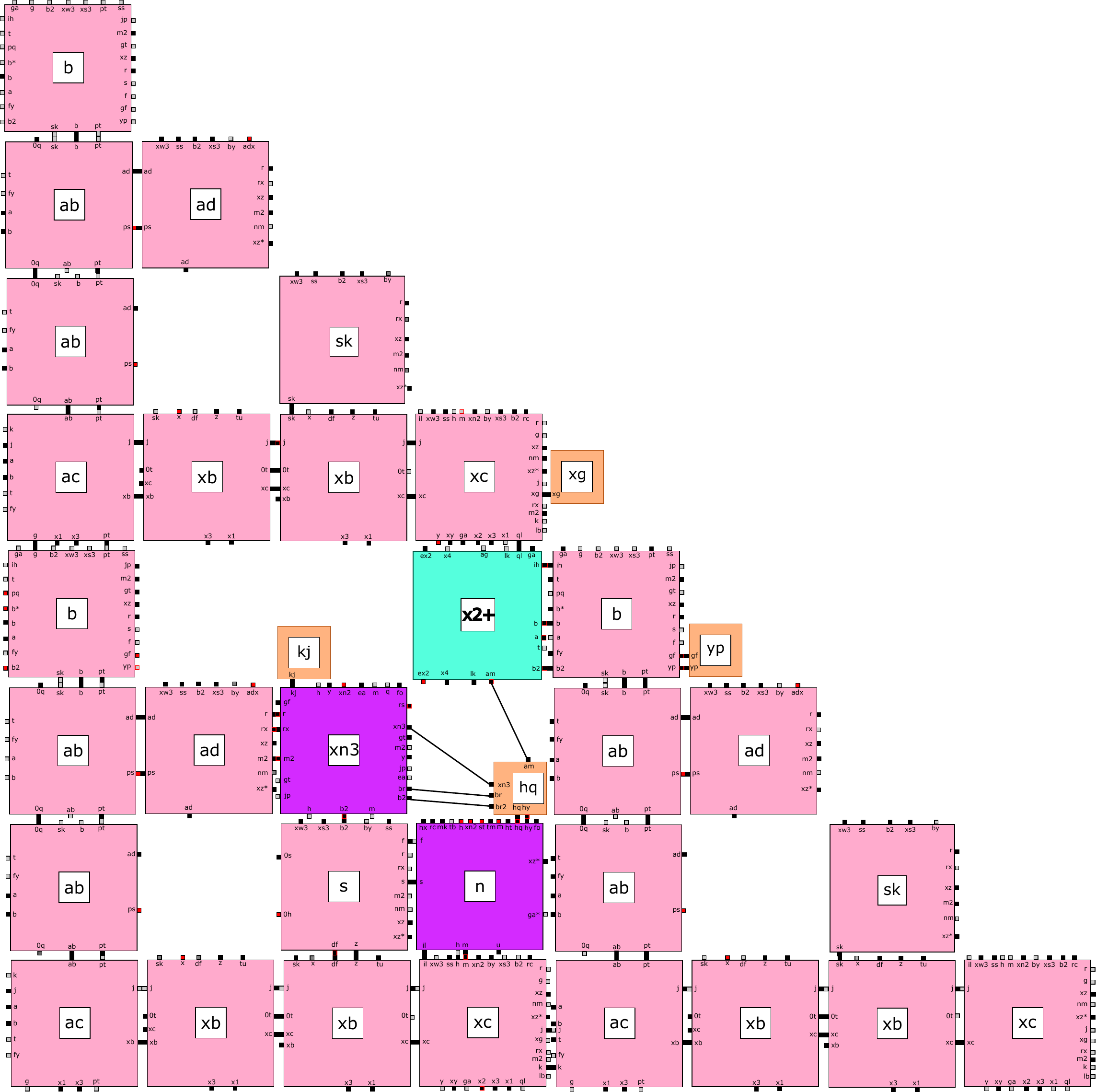}
\caption{Initial combination of the subassembly stages in the final portion of combination and interaction between the west tooth and helper tile.}
\label{fig:combo1b-1}
\end{figure}

After the connecting glues of the $S_{\Delta{w}}$ subassembly and what was the $S_{\Delta{u}}$ subassembly bind, the connecting glues activate a variety of signals in the base tiles.  The southeast \textit{xc} tile of the former $S_{\Delta{u}}$ activates a variety of glues that are used if the resulting stage of $S_\Delta$ differentiates into a $S_{\Delta{w}}$ or $S_{\Delta{s}}$ subassembly.  The southwest \textit{ac} tile of $S_{\Delta{w}}$ turns \on/ a northern glue that carries a signal through the base tiles of the $S_{\Delta{w}}$ subassembly (which can happen because these tiles are never used for the same purpose again, as shown in Figure~\ref{fig:basesignals}).  This signal turns \on/ an eastern glue on the \textit{b} tile, which binds with a helper tile.  The binding of this helper tile \textit{yp} ensures that glues on the northwestern \textit{b} that could potentially interfere with assembly are turned \off/.  Only after the binding of this helper tile can western signals be initiated; these signals result in the detachment of the west tooth, as well as interaction with the helper tile and $S_{\Delta{u}}$ filler tiles.

\begin{figure}[htp]
\centering
\includegraphics[width=\linewidth]{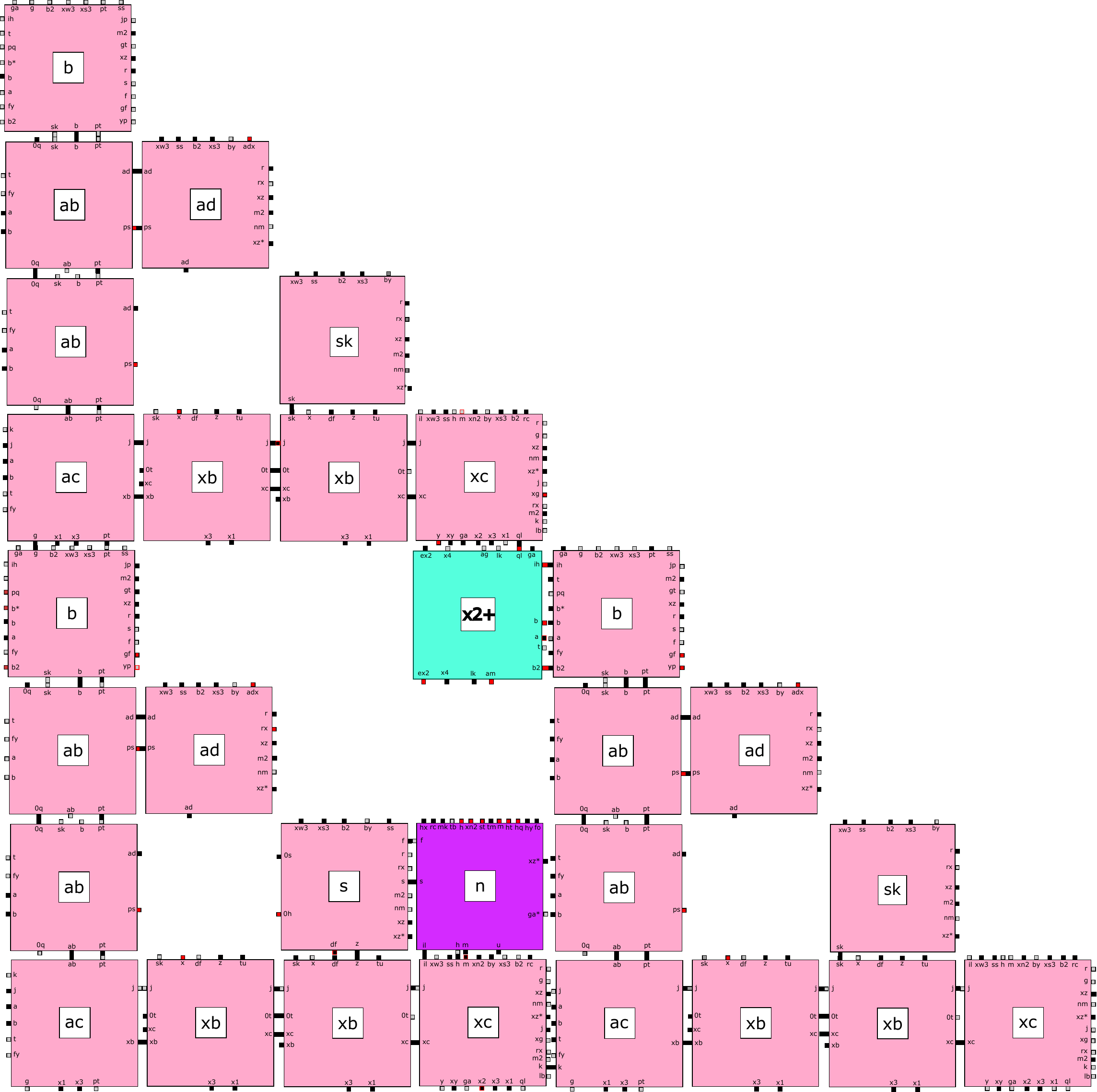}
\caption{The combined assembly after detachment of the final stair-step portion; the west tooth is also capable of detachment at this point.}
\label{fig:combo1b-2}
\end{figure}

The final portion of the stair-step, which includes the helper tile \textit{hp} as shown in Figure~\ref{fig:combo1b-1}, detaches in the same manner as the other portions of the stair-step pattern.  However, the helper tile that results in this detachment does not turn \on/ until \textit{hp} has interacted with the west tooth tile.  This interaction results in the removal of the \textit{x2+} west tooth tile as well.  At larger stages, additional $S_{\Delta{u}}$ filler tiles must be removed; these decay in a manner almost identical to the previous detachment of similar tiles.

\begin{figure}[htp]
\centering
\includegraphics[width=\linewidth]{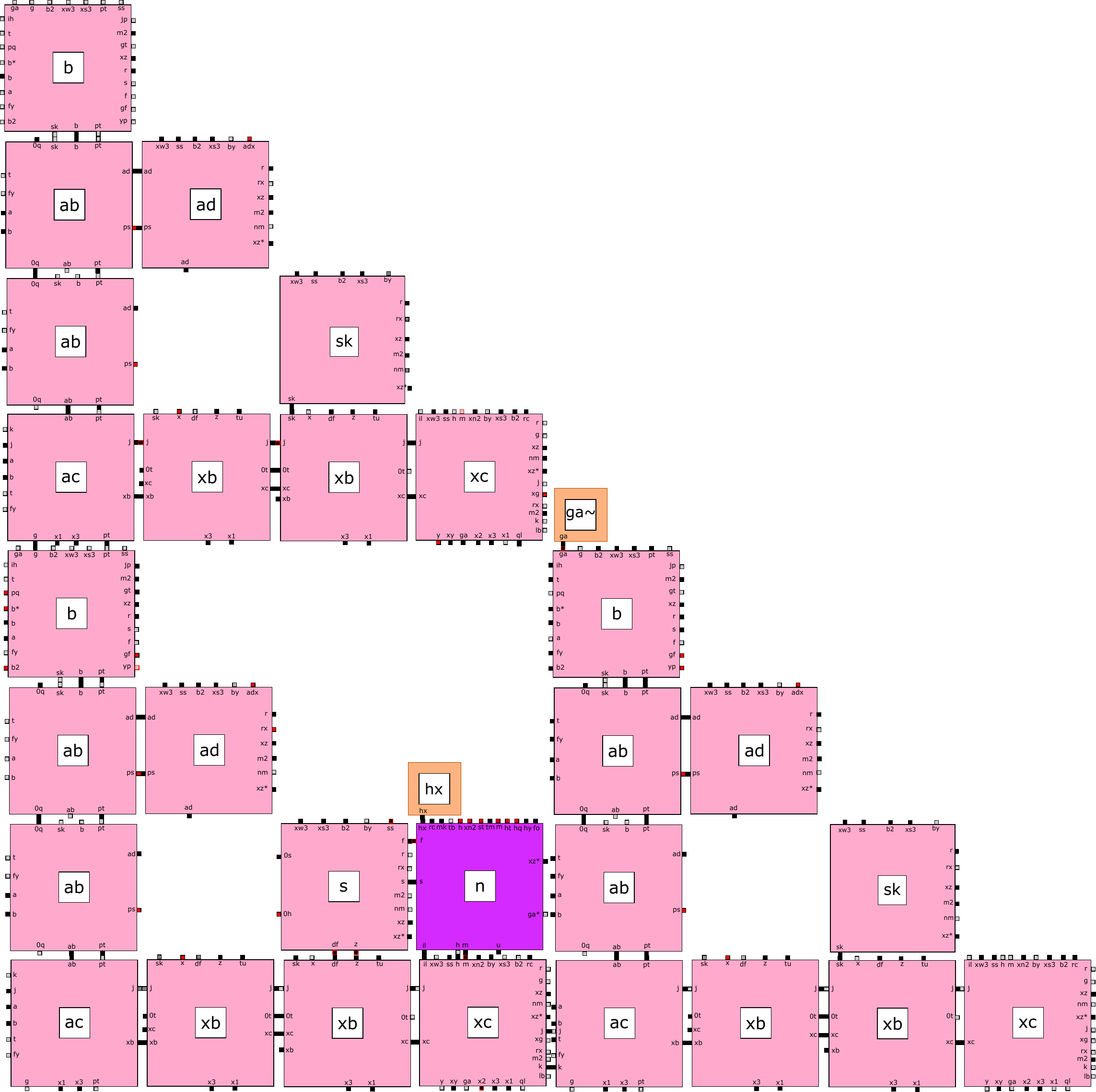}
\caption{Addition of a helper tile to \textit{n} to remove the final portion of the $S_{\Delta{w}}$ filler tiles and activation of the initiator point on \textit{b}.}
\label{fig:combo1b-3}
\end{figure}

The addition of a helper tile to \textit{n} results in the removal of the initiator junk assembly for the $S_{\Delta{u}}$ subassembly.  Removal of the base tile that presents the initiator glue for $S_{\Delta{u}}$ occurs for the same reason as in $S_{\Delta{w}}$ and $S_{\Delta{s}}$, and an \textit{sk} tile will fill in the resulting gap (as shown in Figure~\ref{fig:combo1b-4}).  Activation of the northwestern \textit{b} tile in the former $S_{\Delta{w}}$ subassembly occurs after the binding of a helper tile.  This can happen before the removal of all filler tiles from the interior of the combination of the three $S_\Delta$ subassemblies, and this is able to happen because this interior area will never interfere with differentiation and future growth.

\begin{figure}[htp]
\centering
\includegraphics[width=\linewidth]{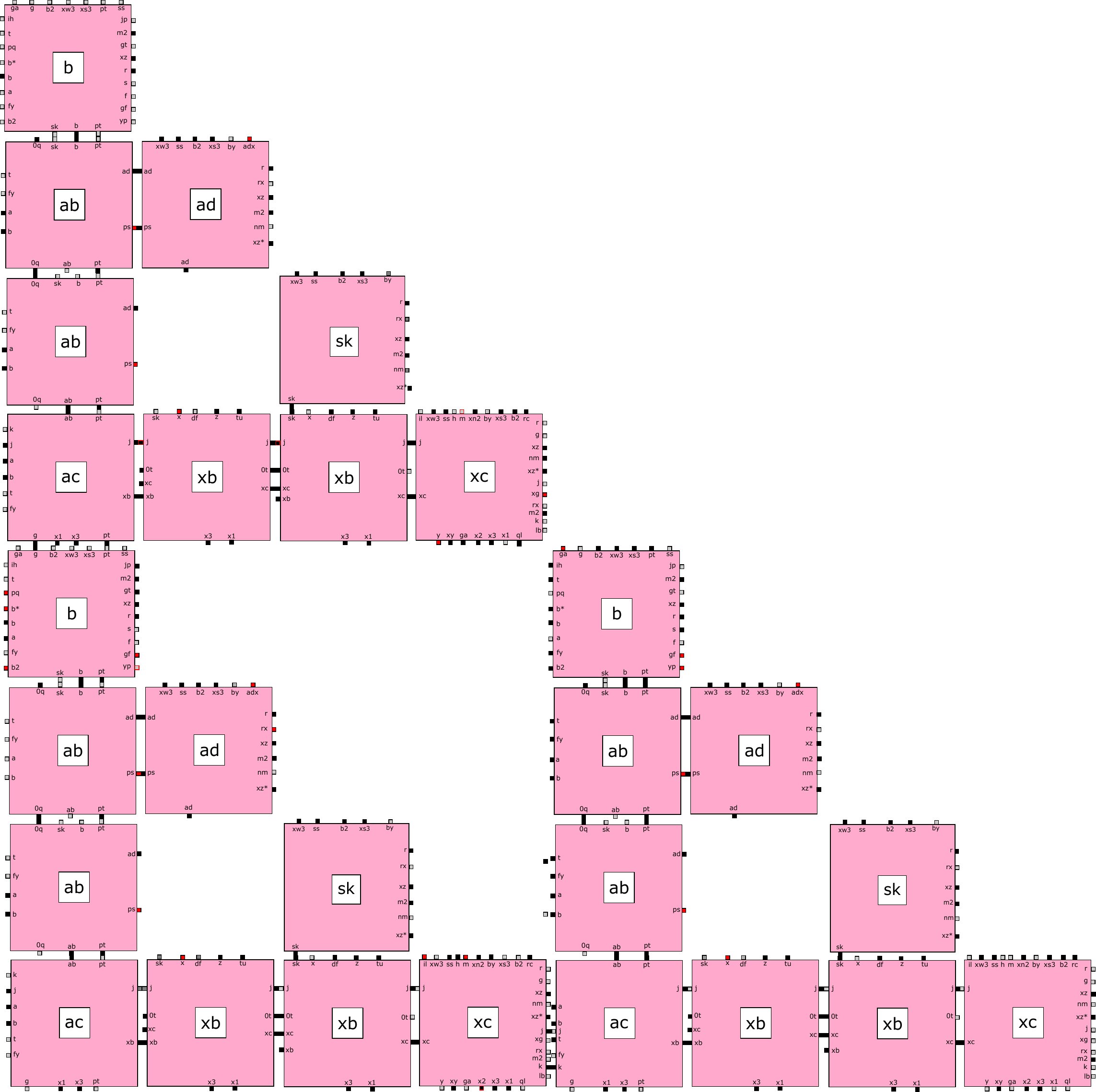}
\caption{Addition of the \textit{sk} tile to fill in the gap left by removal of the $S_{\Delta{u}}$ base initiator tile.  At this point, the $S_\Delta$ is ready for differentiation and further combination.}
\label{fig:combo1b-4}
\end{figure}

After the addition of the \textit{sk} tile and activation of the new base tile initiator point, the resulting $S_\Delta$ stage is complete.

\subsection{Examples of Junk Assemblies}

All junk assemblies that detach during formation of $S_\Delta$ reach a terminal size of $\le 4$ and do not negatively interfere with the assembly of other $S_\Delta$ stages.  Although there are a variety of different junk assemblies, we will discuss the major groups of junk assemblies and the steps taken to ensure that these assemblies cannot interfere with assembly.

One variety of junk assembly is a single tile that detaches from the overall subassembly with a glue still \on/ that could interfere with other assemblies. The \textit{xn4} tile removed to form the upper gap is depicted in Figure~\ref{fig:junk1-1} and falls into this category.  Upon detachment, this \textit{xn4} tile has a western \textit{xn4} glue that could interfere with formation of $S_{\Delta{u}}$.  To prevent this problem, a helper tile can attach to the single junk tile and turn off this interfering glue.  Although the tile can attach to other assemblies before this glue has been turned \off/, any other glues that could interact with active assemblies are turned \off/, resulting in the interfering \textit{xn4} tile eventually being removed with no damage.

\begin{figure}[htp]
\centering
\includegraphics[width=\linewidth]{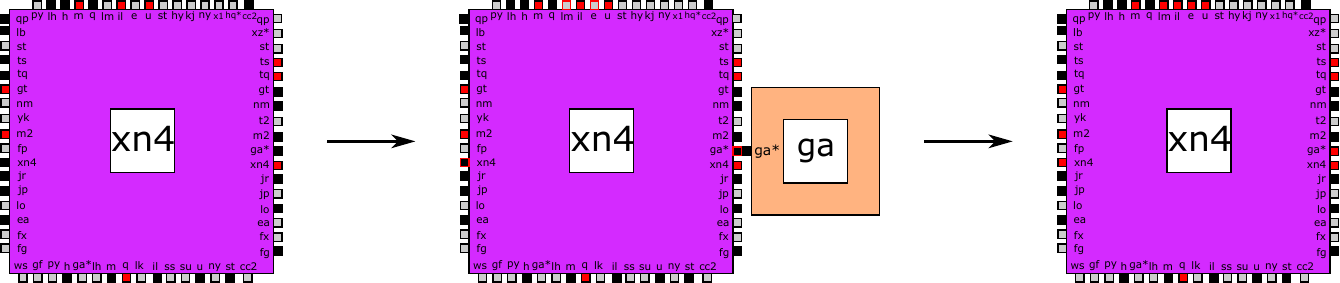}
\caption{An example of a helper tile attaching after a single tile detaches to turn off a glue that could be potentially problematic}
\label{fig:junk1-1}
\end{figure}

To prevent interference between the initiator tile and the initiator point on the base assembly, the two tiles are removed as a single junk assembly with the volatile interface hidden from solution.  This is depicted with the $S_{\Delta{s}}$ initiator junk assembly in Figure~\ref{fig:junk1-2}.  The gap in the overall $S_\Delta$ assembly formed by this removal is filled by a \textit{sk} tile.  All northern glues that are potentially \on/ after detachment of the initiator junk assembly are eventually turned \off/ and do not negatively interfere with the growing assembly.

\begin{figure}[htp]
\centering
\includegraphics[width=\linewidth]{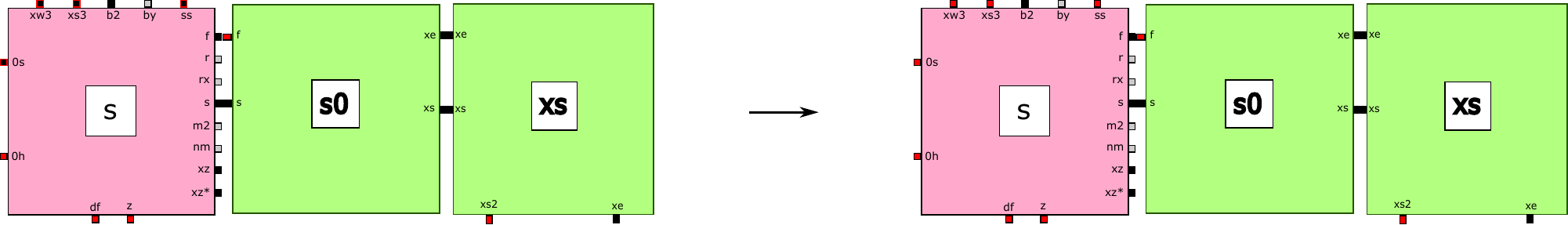}
\caption{An illustration of the initiator and two other tiles detaching in the $S_{\Delta{s}}$ subassembly. This example also applies to $S_{\Delta{w}}$.}
\label{fig:junk1-2}
\end{figure}

Although the initiator tile and initiator point on the base assembly for the $S_{\Delta{u}}$ subassembly are removed in a pair for the same reason as the other two subassemblies, it undergoes a slightly different process due to the $S_{\Delta{u}}$ subassembly's formation.  Because $S_{\Delta{u}}$ grows an upwards stair-step instead of one moving down, like in $S_{\Delta{s}}$, the junk assembly displayed in Figure~\ref{fig:junk1-3} requires a blocker tile before removal to prevent interference by the \textit{n} tile.  Other than these two differences, the initiator junk assembly for $S_{\Delta{u}}$ follows the same pattern as the other two initiator junk assemblies.

\begin{figure}[htp]
\centering
\includegraphics[width=\linewidth]{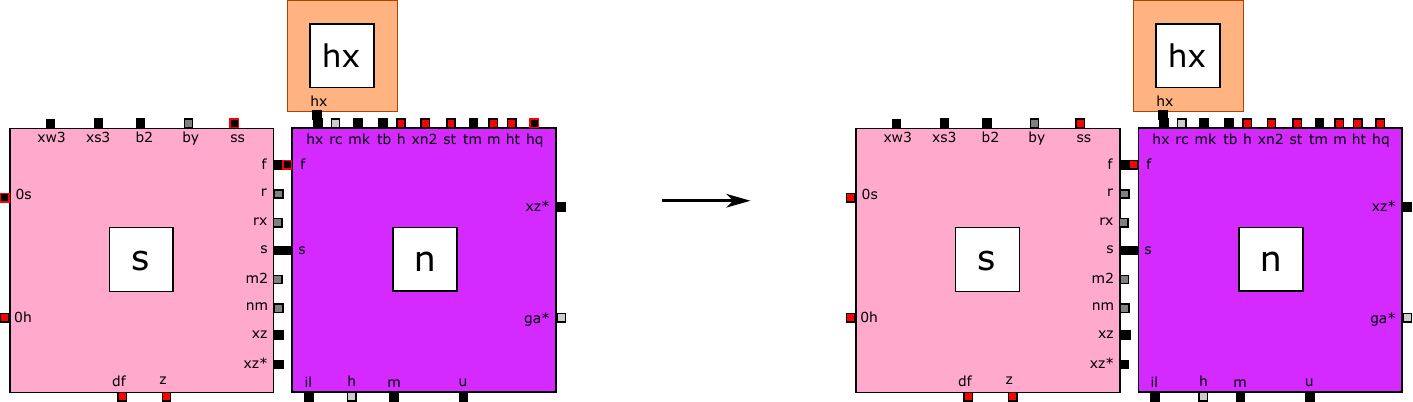}
\caption{The initiator tile detachment for the $S_{\Delta{u}}$ subassembly.}
\label{fig:junk1-3}
\end{figure}

Due to potentially problematic glues that cannot otherwise be guaranteed to turn \off/ or not interfere with assembly, many tiles are removed in groups of two or three with the addition of a blocker tile.  This is depicted in Figure~\ref{fig:junk1-4} for a set of tiles from the $S_{\Delta{w}}$ subassembly. Because these junk assemblies cannot detach from the overall assembly until a blocker tile is in place, all volatile glues are prevented from interfering with parallel and future assemblies.

\begin{figure}[htp]
\centering
\includegraphics[width=\linewidth]{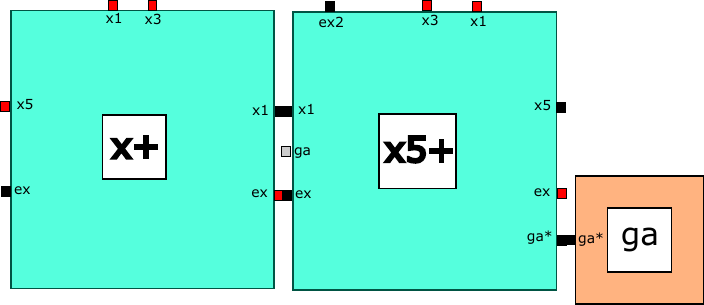}
\caption{An example of a group of tiles removed by helper tiles.}
\label{fig:junk1-4}
\end{figure}

The same principle is used for single tiles as well, as shown in Figure~\ref{fig:junk1-5}.  In these scenarios, individual junk tiles cannot detach until a blocker tile covers a potentially problematic glue.  Any other glues that could bind to separate assemblies are ensured to not cause problematic growth and are eventually turned \off/ by a helper tile that binds to the junk tile.

\begin{figure}[htp]
\centering
\includegraphics[width=\linewidth]{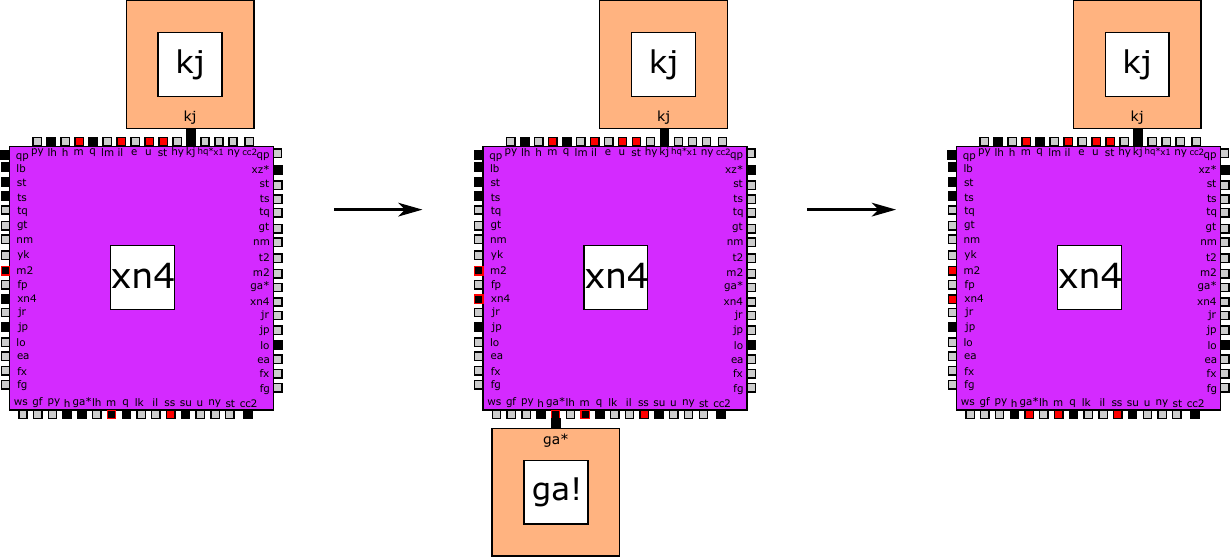}
\caption{An example of a single tile being removed by helper tiles and then acted upon by additional helper tiles to prevent interference.}
\label{fig:junk1-5}
\end{figure}

In this way, the discrete Sierpinski triangle is assembled at temperature 1 and all produced junk assemblies are of size $\le 4$.

\fi
\def\calTF{{\mathcal{T}_{\mathcal{F}}}}
\def\gpreinit/{{\texttt{preinit}}}
\def\ginit/{{\texttt{init}}}
\def\gend/{{\texttt{end}}}
\def\gpreconnect/{{\texttt{preconnect}}}
\def\gconnect/{{\texttt{connect}}}
\section{Self-Assembly of Arbitrary Discrete Self-Similar Fractals} \label{sec:app-gen}

In this section, we state and prove our main theorems, which are generalizations of the techniques of Theorem~\ref{thm:triangle}, related to the strict self-assembly of arbitrary connected discrete self-similar fractals using temperature 2 STAM systems, and an infinite set of such fractals using temperature 1 STAM systems.

\begin{theorem}\label{thm:fractals-temp2}
For any connected discrete self-similar fractal $\calF$, there exists an STAM system $\calTF = (T,2)$ such that $\calTF$ has exactly one infinite terminal supertile $\alpha$, and $\dom(\alpha) = \calF$, i.e. is exactly the discrete self-similar fractal $\calF$, and for all $\gamma \in \termasm{T}$ such that $\gamma \not= \alpha, |\dom(\gamma)| \le 2$.
\end{theorem}

\begin{theorem}\label{thm:fractals-temp1}
For any connected discrete self-similar fractal $\calF$ which is singly-concave, there exists an STAM system $\calTF = (T,1)$ such that $\calTF$ has exactly one infinite terminal supertile $\alpha$, and $\dom(\alpha) = \calF$, i.e. is exactly the discrete self-similar fractal $\calF$, and for all $\gamma \in \termasm{T}$ such that $\gamma \not= \alpha, |\dom(\gamma)| \le 2$.
\end{theorem}

\begin{figure*}[tp]
\centering

  \subfloat[][]{%
        \label{fig:initiators-triangle}%
        \includegraphics[width=1.2in]{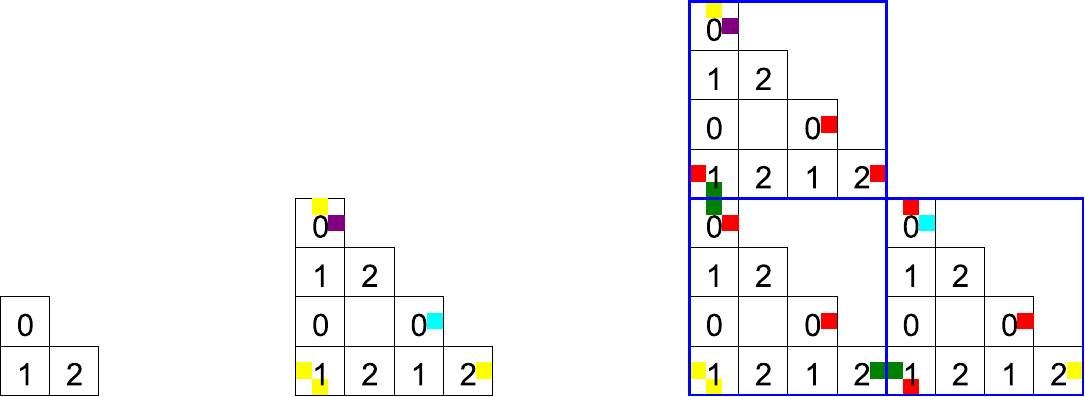}
        }\quad\quad
  \subfloat[][]{%
        \label{fig:initiators1}%
        \includegraphics[width=2.5in]{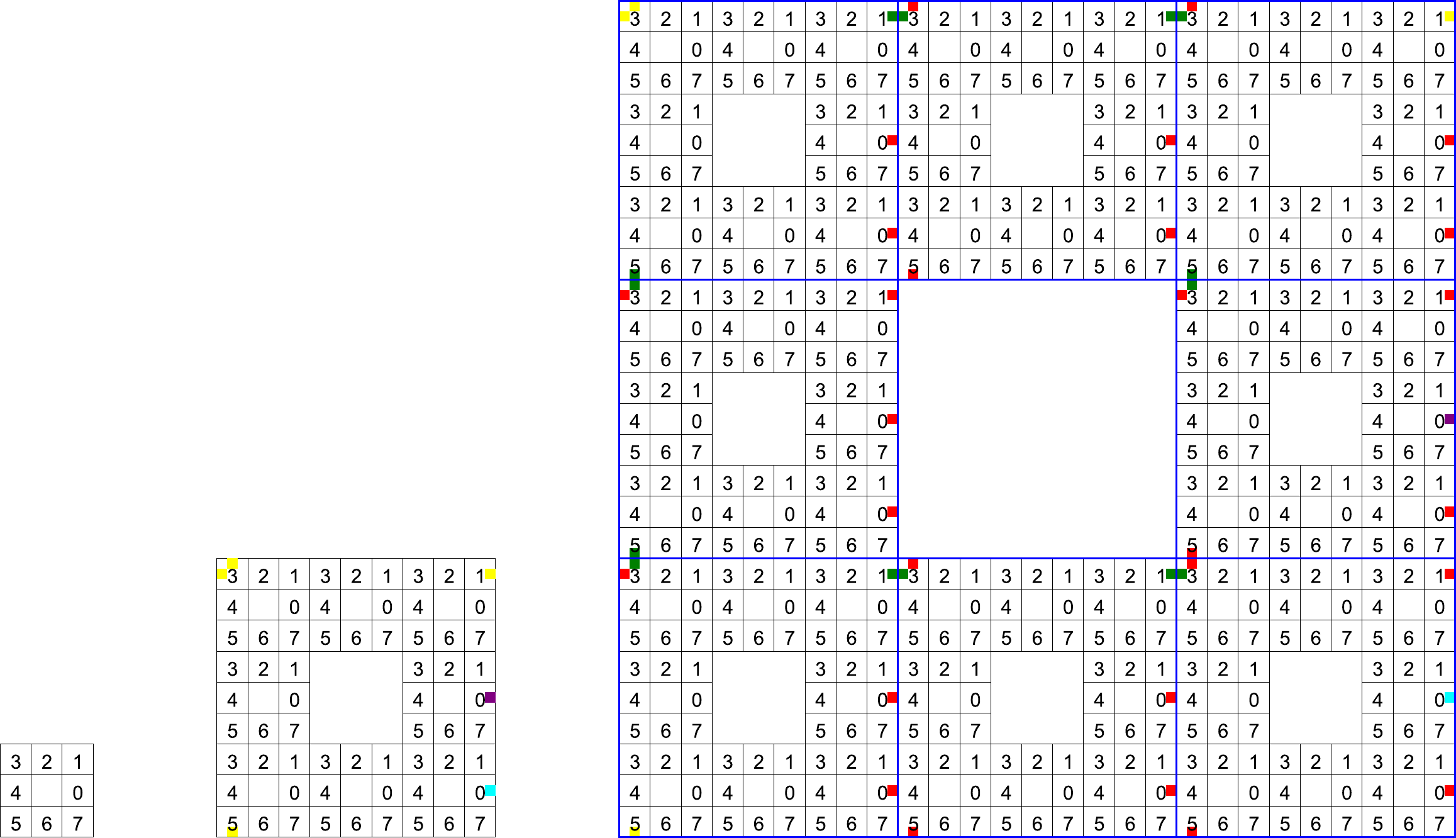}
        }\quad\quad
  \subfloat[][]{%
        \label{fig:initiators2}%
        \includegraphics[width=2.5in]{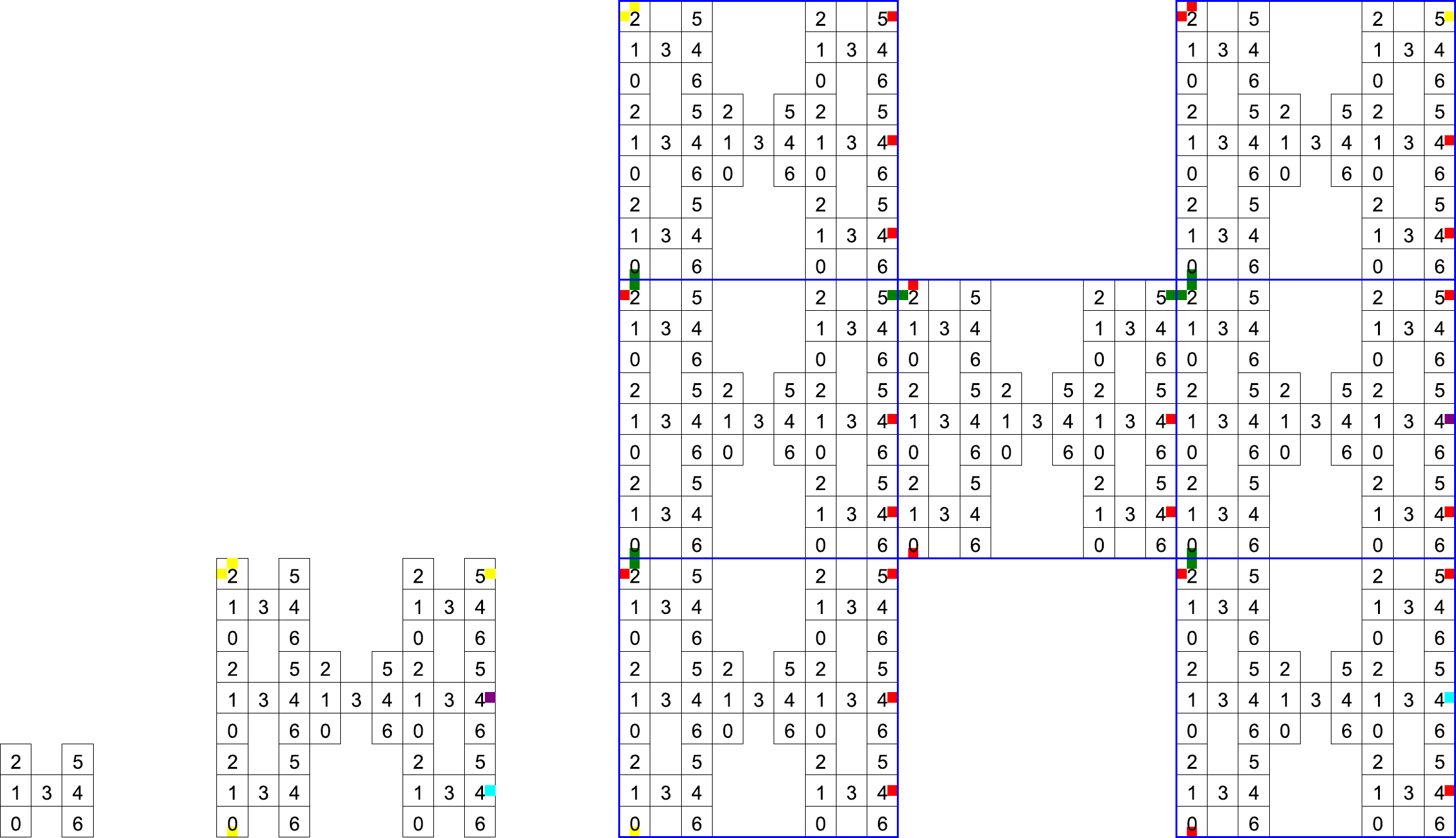}
        }

  \caption{Three fractals, each with its generator on the left, stage 2 in the middle, and stage 3 on the right. Initiator locations with \ginit/ glues active are colored aqua, and locations with \gpreinit/ glues active in purple.
  }
  \label{fig:initiators}
\end{figure*}

\begin{proof}

The proofs of both Theorem~\ref{thm:fractals-temp2} and Theorem~\ref{thm:fractals-temp1} are nearly identical, with only a very small difference, so we present a single proof and will mention that single point of divergence within it.  Our proof is by construction of the necessary STAM systems given fractal definitions as input.  The following provides an overview of our construction, and we note that a very large portion of the basic tiles and gadgets needed to carry out this construction are based on those used in the construction for the Sierpinski triangle, often requiring differences in just scaling, rotation, or number of signals.  Therefore, we present most of the tiles and signals at a high enough level to convey the design and rely on subsets of the specific designs of tiles and signals from the previous construction, only presenting specific tiles and signals in locations which are significantly unique.

Let $\calF$ be a discrete self-similar fractal, and let $G$ be the generator for $\calF$.  We now present a description of how to algorithmically generate a tile set $T$ such that the 2HAM system $\calTF = (T,1)$ creates a unique infinite terminal supertile of shape $\calF$, along with all other terminal assemblies being ``junk'' assemblies whose sizes are all $\le 2$.

Our construction will start growth of $\calF$ at stage 2, i.e. $\calF_2$, with tiles from $T$ directly combining to form $\calF_2$ in such a way that every point of $\calF_2$ receives a unique tile type, which is done by creating a tile type for each position with a glue unique to each pair of adjacent sides between two such tiles.  This means there are unique glues to each pair of interior edges, and to create the glues for the exterior edges we will perform an analysis of the generator as follows.  At this point we also mention that, since we are discussing the constructions of both Theorem~\ref{thm:fractals-temp1} and Theorem~\ref{thm:fractals-temp2} in parallel, all signals, tiles, and glues throughout both constructions are identical except for the fact that in each, all glues are $\tau$-strength, where $\tau=2$ for the construction of Theorem~\ref{thm:fractals-temp2} and $\tau=1$ for the construction of Theorem~\ref{thm:fractals-temp1}, for every case except for a very few glues which will be specifically mentioned later.

\begin{figure}[htp]
\begin{center}
\includegraphics[width=3.2in]{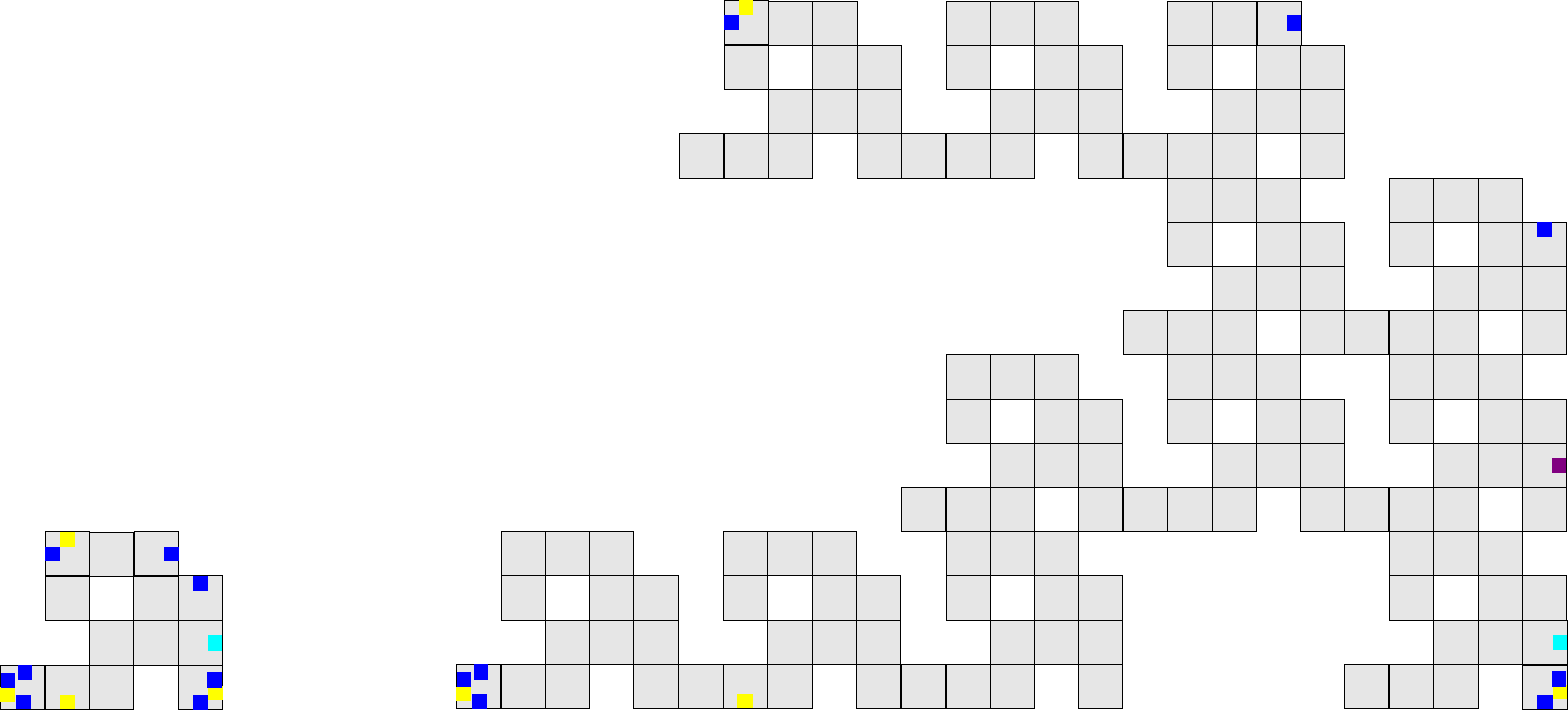}
\caption{An example of a generator and stage 2 with the locations of \gend/ (blue), \gpreconnect/ (yellow), \ginit/ (aqua), and \gpreinit/ (purple) glues shown. Note that those positions are first marked in stage 2 and not in the generator, but they are shown here with markings to demonstrate their locations relative to the generator.}
\label{fig:generator-example1-connectors-append}
\end{center}
\end{figure}

First, we will locate a perimeter location on $\calF_2$ to serve as the \emph{initiator} location where tiles can attach to begin growth of stage 3. Let $\vec{\alpha} \in \calF_2$ be the rightmost point in the second row (from the bottom) of the rightmost copy of $\calF_1$ in the bottom row of copies of $\calF_1$ in $\calF_2$. (See Figure~\ref{fig:initiators} for a few examples marked as aqua tiles in each stage 2.)  The tile which attaches in position $\vec{\alpha}$ has an east facing glue of type \ginit/ (which will allow for the attachment of an ``initiator'' tile).  Additionally, let $\vec{\beta} \in \calF_2$ be the rightmost point in the second row of the rightmost copy of $\calF_1$ in the second row (from the bottom) of copies of $\calF_1$ in $\calF_2$. The tile which attaches in position $\vec{\beta}$ has an east facing glue of type \gpreinit/.  Intuitively, the tile at this position will have the ability to transition into a state with an \ginit/ glue active if it happens to eventually be located in the unique position of a fully completed stage allowing it to begin the transformation of that assembly into one of the substages of the next stage.  (Note that the initiator location must be in a location which is maximal for the side it is on, in this case the east side.  Therefore, if the location that was selected above does not happen to be on an easternmost tile, we will instead put it on either location $(1,1)$ or $(0,1)$, one of which must be in the generator of a connected discrete self-similar fractal and thus also $\calF_2$.  Then, all directions of the following construction can be either rotated 90 or 180 degrees clockwise, respectively, and references to the location of the initiator point can be updated to the new location.  This results in conceptually the same construction so, without loss of generality, we will continue or description assuming the originally stated location and assuming it is on an easternmost location of $\calF_2$.)

We mark the extreme points of each side with special exterior glues.  In the westernmost column of tiles which form $\calF_2$, the topmost has a northern glue of type \gend/ which is specific to that corner, and similarly so does the southern glue of the bottommost tile.  Analogously, such ``end marker'' glues are located on all 4 sides.  See Figure~\ref{fig:generator-example1-connectors-append} for an example (specifically, the blue glues).
We also place special \gpreconnect/ glues at the four locations which could possibly be used to connect one copy of a stage to another, depending on which substage locations they represent within the next larger stage.  These are located on the north and south sides of the leftmost location which both the top and bottom row, respectively, have in common.  There is guaranteed to be at least one such location since the fractal is connected, and therefore at least one position aligns between two copies which are placed with one directly above the other.  Similarly, the horizontal connection points are marked.  See Figure~\ref{fig:generator-example1-connectors-append} for an example (specifically, the yellow glues).  Note that on there will actually be a set of different types of \gpreconnect/ glues in each such position, and the determination of those types will be discussed later.  The final glues that we add to the perimeter of the tiles of $\calF_2$ are glues which are specific to the direction that each of the external edges of any of those tiles face, i.e. there are \emph{directional} glues labeled $N$,$E$,$S$, and $W$ and $N'$,$E'$,$S'$, and $W'$ which are on each of the corresponding sides of tiles which are on the exterior of $\calF_2$.  Those glues are \on/ on any tile sides which do not have any of the previously mentioned special glues such as \gpreinit/, \gpreconnect/, etc., active.  For sides with one of those special glues \on/, the directional glues begin in the \latent/ state, and when the special glues are all turned \off/ on such a side, the directional glues of that side is then switched \on/.  Additionally, these are some of the few glues which differ between the constructions of Theorem~\ref{thm:fractals-temp1} and Theorem~\ref{thm:fractals-temp2}, and instead of being $\tau$-strength for each, the directional glues are strength-1 in both constructions.

We now define a spanning tree $T_G$ through the graph of $G$ such that the root is at the rightmost location of the bottom row of $G$ (i.e. $(x_{max},0)$ for \\ $x_{max} = \max(\{x | (x,y) \in G \}$) and which is constructed as follows.  Beginning from point $(x_{max},0)$ and assuming an orientation facing to the west, add the second node by adding either the node to its west if it exists in the generator, and if not, turn right (facing north) and add that point.  Continue by following a depth-first-search to build a spanning tree which always attempts to take the leftmost available option when there are more than one adjacent points which have not already been added to the tree.  See Figure~\ref{fig:generator-trees} for examples.

\begin{figure*}[htp]
\begin{center}
\includegraphics[width=\linewidth]{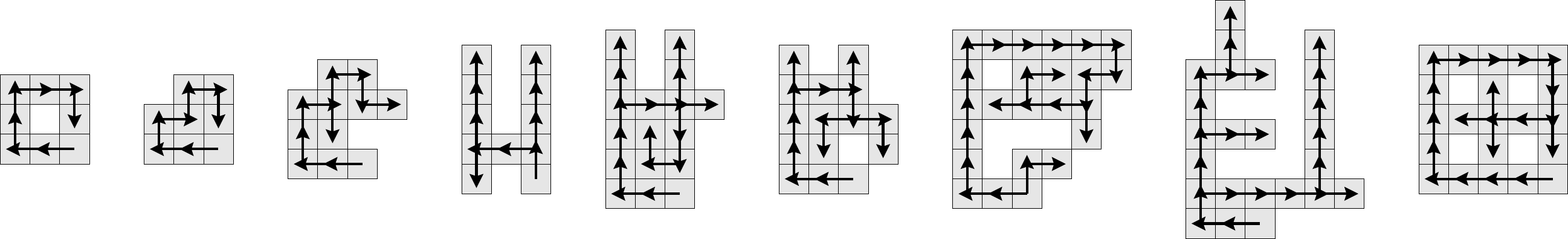}
\caption{Examples of generators and spanning trees through them.}
\label{fig:generator-trees}
\end{center}
\end{figure*}

Using the spanning tree $T_G$, we will now define an ordering by which the substages of stage $i$, for $i > 2$, combine to form that stage. Simply, we start from the furthest leaf from the root in $T_G$ (and if there are multiple leafs at the same distance, randomly select one).  The substage representing that node will first connect to one representing the location which comes immediately before it on the path backward to the root. Subassemblies representing each leaf location will combine to those preceding them on their branches, with each branch assembling in parallel. Whenever there is a ``base'' node from which more than one branch goes outward, an ordering by which the branch assemblies attach to the assembly representing that base node is set as follows.  From the side $d$ of that node which attaches to the node along the path back to the root, i.e. the one which is closer to the root in $T_G$, traverse the sides in counterclockwise (CCW) order and set the ordering of attachment to be the order in which the branches are encountered, with side $d$ being the final side in the ordering. After this attachment, growth continues toward the root in the same way until reaching the root.  Examples of orderings can be seen in Figure~\ref{fig:generator-trees} by following the paths in reverse.

Beginning from the initial formation of stage 2, our construction will guarantee that only once a stage is complete, i.e. the assembly representing it has a tile in every position which is defined for that stage of the fractal (assuming the assembly is translated to $(0,0)$), can the necessary glues activate and structures form which allow it to combine to other assemblies representing exactly the same stage.  For stage 2 assemblies, this is enforced by a ``circuit'' of signals which are encoded in the tiles so that as tiles bind to form a stage 2 assembly, a ``completion'' signal propagates so that if and only if every tile of stage 2 has attached does that signal arrive at the single tile which has the \ginit/ glue on its east side.  At that point, that glue is turned from \latent/ to \on/.  This guarantees that only completed stage 2 assemblies can begin the process of combining with other assemblies, and starting from that base case we will prove inductively by our construction that all subsequent stages must correctly complete before allowing additional combinations.

Therefore, assume that only assemblies representing complete stages of $\calF$ have \ginit/ glues which are \on/, and that each such assembly has exactly one \ginit/ glue active on its perimeter.  (Note that it may be the case that there are some extra tiles bound to such an assembly for stages $> 2$, but any such tiles not included in the domain of the matching stage of $\calF$ will be guaranteed to fall off eventually.  Although the exact timing of their detachments cannot be guaranteed due to the asynchronous nature of the STAM, no such lingering tiles can cause incorrect growth or attachments, but can only temporarily delay correct forward growth of the ultimately infinite structure.)

Let $\alpha_i$ be a producible supertile in $\calTF$ which represents stage $i$ of $\calF$ and which has a single \ginit/ glue active on its perimeter.  In $T$, there are a subset of tile types called \emph{initiators}, and there are exactly $|G|$ of them, i.e. one for every point in the generator of $\calF$.  Nondeterministically, a tile of any initiator type can bind to $\alpha_i$.  Assuming that the initiator which corresponds to $\vec{v} \in G$ binds to the \ginit/ glue of $\alpha_i$, and that in the tree $T_G$ the node at location $\vec{v}$ is connected to nodes in the set $\vec{V}$, that will begin the following series of tile attachments and detachments which ultimately result in $\alpha_i$ binding to supertiles representing each of the points in $\vec{V}$.  We call this the process of \emph{differentiation} since an identical set of supertiles nondeterministically (by the type of initiator which attaches) each transform into supertiles able to represent different substages of the next stage.

\subsubsection{Substage differentiation}

For each initiator tile type, there is a unique set of tile types which bind to it to grow paths and portions of assemblies necessary for the particular process of differentiation needed for the generator point that the supertile will eventually represent.  The actions needed for differentiation to any particular type are completely determined by the shape of the generator, and independent of the stage of $\calF$ being formed.

As previously mentioned, as soon as a supertile completes growth into a stage of $\calF$, it has exactly one \ginit/ glue active on its perimeter.  Another invariant which is maintained for all such stage-complete supertiles is that they also have exactly the following numbers of each of the following types of glues active on their perimeters:  one \gpreinit/ glue (which is on the east side), one \gpreconnect/ glue on each side, and one \gend/ glue on each of the two extreme points in each direction (e.g. the north sides of the northernmost of the westernmost tiles and the northernmost of the easternmost tiles are marked as the extreme left and right points, respectively, for the north side).  Since every such location in a stage 2 supertile has the corresponding glue \on/, as those supertiles combine to form larger stages, in order for the proper number of each type of glue to be active, appropriate subsets of those glues must be deactivated (and sometimes others, such as \gconnect/ activated).  We now discuss how it is determined which such glues of each differentiating substage need to be activated or deactivated.

\begin{enumerate}
    \item \gend/ glues:  A key observation to the management of the end marker glues is that, for instance, the tile with the left \gend/ marker of the north side of stage $i+1$ is the northernmost of the westernmost tiles of stage $i+1$, which is exactly the northernmost of the westernmost of the tiles that make up the substage $i$ supertile which attaches in the northernmost of the westernmost positions of $i+1$. Thus, for every \gend/ marker glue, there is exactly one substage which its supertile can differentiate into for the next stage in which that glue should remain \on/.  Therefore, during differentiation, every \gend/ marker glue is forced to turn off except for those which will also mark the extreme points of the next stage.
    \item \gpreconnect/ glues:  There is a \gpreconnect/ glue active on the side of each direction.  By noting which sides of the location $\vec{v}$ in $T_G$ are attached to neighboring points, we can determine which subset of directions will require that their \gpreconnect/ glues are deactivated and replaced with \gconnect/ glues (so they can bind to the necessary substage supertiles, to be described below). Furthermore, as previously mentioned, there are actually a set of \gpreconnect/ glues in each such location, with one being specific to each pair of substage locations which could bind along that boundary.  In the following, we assume that the correct \gpreconnect/ glue of that set is selected to turn on a similarly specific \gconnect/ glue and discuss the paths of tiles which do this later. Now, in a manner analogous to the way the \gend/ glues correspond to positions in the supertile of the next stage, we note that only a supertile which forms a substage corresponding to a generator position which has the northern \gconnect/ glue can be in a position in the next stage to possibly contain the northern \gconnect/ glue which may be needed to bind the supertile of that stage to another.  Thus, we can determine if any of the \gpreconnect/ glues which were not replaced with \gconnect/ glues need to remain as \gpreconnect/ glues for the next stage, in which they may be or may not be replaced or deactivated, depending on which substage differentiation occurs.  Finally, for the remaining directions for which the \gpreconnect/ glues were neither deactivated and replaced with \gconnect/ glues, nor was it determined that they must remain as active \gpreconnect/ glues, the \gpreconnect/ glues of those sides are turned \off/.
    \item \ginit/ glue: This glue is directly bound to by the initiator tile, and is turned \off/ after attaching to the initiator tile and then receiving a signal from that tile (which it sends after verifying that the first tile along the path which grows from it has connected to the substage supertile).
    \item \gpreinit/ glue: The final substage of stage $i+1$ which will bind to complete the formation of stage $i+1$ will be that corresponding to the location of the root of the tree $T_G$, which will contain the single \gpreinit/ glue of the forming stage which needs to be replaced with a \ginit/ glues.  For each other substage, its \gpreinit/ glue will be turned \off/.
\end{enumerate}

Given the fixed set of glues which need to be activated and deactivated for each differentiating substage, along with the locations and orderings of the sides with which the substage will need to bind to other substages, a unique set of tiles for each substage type (i.e. corresponding generator location) is created and added to $T$ which can specifically attach to that type's initiator tile and form the paths and growths necessary for the differentiation process.

\begin{figure}[htp]
\centering
  \subfloat[][]{%
        \label{fig:generator-example1-filled-paths}%
        \includegraphics[width=2.2in]{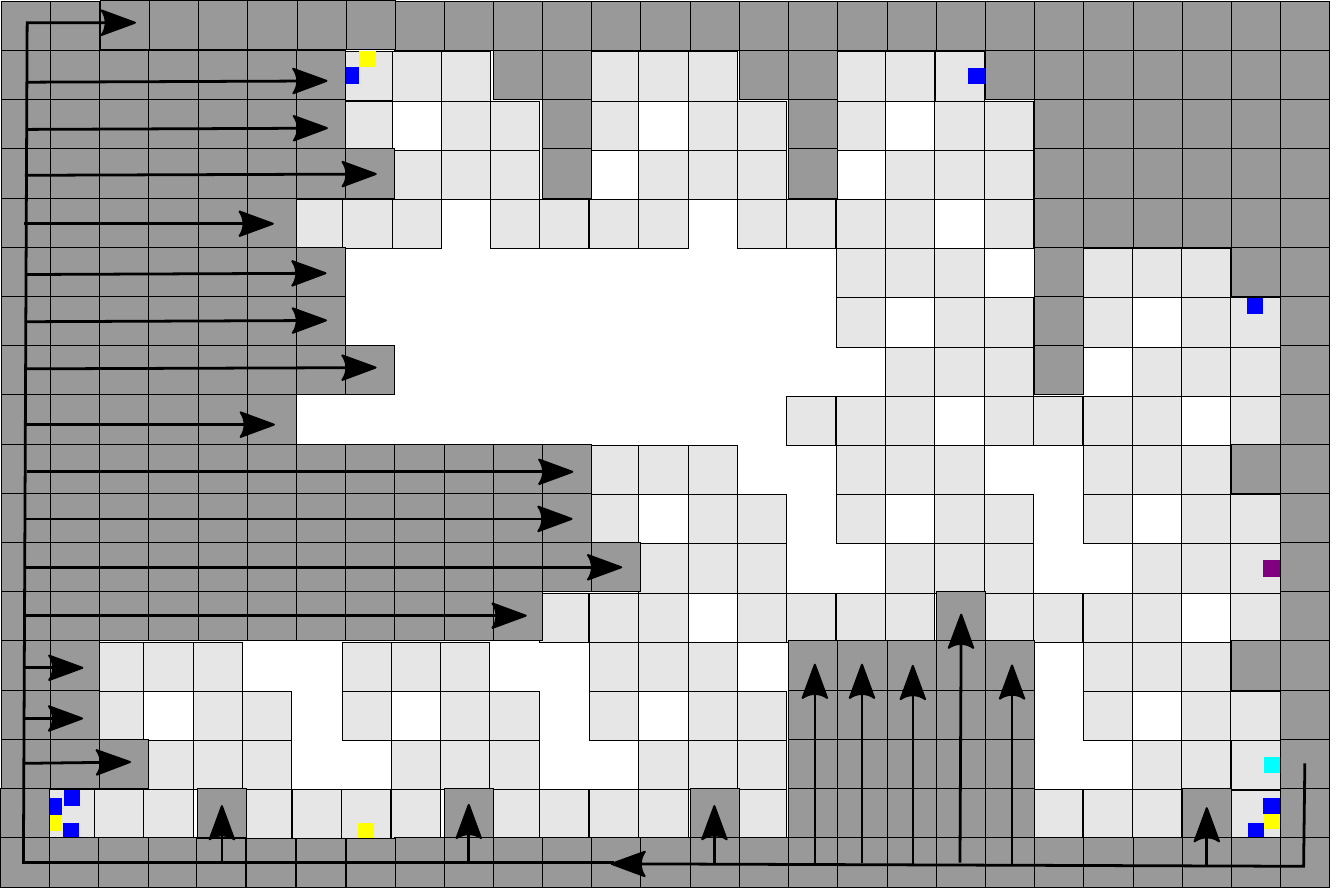}
        }
  \quad\quad
  \subfloat[][]{%
        \label{fig:generator-example1-filled}%
        \includegraphics[width=2.0in]{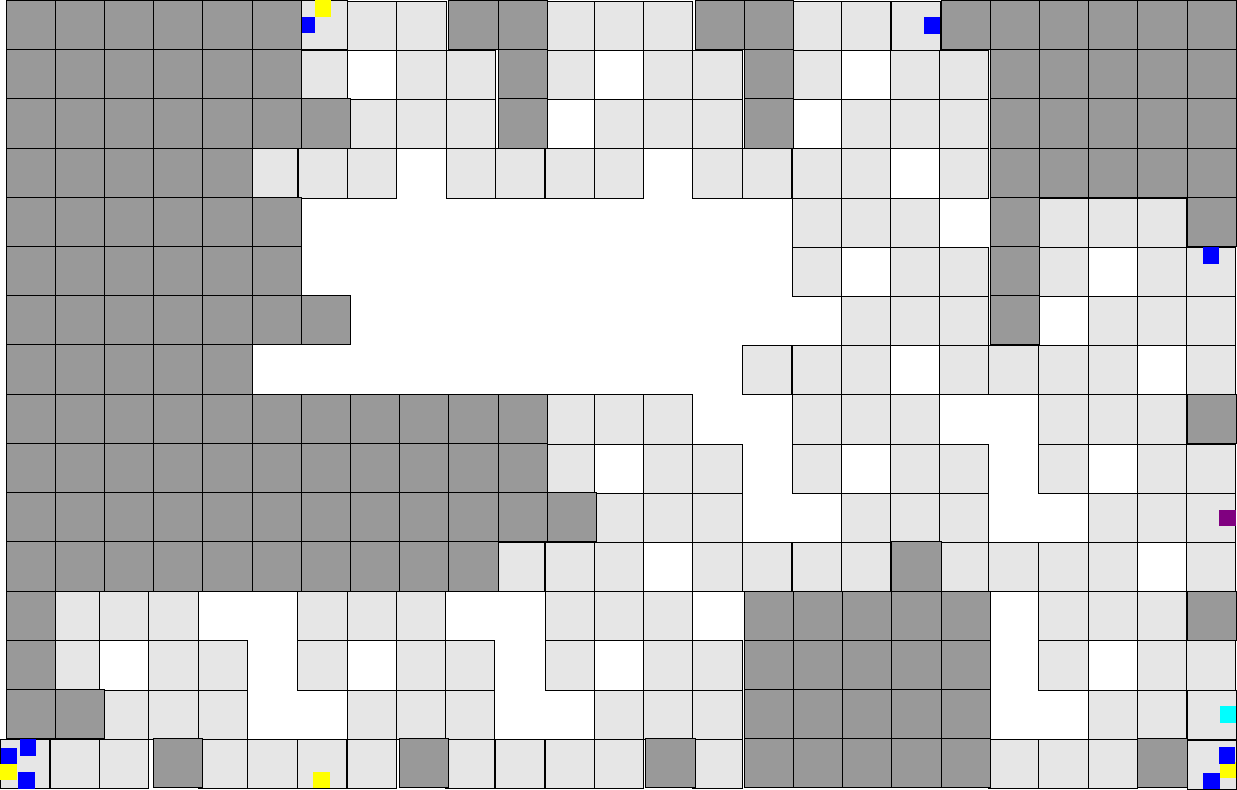}
        }
  \caption{Example of how sides of a substage are filled in by a perimeter path.}
  \label{fig:generator-example1-filled}
\end{figure}

Since the ordering in which sides of a substage will need to attach to others is fixed and occurs in CCW order, a path of tile grows from the location of the \ginit/ glue around the perimeter in a CCW direction until reaching the side which first needs to allow a connection.  The specific way in which such a path is able to grow around the perimeter of a substage is explained in more detail in Section~\ref{sec:perimeter-paths}, but the general scheme is depicted in Figure~\ref{fig:generator-example1-filled-paths}.  It is worth noting that the final difference between the constructions of Theorem~\ref{thm:fractals-temp2} and Theorem~\ref{thm:fractals-temp1} occur in the tiles which form the paths which fill in the perimeters.  Specifically, the glues on the perimeter paths and rows of tiles which grow from them into concavities of the substages have $N$,$E$,$S$, and $W$ glues which are able to attach to those on the sides of the substage to ``detect'' those sides, and these are all strength-1 in both constructions (to match those on the substages).  Only once they attach to matching glues on the substage tiles do they activate the corresponding $N'$,$E'$,$S'$, or $W'$ glues which then bind to strengthen the connection (to be required later). This single difference in the two constructions, in the strength of initial binding of these filler row tiles to the sides of the substages, is necessary to prevent a problem where at temperature 1 those filler rows could allow binding between supertiles representing substages of different stages if those substages have appropriately shaped concavities.  (An example can be seen in Section~\ref{sec:perimeter-paths} for more details.)

Each initiator tile initially has only a single glue in the \on/ state which binds to the \ginit/ glue.  After this binding, it activates a glue which allows the first tile of a new path to bind, and each subsequent tile in the path attaches in a similar fashion.  While growing to the first side which needs to be configured to allow a connection, if and when the path passes the last such side in the ordering, for all sides in between the last and the first, it sends the necessary signals to deactivate and/or activate any of the marker glues (e.g. \gend/ glues) necessary for that side, as determined following the procedure mentioned above. More details about the process of activating and deactivating the marker glues can be found in Section~\ref{sec:perimeter-paths}.  Whenever growth of a side has completed and the tiles which form the corner for the new side have attached, signals begin the detachment of the filler tiles and then the perimeter path of the completed side in the manner described in Section~\ref{sec:perimeter-paths}.

Upon reaching the first side which is to be configured for a connection, the following process occurs:
\begin{enumerate}
    \item The perimeter path and filler tiles grow to completely cover the side and reach the far \gend/ marker.
    \item Then, a signal is passed back through the outermost row which causes every location other than that which will become the \emph{tooth} to detach and also for the adjacent location which will be the \emph{gap} for the complementary supertile's tooth, to detach.  By default, on a western side of a supertile, the tooth is in the very top location and the coordinate of the gap is offset by $(1,-1)$ from that, and an eastern side has its tooth in the second to topmost location and its gap offset from that by $(-1,1)$.  (See Figure~\ref{fig:ew-interface1} for a high-level example.)  For north sides, by default the tooth is in the easternmost location with the gap offset by $(-1,-1)$ and for south sides the tooth is in the position left of the easternmost with the gap offset by $(1,1)$.  A special case can occur when one side of a generator has only one point filled in.  For example, see the west side of the generator in Figure~\ref{fig:generator-example1-connectors-append} where only the southernmost location exists.  In such a side, it must be the case that the tooth and gap geometry occurs opposite that location so that location can contain the \gconnect/ glue used for the attachment.  Therefore, if the generator in Figure~\ref{fig:generator-example1-connectors-append} happened to have its only western point at the top, then the teeth and gaps of supertiles binding along east and west edges would be moved to the southern sides of those edges.
    \item The signal passing through the outermost row will pass the tile adjacent to the tile with the \gpreconnect/ glue.  There, a signal is sent to that tile to activate its \gconnect/ glue.  (See Figure~\ref{fig:ew-interface1} for an example.)
    \item The activation of the \gconnect/ glue indicates that the side is configured for its connection.  Once the outer row of tiles and the tile in the location of the gap have detached, the \gconnect/ glue is exposed and the side is able to connect to the complementary substage supertile.  Note that the \gconnect/ glue which is actually turned on is selected by the specific type of path tile which activates it (which was ultimately determined by the type of initiator tile which began the path) which correctly selects from the set of \gpreconnect/ glues which are specific to each pair of adjacent sides which could be connecting. Given the requirement that the tooth of a side is in place before its \gconnect/ glue can be activated and that the rest of the side must already be completely filled in other than the single gap location, no supertile can combine to a supertile representing a substage different stage, an incorrect substage of the same stage, or a misaligned supertile.  (Please see Section~\ref{sec:correct-bindings} for more details.)
\end{enumerate}

\begin{figure}[htp]
\begin{center}
\includegraphics[width=\linewidth]{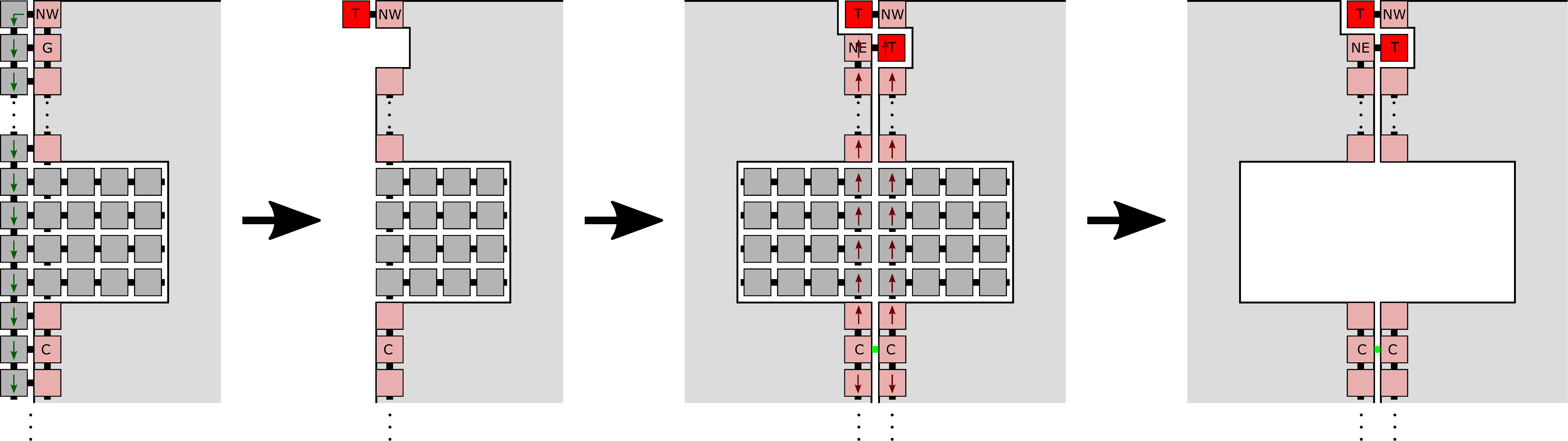}
\caption{An example of the growth process of a tooth and gap on the western side of a supertile preparing to bind to another supertile along that side.}
\label{fig:ew-interface1}
\end{center}
\end{figure}

A complementary supertile binds to this side via the \gconnect/ glue, and this glue causes a signal to begin which initiates the detachment of the filler tiles used on that side.  A signal also propagates to the next side in the CCW direction to continue the differentiation process for each side.  Details of this process can be seen in Section~\ref{sec:perimeter-paths}.  Whether or not a side is configured for, and then participates in, a connection with another substage supertile, the perimeter path and filler tiles are signaled to detach when necessary, and the ordering of those detachments as well as the specific glues which are deactivated are mitigated by ``helper'' tiles similar to the way there were in the construction of the Sierpinski triangle.  These helper tiles generally bind on one side of a tile which must detach in order to cover active glues which could otherwise allow the dissociated and freely floating tile to later incorrectly bind to another assembly.  The active glues of the other side can then be guaranteed to be deactivated because only then can the tile dissociate, resulting in sterile size-2 junk assemblies.  (See Section~\ref{sec:perimeter-paths} for additional details of this process.)

We now note that the signals which are used to configure the consecutive sides of the substage supertiles either (1) pass through tiles which will later detach and become junk assemblies, or (2) pass through tiles which will remain within the fractal-shaped assembly but which are contained along a side which has just bound to another substage.  Since such tiles are now buried within the interior of the newly formed supertile, they will never again be required to pass signals and therefore there is no need for signal reuse (which is not allowed within the STAM).

Finally, we note that by design, the last substage supertile to differentiate will be the one representing the root position of $T_G$ and thus which has the \gpreinit/ glue which will need to be replaced with a \ginit/ to allow differentiation for the next stage.  Additionally, the ordering of side differentiation for that substage ensures that that side is the final to be reconfigured, meaning that only after all necessary connections have been made will the \ginit/ glue be activated. (And since this supertile represents the root node, any which attached to it could only have attached after completing all of their other attachments.) This ensures that all necessary growth has occurred for the newly completed stage, although there may still be perimeter path and filler tiles which need to finish their deactivations and detach.  (Again, this is due solely to the asynchronous nature of the STAM which makes it impossible to know exactly when a tile which has received a signal to deactivate a glue will actually turn that glue \off/.)  Nonetheless, such lingering tiles can cannot cause any incorrect binding behaviors, but can only temporarily delay the correct growth of the next stage.

At this point, the assembly of a new stage is complete, with the guarantee of all necessary, correctly formed and sized substage supertiles being attached, with the exactly correct number and placement of marker glues being active on its perimeter.  It then nondeterministically binds with one of the initiator tiles to begin the differentiation process for the next stage.  By the guarantees of the complete and correct construction of a supertile representing stage $i+1$ assuming the complete and correct assembly of the supertiles representing stage $i$, the correctness of the construction in generating the infinite number of stages in the shape $\calF$ is proven.
\end{proof}

\subsection{Errors at $\tau=1$ with general fractal shapes} \label{sec:errors}

\begin{figure}[htp]
\centering
\includegraphics[width=3.0in]{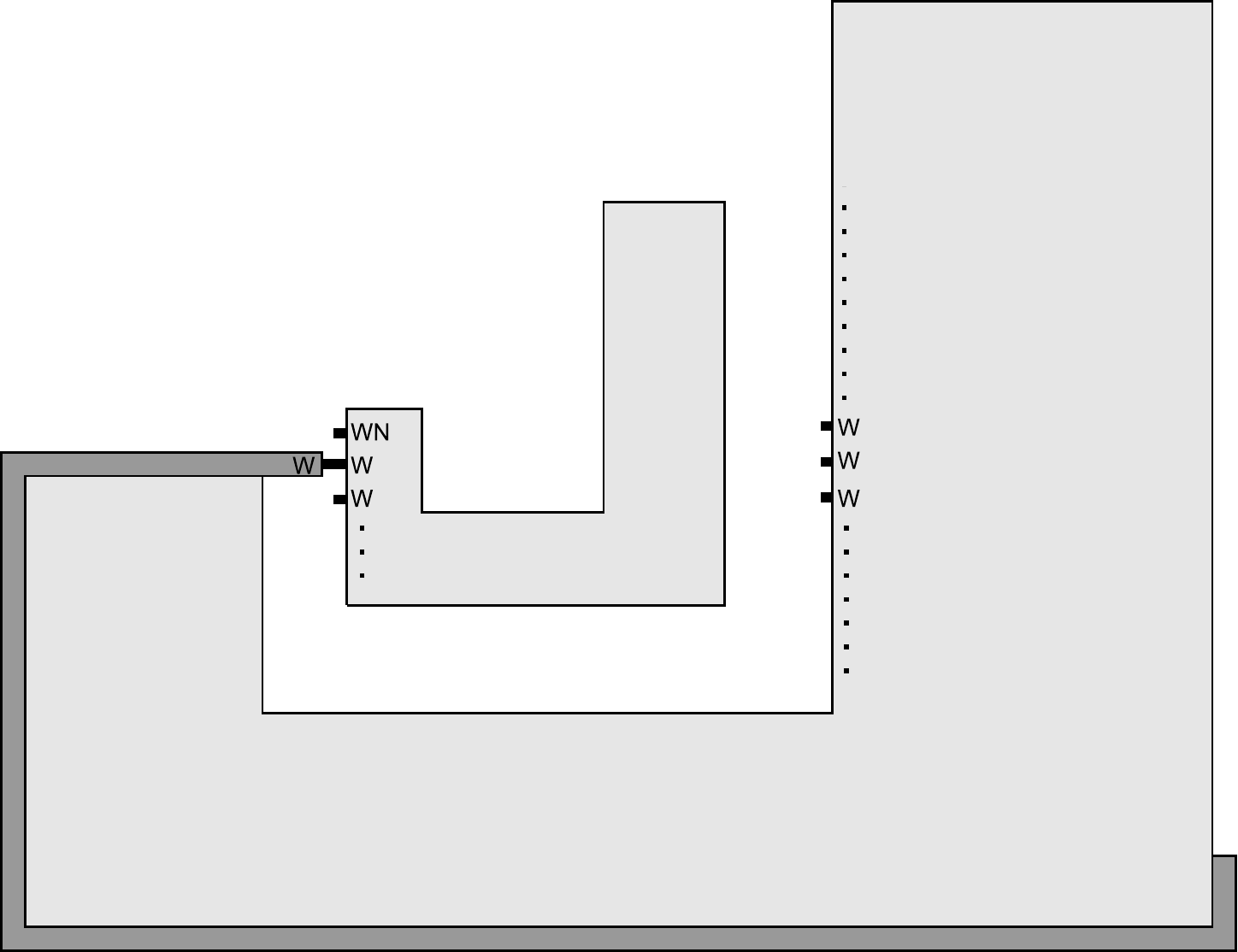}
\caption{A schematic depiction of the error that could occur in a $\tau=1$ system in which the fractal shape contains certain types of concavities.}
\label{fig:temp1-stickiness-problem}
\end{figure}

Figure~\ref{fig:temp1-stickiness-problem} shows why there must be a difference between the temperature-2 construction which works with general discrete self-similar fractal shapes, and the $\tau=1$ construction which only works for a constrained subset of fractal shapes.  Specifically, if a row of filler tiles can grow to a position where there is nothing below it, but its tiles also have active glues which are able to detect their eventual collision with the far wall, then at $\tau=1$ a supertile representing a smaller substage assembly could possibly bind to that row, thus incorrectly connecting supertiles of different stages.  This problem is avoided in the $\tau=2$ construction by making these glues strength-1, thus preventing them from having enough strength to bind the supertiles, along with the fact that only one such row can be growing at a time and thus exposing such glues, and the pattern of the eventual decay of those filler rows is also carefully designed to prevent the ability for erroneous binding.

\def\gend/{{\texttt{end}}}
\subsection{Technical details for the general fractal construction}

Here we provide additional technical details about the construction to prove Theorems~\ref{thm:fractals-temp2} and \ref{thm:fractals-temp1}.

\subsubsection{Perimeter paths and filler tiles}\label{sec:perimeter-paths}

In the general construction, to prepare a supertile representing a stage of the fractal for binding with another such supertile, we must first grow paths of tiles around the supertile in such a way that filler tiles are placed to make the side of a supertile a contiguous set of tile locations that actually contain tiles. The construction of this path is depicted at a high level in Figure~\ref{fig:generator-example1-filled-paths}, where tiles grow a path along the south and west edges of a supertile. Tiles to grow a path along the north and east sides of a supertile are similar. As this path grows around a supertile, it forms a perimeter around the supertile. If a tile of the path is placed next to an empty tile location contained in this perimeter (i.e. a concavity in the supertile), before proceeding with growth of the perimeter path, tiles are placed one after another (using a sequence of signal firings and binding events) so that a single tile wide path of tiles forms. This latter path, growing into the concavity, will eventually place a tile that exposes a glue that binds to a tile of the existing supertile; this binding event will trigger another sequence of signal firings and binding events that cause the perimeter path to continue growth. See Figure~\ref{fig:generator-example1-filled}. Once this perimeter path and filler tiles have been placed, we can detach the outer most perimeter tiles and form ``tooth and gap'' geometries that will ensure that supertile binding takes place between two supertiles in the same stage. We next describe the formation of the tooth and gap geometry.

\subsubsection*{Tooth and gap formation}

Figures~\ref{fig:ew-interface1} and~\ref{fig:ew-interface2} depict the growth process of a tooth and gap on the western side of a supertile as it is being prepared to bind to another supertile along that side. There are two cases to consider depending on whether or not the top of the western side does not have a concavity (that is, the tile WN is the westernmost tile of the supertile). Figure~\ref{fig:ew-interface1} depicts the first case where there is not a concavity at the top of the western side. Figure~\ref{fig:ew-interface2} depicts the second case.

\begin{figure}[htp]
\begin{center}
\includegraphics[width=\linewidth]{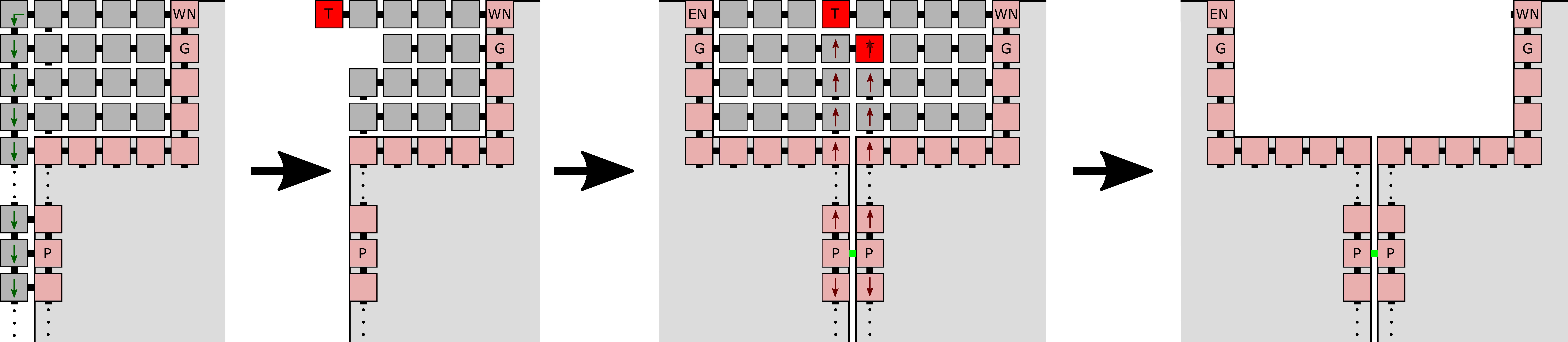}
\caption{An example of the growth process of a tooth and gap on the western side of a supertile preparing to bind to another supertile along that side.  This is similar to that of Figure~\ref{fig:ew-interface1}, but it shows the case where the top of the western side of the supertile has a concavity. Here the tile that detaches to leave the gap labeled $G$ is not part of the fractal assembly though this need not be the case in general. }
\label{fig:ew-interface2}
\end{center}
\end{figure}

In both cases a sequence of signals shown in green in both figures fires via a propagated sequence of glues changing to the \on/ state. This causes the perimeter path (consisting of the tiles with green signals in Figure~\ref{fig:ew-interface2}) to detach. This detachment is described in Section~\ref{sec:detach-perim}. Here we describe the gadgetry for assembling the tooth and gap. The left figures in Figures~\ref{fig:detach-gap-concave} and~\ref{fig:detach-gap-convex} show the tiles and glues making up this gadgetry. Figure~\ref{fig:detach-gap-convex} can be thought of as a portion of the top left corner of the supertile shown in Figure\ref{fig:ew-interface1} and Figure~\ref{fig:detach-gap-concave} as the top left corner of the supertile shown in Figure~\ref{fig:ew-interface2}. Tile $G$ in both figures is the tile that detaches to form the gap (an empty tile location).

\begin{figure}[htp]
\begin{center}
\includegraphics[width=0.9\linewidth]{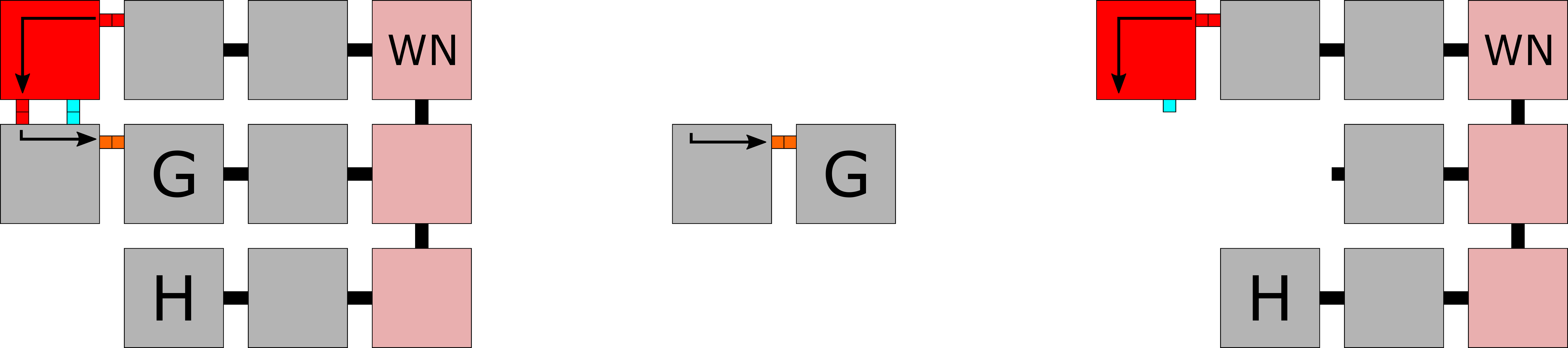}
\caption{Left: A depiction of the the top left region consisting of tiles in the case where there the supertile representing a stage of the fractal assembly has a concavity and the left side of the supertile is preparing to bind along this left side. The red, orange, and aqua glues fire signals that result in the tile labeled $G$ detaching along with the tile to its west. The red tile is the tooth for this side of the supertile. Middle: The duple (a two tile assembly) that detaches to for a gap. Right: The remaining portion of the assembly shown on the left after the gap has formed.}
\label{fig:detach-gap-concave}
\end{center}
\end{figure}

Referring to the left figure of Figure~\ref{fig:detach-gap-convex}, after the perimeter path of tiles detaches, a tile shown in red attaches to the tile labeled $WN$ via the red glue. This fires a signal to allow another tile to attach along the south edge of the red tile via the red glue. This binding event fires an orange glue which binds to an orange glue exposed on the tile $G$. In the case where the top of the supertile side does not have a concavity (shown on the left), $G$ is a tile of the fractal assembly. Suppose that the original input glue allowing this tile to bind to the first stage of the fractal is the blue glue on its south edge. (The north and east edge cases are similar.) When the orange glue of $G$ binds, this binding event fires a signal which turns \on/ a yellow glue along the south edge of $G$, which in turn fires signals to turn \off/ the blue glue of the tile to the south of $G$ (labeled $H$ in the figure) and turn \off/ all glues on all sides of $G$ except the orange glue and the blue glue. Moreover, the binding event of the yellow glue fires an \on/ signal to expose a glue to bind to the pink glue of the tile to the west of $G$. When the pink glue binds, this event triggers the aqua glue to turn \on/ and bind, which finally triggers the aqua glue on the north edge of the tile to the west of $G$ and the red glue on the south edge of the red tile to turn \off/. Following this sequence of binding events and signal firings, the duple $D$ consisting of the tile $G$ and the tile to the west of $G$ will detach with only a red glue and blue glue exposed. See Figure~\ref{fig:detach-gap-convex} for the glue locations. Note that a red glue could be exposed on the south edge of a tooth tile of some other supertile in the system. However also note that the only time such a red glue is exposed on the south edge of a tooth tile is prior to the detachment of the gap tile. Therefore, the duple $D$ would be blocked and unable to bind to a supertile with a red glue exposed on the south edge of a tooth tile. Finally, the yellow glue serves another purpose. The binding event of the yellow glue starts a propogation of signal firings that result in the \gconnect/ glue (shown as the green glue of tile $P$ in Figures~\ref{fig:ew-interface1} and~\ref{fig:ew-interface1}) being turned \on/.

\begin{figure}[htp]
\begin{center}
\includegraphics[width=0.8\linewidth]{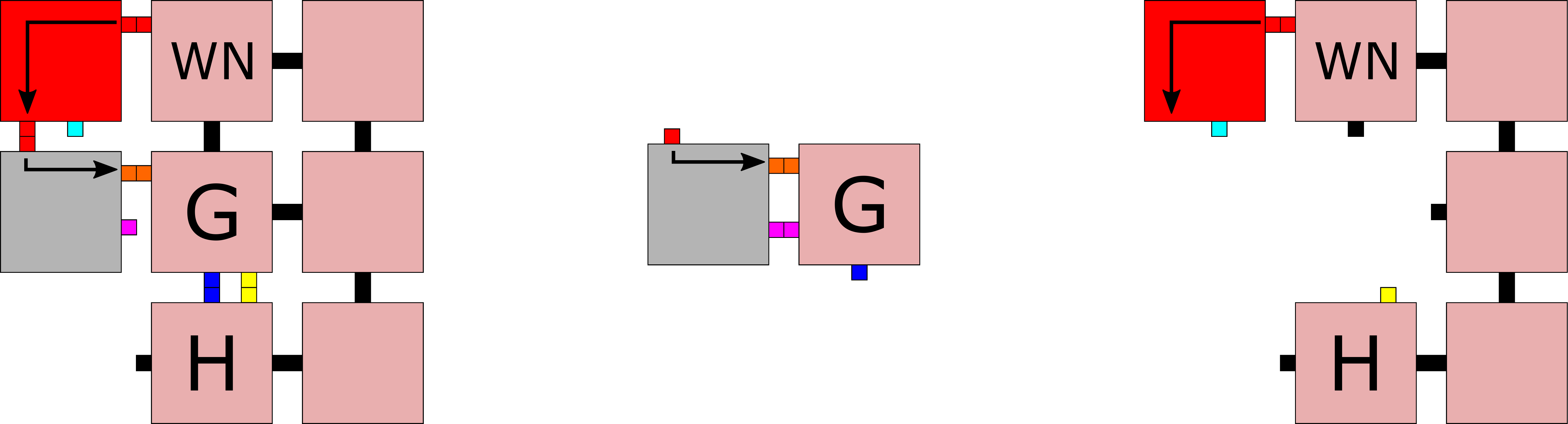}
\caption{Left: A depiction of the the top left region consisting of tiles in the case where there the supertile representing a stage of the fractal assembly does not have a concavity and the left side of the supertile is preparing to bind along this left side. The red, orange, aqua, blue, pink, and yellow glues fire signals that result in the tile labeled $G$ detaching along with the tile to its west. The red tile is the tooth for this side of the supertile. The blue glue is the glue that allowed $G$ to bind to for the fractal assembly of the initial stage of the fractal. Middle: The duple that detaches to for a gap. Right: The remaining portion of the assembly shown on the left after the gap has formed. Note that the yellow glue and an aqua glue are signaled to change to the \off/ state, thought this may not occur until after the gap forms.}
\label{fig:detach-gap-convex}
\end{center}
\end{figure}

$D$ also exposes a blue glue on its south edge. This glue could bind to a tile of the same type as the tile labeled $H$ in Figure~\ref{fig:detach-gap-convex}. Such an $H$ tile must belong to a subassembly of an assembly representing the initial stage of the fractal being assembled. As it stands, this binding could cause erroneous growth. Therefore, we modify the tile types used in the construction of the assembly representing the initial stage of the fractal being assembled. Figures~\ref{fig:gap-is-junk1} and~\ref{fig:gap-is-junk2} describe this modification.

\begin{figure}[htp]
\begin{center}
\includegraphics[width=0.8\linewidth]{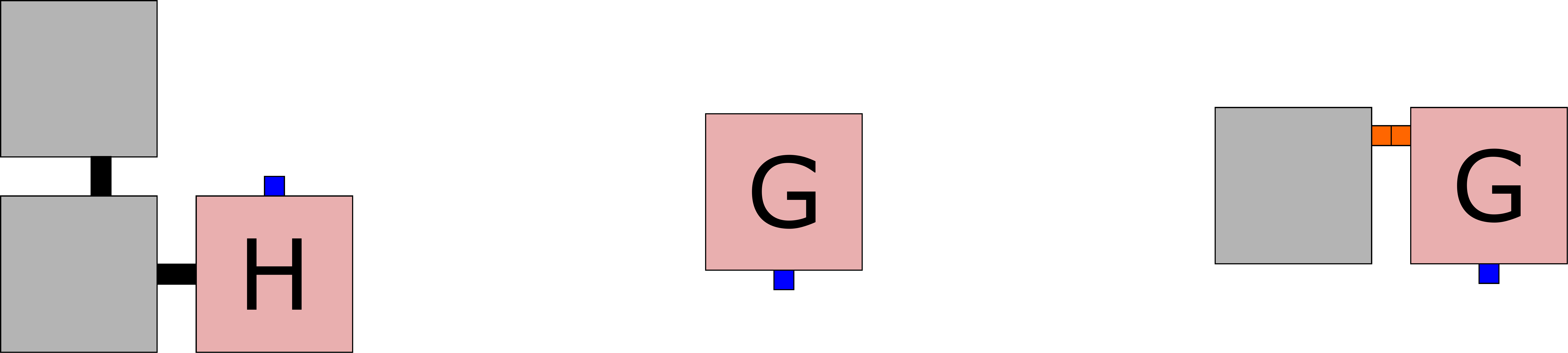}
\caption{This figure shows how $H$ and $G$ bind to additional tiles before binding to a larger assembly to form the initial fractal stage assembly. Here, we are assuming that the $H$ tile binds to the south of the $G$ tile. The case where the $G$ tile binds to the south of the $H$ tile is similar. Left: Using signals, the gray tiles bind first before firing a signal to turn \on/ a glue on the east edge of the southern gray tile. This glue binds to a glue on the west edge of the $H$ tile. In turn, this binding event triggers a blue glue on the north edge of $H$ to turn \on/. Only after this blue glue of the $H$ tile is exposed can the $G$ tile attach via the blue glue on its south edge. Middle: a $G$ tile that can take part in the assembly of an initial fractal stage. Right: A duple containing a $G$ tile with an exposed blue glue. Note that this duple cannot bind the $H$ tile of the left figure, and therefore, cannot take part in the assembly of the initial fractal stage.}
\label{fig:gap-is-junk1}
\end{center}
\end{figure}

\begin{figure}[htp]
\begin{center}
\includegraphics[width=0.9\linewidth]{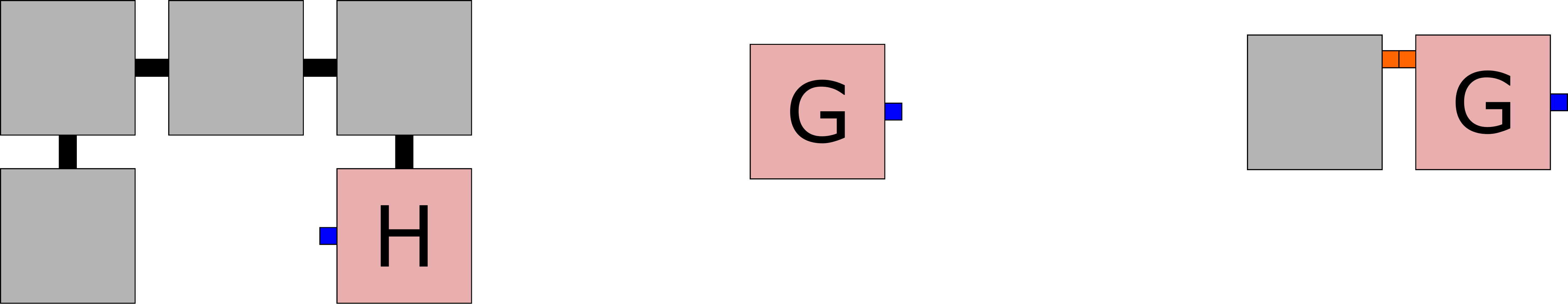}
\caption{This figure shows how $H$ and $G$ bind to additional tiles before binding to a larger assembly to form the initial fractal stage assembly. Here, we are assuming that the $H$ tile binds to the east of the $G$ tile. Left: Using signals, the gray tiles bind first before firing a signal to turn \on/ a glue on the south edge of the easternmost gray tile. This glue binds to a glue on the north edge of the $H$ tile. In turn, this binding event triggers a blue glue on the west edge of $H$ to turn \on/. Only after this blue glue of the $H$ tile is exposed can the $G$ tile attach via the blue glue on its east edge. Middle: a $G$ tile that can take part in the assembly of an initial fractal stage. Right: A duple containing a $G$ tile with an exposed blue glue. Note that this duple cannot bind to the $H$ tile of the left figure, and therefore, cannot take part in the assembly of the initial fractal stage.}
\label{fig:gap-is-junk2}
\end{center}
\end{figure}

The modifications shown in Figures~\ref{fig:gap-is-junk1} and~\ref{fig:gap-is-junk2} take care of the case where the gap forms on the west side of a supertile. The cases where the gap forms on the north, south, or east side of a supertile is similar. With the tooth attached and gap detached, we have ensured that the tooth and gap form and the detachment of the gap does not cause erroneous growth. Note that a tooth and gap must form on the east, south, and north sides of appropriate supertiles. The gadgets for doing so are similar to the the gadgets given here for the west side tooth and gap formation.

Now we describe what happens when the west side of a supertile binds to the east side of a supertile via the \gconnect/ glue. First, the placement of the tooth and gap ensure that the east and west supertiles are substages of the same stage. When the \gconnect/ glue binds, this results in a sequence of binding events dictated by signal firings that cause the detachment of filler tiles. These signals are described later in this section. Here we note that it may be the case that the tile location of the gap should contain a tile in the fractal assembly. For example, the gap of the supertile that binds to the east of a west supertile may have an empty tile location corresponding to a location of the fractal being assembled. In this case, after the east and west supertiles bind and as remaining filler tiles detach, we must take care that the tooth of the western supertile which ``fills the gap'' of the eastern supertile does not detach. A scheme for doing this is shown in Figure~\ref{fig:tooth-gap-after-binding}.

\begin{figure}[htp]
\begin{center}
\includegraphics[width=2in]{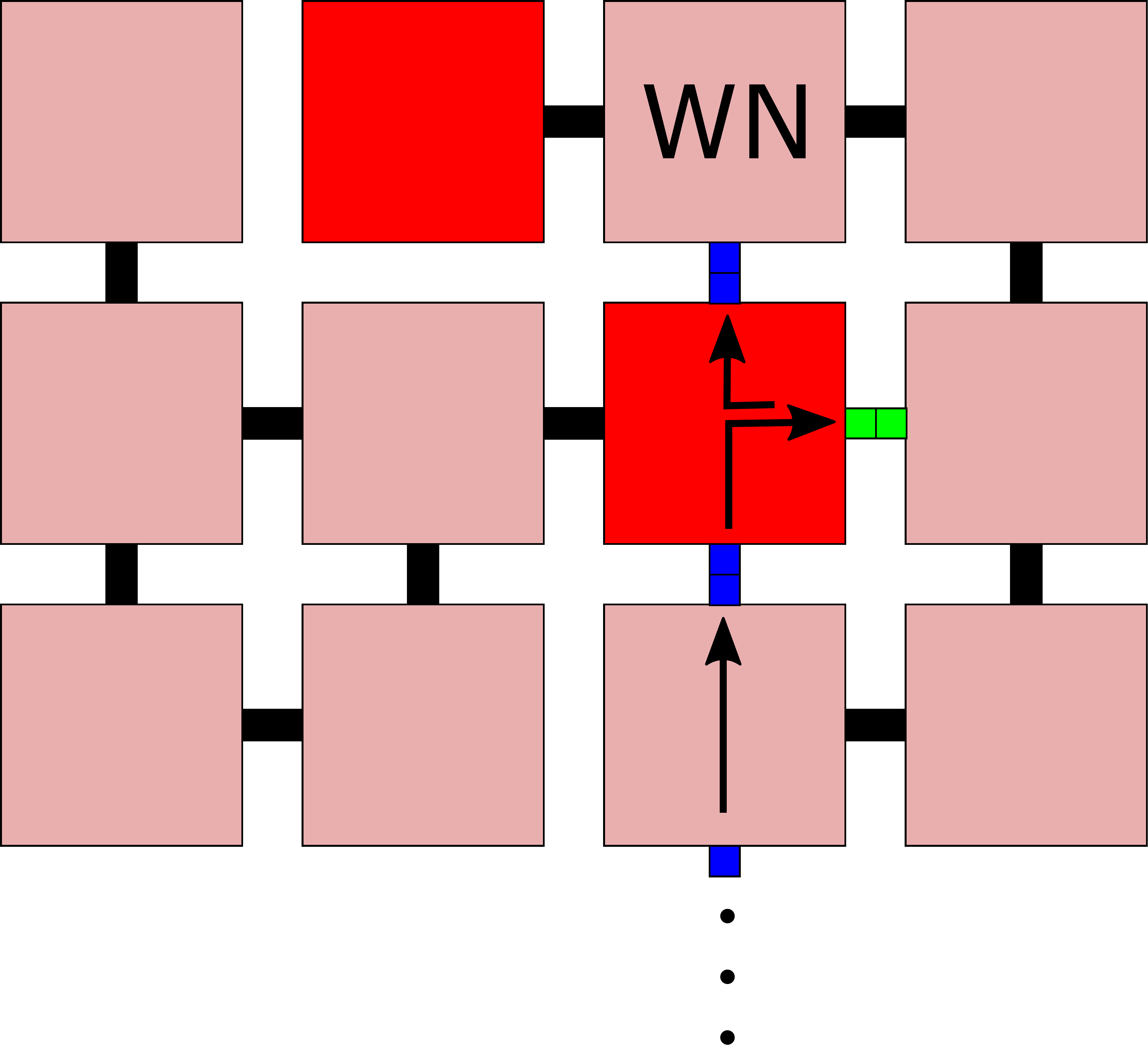}
\caption{The arrows here depict the signals whose firing results in the detachment of filler tiles. These signals are being propagated by glue binding events starting from the binding of the \gconnect/ glue. When a gap is formed by the detachment of a tile at a location that corresponds to a point in the fractal, we can ensure a tooth tile binds in the location of the gap. Such a tile is shown in red. Then, instead of firing signals to detach a tooth tile, a glue is triggered to turn \on/ which bind to some tile located to the north, east, or south of the tooth tile. This glue is shown in green in the figure, where we assume that the tile to the east of the tooth is part of the fractal assembly. The cases where the tile to the north or south of the tooth tile is a tile of the fractal assembly is similar. The binding of this green glue ensure that the tooth tile does not detach and the binding event of this green glue triggers a glue to turn \on/ so that the signal to detach filler tiles can continue to propagate.}
\label{fig:tooth-gap-after-binding}
\end{center}
\end{figure}

\subsubsection*{Detaching the perimeter path and filler tiles.}\label{sec:detach-perim}

We now describe the mechanisms used to detach the perimeter path and filler tiles.  Sets of tile types are designed for two scenarios: 1) the side of the substage assembly the perimeter path is attached to does not need to connect to another substage assembly or 2) the side of the substage assembly the perimeter path is attached to does need to connect to another substage assembly.

Figure~\ref{fig:decay-noncon} shows a schematic which gives a high level overview of how a southern perimeter path dissociates (note that all perimeter paths dissociate the same up to rotation).  During the growth of the perimeter path and filler tiles, glues with the label $DT$ are activated on all filler and path tiles.  The position of the $DT$ glue is such that the $DT$ glue always appears on the input side of the tile (the side of the tile which has a glue that allows the tile to bind to the assembly).  This means that as the path is growing around the perimeter, the $DT$ glue on any tile is not exposed.  Once a tile in the perimeter path binds to the the westernmost southern \gend/ glue, a cascade of signals in the perimeter path is triggered as shown in part (b) of Figure~\ref{fig:decay-noncon} (this cascade of signals is propagated by the $DB$ glue on tile $P$ in Figure~\ref{fig:decay-deats1}).  Eventually, this cascade of signals causes the tile bound to the southernmost western \gend/ glue, to expose a $DT$ glue as shown in part (b) of Figure~\ref{fig:decay-noncon} where the $DT$ glue is depicted as a red glue.  This allows a tile which only exposes a $DT$ glue (and has no other glues), to bind to the exposed $DT$ glue.  Upon binding, the $DT$ glue deactivates all glues on the tile except itself.  Once this tile detaches, it exposes a $DT$ glue on the next tile.  This process is shown in parts (c) and (d) of Figure~\ref{fig:decay-noncon}.  Eventually the perimeter path decays down so that only the tile $t$ on the same row as the westernmost southern \gend/ glue remains along with the tiles above it as shown in part (e) of Figure~\ref{fig:decay-noncon}.  Now, a path of tiles grows from the tile $t$ and binds to the westernmost southern \gend/ glue. This propagates a signal back to the tile $t$ and allows for a tile to bind which begins the growth of the perimeter path on the west side of the substage assembly as shown in part (f) of Figure~\ref{fig:decay-noncon}.  When this occurs, the tiles above $t$ in the same column receive a signal which causes them to deactivate their glues.  Now, the same process repeats on the west side of the substage assembly.

\begin{figure}[htp]
\begin{center}
\includegraphics[width=\linewidth]{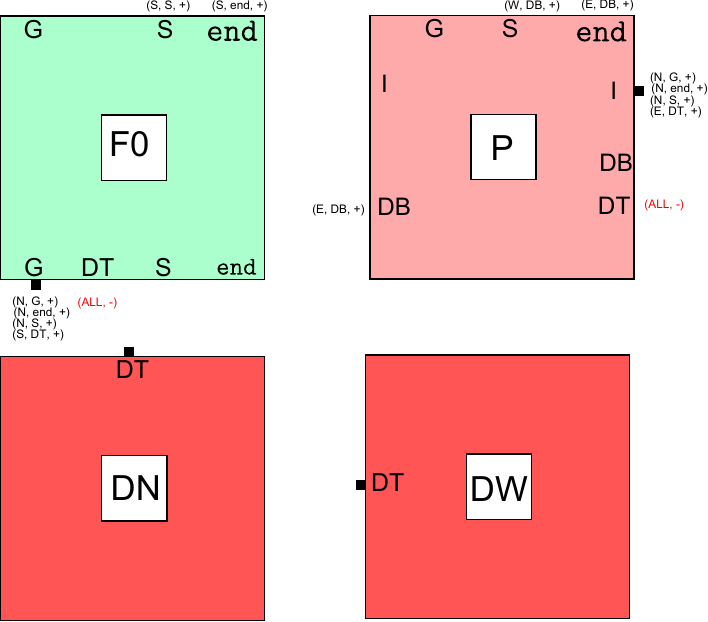}
\caption{ The tiles which implement the detachment of the perimeter path and filler tiles for case 1.  Note that the glue $DT$ deactivates all glues on the tile except itself.
}
\label{fig:decay-deats1}
\end{center}
\end{figure}

\begin{figure}[htp]
\begin{center}
\includegraphics[width=\linewidth]{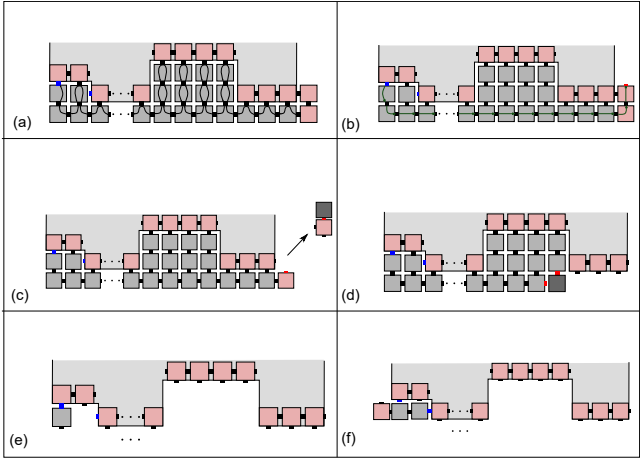}
\caption{ A schematic which shows an overview of the process by which the perimeter path and filler tiles dissociate for the case where the side of the substage assembly on which the perimeter path tiles are growing does not need to form a connection between another substage assembly.  The blue glues represent the \gend/ glues. The red glues represent the $DT$ glues.  The lines and arrows show how signals propagate.
}
\label{fig:decay-noncon}
\end{center}
\end{figure}

Figure~\ref{fig:decay-con} shows a schematic which gives a high level overview of how a southern perimeter path dissociates (note that all perimeter paths dissociate the same up to rotation).  During the growth of the perimeter path and filler tiles, glues with the label $DT$ are activated on all filler and path tiles except the filler tiles which bind to the perimeter path tiles.  These tiles are shown in green in part (a) of Figure~\ref{fig:decay-con}.  In these tiles, a glue $DTW$ is activated on their east side.  The position of the $DT$ glue is such that the $DT$ glue always appears on the input side of the tile (the side of the tile which has a glue that allows the tile to bind to the assembly).  This means that as the path is growing around the perimeter, the $DT$ glue on any tile is not exposed.   In this case, we design the filler tiles so that they bind with strength $\tau$ (the temperature of the system) to the south side of the substage assembly.  Section~\ref{sec:errors} explores the potential downfalls of this and explains why they are not a problem.  If the system being designed has $\tau=2$, we use the mechanism shown in Figure~\ref{fig:t2wall} to bind the filler tiles to the south side of the substage assembly with strength $2$. Similar to the previous case, the binding of the westernmost southern \gend/ glue to a tile in the perimeter path triggers a cascade of signals in the perimeter path.  Eventually, this cascade of signals causes the tile bound to the southernmost western \gend/ glue, to expose a $DT$ glue as shown in part (b) of Figure~\ref{fig:decay-con} where the $DT$ glue is depicted as a red glue.  This allows a tile which only exposes a $DT$ glue (and has no other glues), to bind to the exposed $DT$ glue.  Upon binding, the $DT$ glue deactivates all glues on the tile except itself.  Once this tile detaches, it exposes a $DT$ glue on the next tile.  This process is shown in parts (c) and (d) of Figure~\ref{fig:decay-con}.  Note that this case differs from the previous case in that during this process, the filler tiles do not dissociate.  Recall that this is because the filler tiles which bind to the perimeter path tiles do not have $DT$ glues activated.  Eventually, the only tiles left on the south side of the substage assembly are the filler tiles as shown in part (e) of Figure~\ref{fig:decay-con}.  Once the connection glue is activated, a cascade of signals is activated as shown in part (f) of Figure~\ref{fig:decay-con} and is propagated by the tiles shown in green and the perimeter tiles of the substage assembly via the $DTW$ glue.  Receiving this signal causes the green filler tiles to detach and deactivate all of their glues.  Note that a tile binds to westernmost green filler tile which we call $t$ in part (g) of Figure~\ref{fig:decay-con}.  This prevents $t$ from exposing a glue which if activated, would deactivate all of its glues.  Also, this glue could trigger a $DTW$ glue on a substage assembly which would have unintended consequences.  Consequently, we design this tile so that it never separates from $t$.  Furthermore, the tile which binds to the west of $t$ serves as a starting point for the growth of perimeter path for the west side of the substage assembly.  Now, a path of tiles grows from the tile $t$ and binds to the westernmost southern \gend/ glue . This propagates a signal back to the tile $t$ and allows for the tile to its west to begin the growth of the perimeter path on the west side of the substage assembly as shown in part (h) of Figure~\ref{fig:decay-noncon}.  When this occurs, the tiles above $t$ in the same column receive a signal which causes them to deactivate their glues.  Now, the same process repeats on the west side of the substage assembly.

\begin{figure}[htp]
\begin{center}
\includegraphics[width=1.5in]{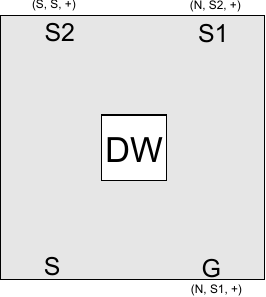}
\caption{The mechanism used to attach to the wall when designing temperature-2 systems.  The tile requires both $S1$ and $S2$ glues to be bound before sending the signal to activate $S$.  The $S1$ and $S2$ glues are strength 1 glues while the $G$ glue is to be strength 2.
}
\label{fig:t2wall}
\end{center}
\end{figure}

\begin{figure}[htp]
\begin{center}
\includegraphics[width=\linewidth]{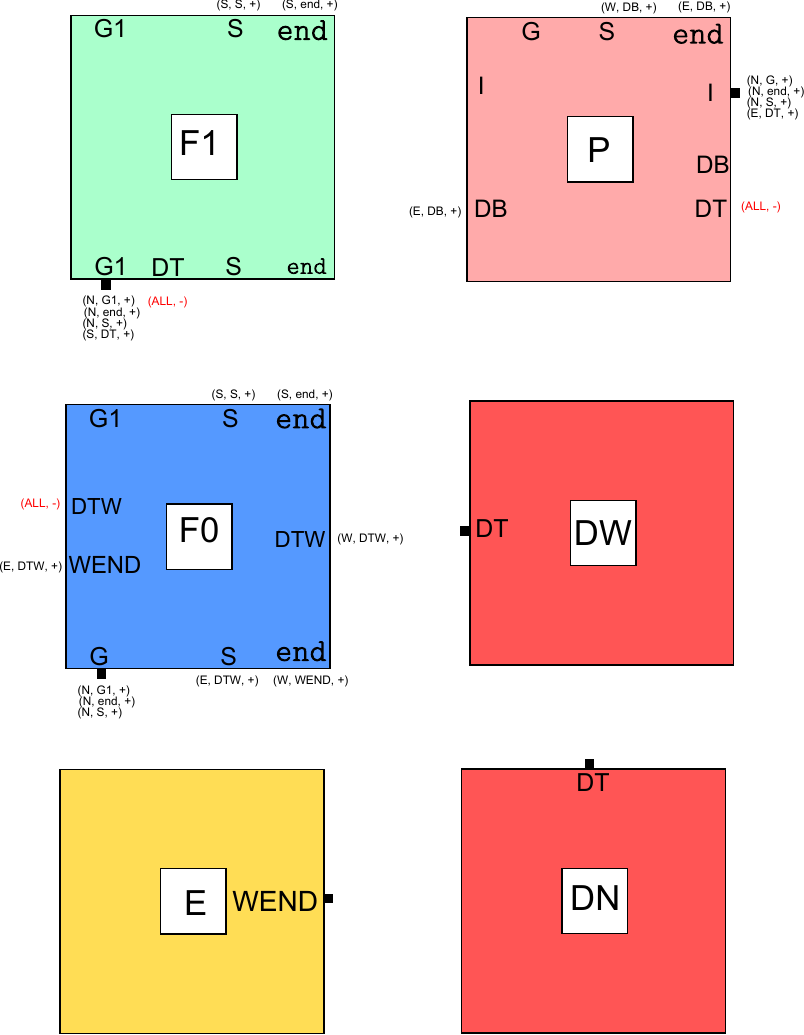}
\caption{ The tiles which implement the detachment of the perimeter path and filler tiles for case 2.  Note that the glue $DTW$ deactivates all glues on the tile.
}
\label{fig:decay-deats2}
\end{center}
\end{figure}

Note that while the description of mechanisms here were focused on dissociating the perimeter path and filler tiles on the south side of the substage assembly, these mechanisms can be easily adapted for any side of the substage assembly by rotating the tiles discussed here.  Also note that while both of the scenarios we consider above examine the hard case where the south west corner of the substage assembly is a concave corner, the tile set can easily be adapted to handle the case where the corner is convex.

Every tile that we added will eventually have either a $DT$ glue activated or a $DTW$ glue activated.  Since both of these glues deactivate all other glues on the tile (with the potential exception of themselves), eventually, all of the tiles will become inert or the tile will be bound to the tile which has a single $DT$ glue and will not expose any glues.  Thus, we need only ensure that the binding of glues to the glues in the queue to be deactivated is benign.  First, note that it follows from the construction of the tiles above that no tiles can bind to ``non-static'' glues (that is, glues which trigger an event) on any substage assemblies.  Also, note that by the design of the tiles, the tiles could not unintentionally ``link'' two stages by exposing glues so that two separate stages could bind.  Consequently, the only glues the glues in the queue to be deactivated can bind to are glues which will eventually detach and do not cause any binding events to occur within the glues of on a substage assembly.  It follows that the tiles which have dissociated cannot interfere with the proper growth of substage assemblies.
\begin{figure}[htp]
\begin{center}
\includegraphics[width=\linewidth]{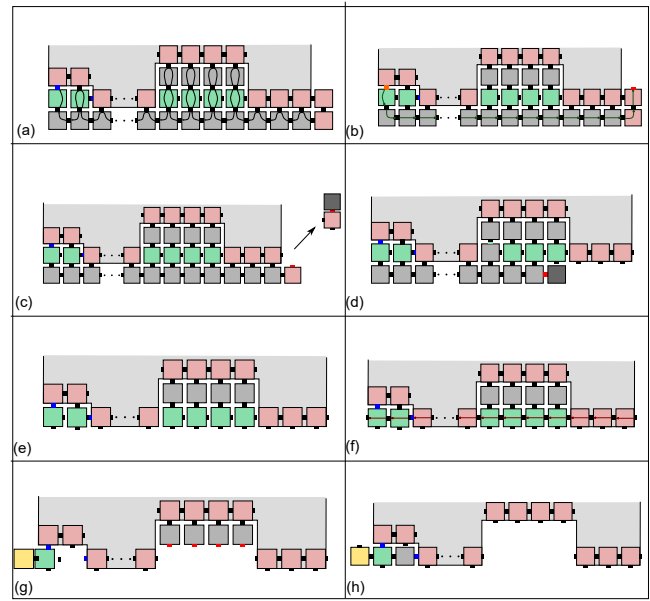}
\caption{ A schematic which shows an overview of the process by which the perimeter path and filler tiles dissociate for the case where the side of the substage assembly on which the perimeter path tiles are growing does need to form a connection between another substage assembly.  The blue glues represent the \gend/ glues.  The red glues represent the $DT$ glues.  The lines and arrows show how signals propagate.
}
\label{fig:decay-con}
\end{center}
\end{figure}

\subsection{Correctness of supertile bindings}\label{sec:correct-bindings}
We now discuss the geometric and timing constraints of substage connections.  The key insight is that the connection glue is exposed only once a tooth and gap are present.  To see this, an overview of the substage connection process is shown in Figure~\ref{fig:ew-interface2}.  As shown in the Figure, the tooth tile sends a signal which causes the connection glue (shown as green in Figure~\ref{fig:ew-interface2}) to be exposed.  This means that the connection glue cannot be exposed until the tooth tile is present.  Also, the attachment of the tooth tile eventually causes a tile to fall off to make a gap on the side as shown in Figure~\ref{fig:ew-interface2}.  Once this occurs, another substage can attach via the connection glue.  Then, and only then, is a signal sent which causes the filler tiles to dissociate.

Another important insight is that the filler tiles are present when the tooth tile attaches and remain present until after another substage has attached.  The presence of filler tiles makes the side of the substage completely ``flat'' with the exception of a gap and tooth.  This means, that only the same substages can attach to one another since exactly those substages have the tooth and gap in the same location.  All other substages have their gaps and tooth tiles in different places relative to the connection glue which means they are geometrically prevented from binding.

\section{Conclusion}
In this paper, we have shown how signal-passing tiles can be designed to self-assemble discrete self-similar fractals in a very hierarchical manner. That is, separately and in parallel, the copies of each stage of a given fractal self-assemble, and then those copies combine to form the next stage.  During this process, the sides of each copy which must combine with each other are prepared for that combination in such a way that only copies of the same stage can combine, and only once all pieces of the previous stage have attached.  This is done by using geometric hindrance created by small bumps and dents at carefully spaced locations, along with well-controlled timing of glue activations.  Once copies of a stage have bound together, any ``filler'' tiles which attached to create the necessary geometry along the combining edge, but which are not part of the final fractal shape, then detach.  These detached tiles, called ``junk'', are designed so that they always eventually become subassemblies which are no greater than size 2, and they cannot interfere with any additional assembly.

While these results allow for the self-assembly of fractals without any increase in scale factor and, for an infinite set of fractals, at temperature 1, there remain a number of improvements which may be possible and open questions which remain.  First, is the cooperativity of temperature 2 strictly necessary for fractals which are not singly concave?  I.e., is there a temperature 1 construction which can correctly self-assemble such fractals?  It appears that such a construction would require non-trivial adaptations to or departures from our current construction. Conversely, proving it impossible also appears to be difficult.  Second, can the maximum size of the terminal junk assemblies be reduced to the optimal value of 1?  This also appears difficult due to the fact that within the STAM glue deactivations (like activations) happen asynchronously, and therefore it is not possible to guarantee that both sides of a bound glue (i.e. the glues on the two adjacent edges of bound tiles) have been deactivated before a tile or subassembly detaches.  To counter problems that may arise from this, we frequently ensure that special ``blocker'' tiles first attach to the soon-to-be-junk tiles to hide glues which may not be \off/ and which may allow the junk assembly to interfere with other supertiles.  Third, is it possible to further reduce the number of junk assemblies which are produced, especially during the self-assembly of fractals such as the Sierpinski triangle, which in the current construction requires a number of tiles and junk assemblies on the order of approximately $1/4$ of the number of tiles which remain to form each given stage?  Fourth, our construction for self-assembling the Sierpinski triangle in the proof of Theorem~\ref{thm:triangle} uses 48 tile types (compared with 19 for the scale factor 2 version of \cite{jSignals}).  Can this be reduced to a similar or smaller number?

\vspace{-10pt}
\bibliographystyle{abbrv} %
\bibliography{tam,experimental_refs}

\begin{thebibliography}{10}

\bibitem{TreeFractals}
K.~Barth, D.~Furcy, S.~M. Summers, and P.~Totzke.
\newblock Scaled tree fractals do not strictly self-assemble.
\newblock In {\em Unconventional Computation \& Natural Computation (UCNC)
  2014, University of Western Ontario, London, Ontario, Canada {\rm July 14-18,
  2014}}, pages 27--39, 2014.

\bibitem{Versus}
S.~Cannon, E.~D. Demaine, M.~L. Demaine, S.~Eisenstat, M.~J. Patitz, R.~T.
  Schweller, S.~M. Summers, and A.~Winslow.
\newblock Two hands are better than one (up to constant factors): Self-assembly
  in the 2ham vs. atam.
\newblock In N.~Portier and T.~Wilke, editors, {\em STACS}, volume~20 of {\em
  LIPIcs}, pages 172--184. Schloss Dagstuhl - Leibniz-Zentrum fuer Informatik,
  2013.

\bibitem{MHAM}
C.~T. Chalk, D.~A. Fernandez, A.~Huerta, M.~A. Maldonado, R.~T. Schweller, and
  L.~Sweet.
\newblock Strict self-assembly of fractals using multiple hands.
\newblock {\em Algorithmica}, pages 1--30, 2015.

\bibitem{AGKS05g}
Q.~Cheng, G.~Aggarwal, M.~H. Goldwasser, M.-Y. Kao, R.~T. Schweller, and P.~M.
  de~Espan\'{e}s.
\newblock Complexities for generalized models of self-assembly.
\newblock {\em SIAM Journal on Computing}, 34:1493--1515, 2005.

\bibitem{DDFIRSS07}
E.~D. Demaine, M.~L. Demaine, S.~P. Fekete, M.~Ishaque, E.~Rafalin, R.~T.
  Schweller, and D.~L. Souvaine.
\newblock Staged self-assembly: nanomanufacture of arbitrary shapes with
  ${O}(1)$ glues.
\newblock {\em Natural Computing}, 7(3):347--370, 2008.

\bibitem{Signals3DArxiv}
T.~Fochtman, J.~Hendricks, J.~E. Padilla, M.~J. Patitz, and T.~A. Rogers.
\newblock Signal transmission across tile assemblies: 3{D} static tiles
  simulate active self-assembly by 2{D} signal-passing tiles.
\newblock Technical Report 1306.5005, Computing Research Repository, 2013.

\bibitem{STAM-fractals}
J.~Hendricks, M.~Olsen, M.~J. Patitz, T.~A. Rogers, and H.~Thomas.
\newblock Hierarchical self-assembly of fractals with signal-passing tiles
  (extended abstract).
\newblock In {\em Proceedings of the 22nd International Conference on DNA
  Computing and Molecular Programming (DNA 22),
  Ludwig-Maximilians-Universität, Munich, Germany {\rm September 4-8, 2016}},
  pages 82--97.

\bibitem{STAM-fractals-arxiv}
J.~Hendricks, M.~Olsen, M.~J. Patitz, T.~A. Rogers, and H.~Thomas.
\newblock Hierarchical self-assembly of fractals with signal-passing tiles.
\newblock Technical Report 1606.01856, Computing Research Repository, 2016.

\bibitem{STAMshapes}
J.~Hendricks, M.~J. Patitz, and T.~A. Rogers.
\newblock Replication of arbitrary hole-free shapes via self-assembly with
  signal-passing tiles (extended abstract).
\newblock In C.~S. Calude and M.~J. Dinneen, editors, {\em Unconventional
  Computation and Natural Computation - 14th International Conference, {UCNC}
  2015, Auckland, New Zealand, August 30 - September 3, 2015, Proceedings},
  volume 9252 of {\em Lecture Notes in Computer Science}, pages 202--214.
  Springer, 2015.

\bibitem{JonoskaSignals1}
N.~Jonoska and D.~Karpenko.
\newblock Active tile self-assembly, part 1: Universality at temperature 1.
\newblock {\em International Journal of Foundations of Computer Science},
  25(02):141--163, 2014.

\bibitem{JonoskaSignals2}
N.~Jonoska and D.~Karpenko.
\newblock Active tile self-assembly, part 2: Self-similar structures and
  structural recursion.
\newblock {\em International Journal of Foundations of Computer Science},
  25(02):165--194, 2014.

\bibitem{STAMPatternRep}
A.~Keenan, R.~Schweller, and X.~Zhong.
\newblock Exponential replication of patterns in the signal tile assembly
  model.
\newblock {\em Natural Computing}, 14(2):265--278.

\bibitem{jSSADST}
J.~I. Lathrop, J.~H. Lutz, and S.~M. Summers.
\newblock Strict self-assembly of discrete {S}ierpinski triangles.
\newblock {\em Theoretical Computer Science}, 410:384--405, 2009.

\bibitem{LutzShutters12}
J.~H. Lutz and B.~Shutters.
\newblock Approximate self-assembly of the sierpinski triangle.
\newblock {\em Theory Comput. Syst.}, 51(3):372--400, 2012.

\bibitem{jSignals}
J.~E. Padilla, M.~J. Patitz, R.~T. Schweller, N.~C. Seeman, S.~M. Summers, and
  X.~Zhong.
\newblock Asynchronous signal passing for tile self-assembly: Fuel efficient
  computation and efficient assembly of shapes.
\newblock {\em International Journal of Foundations of Computer Science},
  25(4):459--488, 2014.

\bibitem{SignalTilesExperimental}
J.~E. Padilla, R.~Sha, M.~Kristiansen, J.~Chen, N.~Jonoska, and N.~C. Seeman.
\newblock A signal-passing dna-strand-exchange mechanism for active
  self-assembly of dna nanostructures.
\newblock {\em Angewandte Chemie International Edition}, 2015.

\bibitem{jSADSSF}
M.~J. Patitz and S.~M. Summers.
\newblock Self-assembly of discrete self-similar fractals.
\newblock {\em Natural Computing}, 1:135--172, 2010.

\bibitem{Winf98}
E.~Winfree.
\newblock {\em Algorithmic Self-Assembly of {D}{N}{A}}.
\newblock PhD thesis, California Institute of Technology, June 1998.

\end{thebibliography}

\end{document}